\documentclass[12pt]{article}

\usepackage[margin=1in]{geometry}
\usepackage{amsmath,amsthm,amssymb}
\usepackage{latexsym,xspace}
\usepackage{caption,subfigure}
\usepackage{tikz}
\usetikzlibrary{positioning}
\usetikzlibrary{calc}
\usepackage{float}
\usepackage{url}
\usepackage{xspace}
\usepackage{nicefrac}
\usepackage{kbordermatrix}
\usepackage[inline]{enumitem}
\usepackage{xr}
\usepackage[linesnumbered]{algorithm2e}
\usepackage[normalem]{ulem}

\setlist[itemize,1]{leftmargin=2cm,labelsep=1cm,itemsep=20pt,topsep=10pt}
\setlist[enumerate,1]{topsep=5pt,itemsep=0pt,label=(\alph*)}

\floatstyle{plain}
\restylefloat{figure}

\newtheorem{theorem}{Theorem}
\newtheorem{lemma}[theorem]{Lemma}
\newtheorem{corollary}[theorem]{Corollary}

\theoremstyle{remark}

\newcommand{\thsp}{\hspace*{0.5pt}}

\newcommand{\Path}[2]{$(#1,\ldots, #2)$ path}

\tikzset{b/.style = {circle,draw,inner sep=0pt,fill=black ,minimum size=3.0pt},
         w/.style = {circle,draw,inner sep=0pt,fill=white,thick,minimum size=3.2pt},
         i/.style = {circle,     inner sep=0pt,thick,minimum size=3.0pt},
       stc/.style = {circle,draw,inner sep=0pt,thick,minimum size=7.0mm},
       tac/.style = {rectangle,draw,inner sep=0pt,thick,rotate=45,minimum size=7mm}
     }
\tikzset{every picture/.style={line width=0.75pt}}
\tikzset{empty/.style={rectangle,draw=none,fill=none}}

\newcommand{\es}{\varnothing}
\newcommand{\sm}{\setminus}
\newcommand{\pre}{pre}
\renewcommand{\phi}{\varphi}
\newcommand{\ol}{\overline}
\newcommand{\qua}{\mbox{quasi-}}
\newcommand{\linfor}{\textsc{LinearForests}\xspace}

\newcommand{\chbip}{\textsc{ChordalBipartite}\xspace}
\newcommand{\bip}{\textsc{Bipartite}\xspace}
\newcommand{\och}{\textsc{OddChordal}\xspace}
\newcommand{\mon}{\textsc{Monotone}\xspace}
\newcommand{\qmon}{\textsc{Quasimonotone}\xspace}
\newcommand{\hf}{\textsc{HoleFree}\xspace}
\newcommand{\ehf}{\textsc{EvenHoleFree}\xspace}
\newcommand{\flaw}{\textsc{Flaw}\xspace}
\newcommand{\noflaw}{\textsc{Flawless}\xspace}
\newcommand{\LR}{\ensuremath{L{:}R}} 
\newcommand{\NP}{\ensuremath{\mathbb{NP}}}
\newcommand{\PP}{\ensuremath{\mathbb{P}}}
\newcommand{\dist}{\operatorname{dist}}
\newcommand{\diam}{\operatorname{diam}}

\newcommand{\cC}{\mathcal{C}}

\newcommand{\cF}{\mathcal{F}}
\newcommand{\cG}{\mathcal{G}}

\newcommand{\cN}{\mathcal{N}}

\newcommand{\cT}{\mathcal{T}}

\title{Quasimonotone graphs}
\author{Martin Dyer\thanks{School of Computing, University of Leeds, Leeds LS2~9JT, UK. Email: \texttt{M.E.Dyer@leeds.ac.uk}. }
\and Haiko M\"{u}ller\thanks{School of Computing, University of Leeds, Leeds LS2~9JT, UK. Email: \texttt{H.Muller@leeds.ac.uk}. }}
\date{February 13, 2018}

\begin{document}

\maketitle

\begin{abstract}
  For any class $\cC$ of bipartite graphs, we define quasi-$\cC$ to be the class of
  all graphs $G$ such that every bipartition of $G$ belongs to $\cC$.
  This definition is motivated by a generalisation of the switch Markov chain
  on perfect matchings from bipartite graphs to nonbipartite graphs.
  The monotone graphs, also known as bipartite permutation graphs
  and proper interval bigraphs, are such a  class of bipartite graphs.
  We investigate the structure of quasi-monotone graphs and hence
  construct a polynomial time recognition algorithm for graphs in this class.
\end{abstract}

\section{Introduction}
In~\cite{DyJeMu17} (with Jerrum) and \cite{DyeMu17a} we considered the \emph{switch} Markov chain on perfect matchings in bipartite and nonbipartite graphs. This chain repeatedly replaces two matching edges with two non-matching edges involving the same four vertices, if possible. (See~\cite{DyJeMu17,DyeMu17a} for details.) We considered the ergodicity and mixing properties of the chain.

In particular, we proved in~\cite{DyJeMu17} that the chain is \emph{rapidly mixing} (i.e.~converges in polynomial time) on the class of \emph{monotone} graphs. This class of bipartite graphs was defined by Diaconis, Graham and Holmes in~\cite{DiGrHo01}, motivated by statistical applications of perfect matchings. The biadjacency matrices of graphs in the class have a ``staircase'' structure. Diaconis et al. conjectured the rapid mixing property shown in~\cite{DyJeMu17}. We also showed in~\cite{DyJeMu17} that this class is, in fact, identical to the known class of \emph{bipartite permutation} graphs~\cite{SpBrSt87}, which is itself known to be identical to the class of \emph{proper interval bigraphs}~\cite{HelHua04}.

In extending the work of~\cite{DyJeMu17} to nonbipartite graphs in~\cite{DyeMu17a}, we showed that the rapid mixing proof for monotone graphs extends easily to a class of graphs which includes, beside the monotone graphs themselves,  all \emph{proper}, or \emph{unit}, \emph{interval graphs}~\cite{BogWes99}. In this class the bipartite graph given by the cut between any bipartition of the vertices of the graph must be a monotone graph. We called these graphs \emph{quasimonotone}.

In fact, ``quasi-'' is an operator on bipartite graph classes, and can be applied more generally. In this view, quasimonotone graphs are quasi-monotone graphs, as formally defined in section~\ref{sec:quasi-pre}, and discussed in section~\ref{ss:examples}, below.

For any class of bipartite graphs that is recognisable in polynomial time, the definition of its quasi-class implies membership in co-\NP, and deterministically only an exponential time recognition algorithm. Thus an immediate question is whether we can recognise the quasi-class in polynomial time. The main contribution of this paper is a proof that quasimonotone graphs have a polynomial time recognition algorithm.

\subsection{Definitions and notation} \label{ss:defs}
If $G$ is a graph, we will denote its vertex set by $V[G]$, and its edge set by $E[G]$. If $U\subseteq V[G]$, then $G[U]$ will denote the subgraph induced by $U$. To ease the notation we do not distinguish between $U$ and the subgraph it induces in $G$ where this does not cause ambiguity. So a cycle in $G$ is either a subgraph or the set of its vertices. Similarly, we will write $G=H$ when $G$ is isomorphic to $H$.

A subgraph of $G=(V,E)$ is a \emph{cycle} in $G$ if it is a connected and $2$-regular.  The \emph{length} or \emph{size} of a cycle is the number of its edges (or vertices). A \emph{chord} of a cycle $(U,F)$ in $(V,E)$ is an edge in $U^{(2)} \cap E \sm F$. A chord in a cycle of even length is \emph{odd} if the distance between its endpoints on the cycle is odd. That is, an \emph{odd chord} splits an even cycle into two smaller cycles of even length. An \emph{even chord} splits an even cycle into two smaller cycles of odd length.

A \emph{hole} in a graph is a chordless cycle of length at least five. A cycle of length three is a \emph{triangle}, and a cycle of length four a \emph{quadrangle}. A hole is \emph{odd} if it has an odd number of vertices, otherwise \emph{even}. Let $\hf$ be the class of graphs without a hole, and $\ehf$ the class of graphs without even holes. For the purposes of this paper, a \emph{long} hole will be defined as an odd hole of size at least 7.

A \emph{bipartition} $L,R$ of a set $V$ is such that $L \subseteq V$ and $R = V \sm L$. Then, if $G=(V,E)$ is any graph, the graph $G[\LR]$ is the bipartite graph with vertex bipartition $L,R$, and edge set the cut $\LR= \{xy\in E : x\in L, y\in R\}$. We refer to $G[\LR$] as a bipartition of $G$.

The \emph{distance} $\dist(u,v)$ between two vertices $u$ and $v$ is the length of a shortest \Path{u}{v} in $G$. If $H$ is an induced subgraph of $G$, and $x,y\in H$, we denote the distance from $x$ to $y$ in $H$ by $\dist_H(x,y)$. If $v\in V$, $\dist(v,H)$ is the minimum distance $\dist(v,w)$ from $v$ to any vertex $w\in H$. The maximum distance between two vertices in $G$ is the \emph{diameter} of $G$.

If $G=(V,E)$ and $v\in V$, we denote the neighbourhood of $v$ by $N(v)$, and $N(v)\cup\{v\}$ by $N[v]$.

\subsection{Structure of the paper}\label{ss:stucture}
Though the focus of the paper is on the class of quasimonotone graphs, in
section~\ref{sec:quasi-pre} we discuss a generalisation of the construction of the class
and some immediate properties. In section~\ref{ss:examples} we give some example.

Sections \ref{sec:qmg} to \ref{sec:nolongholes} show that quasimonotone graphs can be recognised
in polynomial time. We begin, in section \ref{ss:fp}, by proving some properties of quasimonotone graphs
for later use, using the characterisation of monotone graphs by forbidden induced subgraphs. The anticipated recognition algorithm first looks for flaws (defined in~\ref{ss:fp}) and then branches into different procedures depending on the length of a short hole (defined in~\ref{ss:short}) in the input
graph. We describe how to find such a hole in \ref{ss:short}. The
remaining forbidden subgraphs are preholes, also defined in \ref{ss:fp}.

Sections \ref{sec:long} and \ref{sec:prehole} deal with graphs containing a long
hole. Again we start with some technical lemmas showing that the long hole
enforces an annular structure in the absence of flaws. The structure is determined
by \emph{splitting}, described in \ref{ss:splitting}.
Possible preholes must wind round this annulus once or twice.  We complete the process
by checking for preholes, using a procedure given in \ref{ss:findprehole}.

If there is no long hole, we show that a minimal prehole consists of two triangles or 5-cycles
with two vertex-disjoint paths between them. We describe this in more detail in
section \ref{sec:nolongholes}. Clearly, we can list all triangles and 5-cycles in the
input graph in $O(n^5)$ time, so the difficulty arises from preholes of unbounded size.
Our algorithm requires listing only the triangles,
in $O(n^3)$ time, since we show that the prehole must be small when there is a 5-cycle..

Section \ref{sec:algorithm} summarises the algorithm with a formal description,
and discusses its running time.

In section~\ref{sec:NPc}, we give a short discussion of a central question raised in the paper,
recognising a prehole in an arbitrary graph. Though we  do not settle this question,
we show that a related question is \NP-complete. That is, given a graph, is it a prehole?

Finally, section \ref{sec:conclusions} concludes the paper.

\section{Quasi-classes and pre-graphs} \label{sec:quasi-pre}

A \emph{hereditary} class of graphs is closed under taking induced subgraphs.
Let $\bip$ denote the class of bipartite graphs, and let $\cC \subseteq \bip$.
Then we will say that the graph $G$ is \emph{$\qua\cC$} if $G[\LR]\in \cC$ for
all bipartitions $L,R$ of $V$.

\begin{lemma}
  If $\cC \subseteq \bip$ is a hereditary class that is closed under
  disjoint union then $\cC = \bip \cap \qua\cC$.
\end{lemma}

\begin{proof}
  First let $G = (L \cup R,E)$ be any bipartite graph that does not belong
  to $\cC$. Since $G = G[\LR]$ the graph $G$ does not belong to $\qua\cC$.
  Hence $\cC \supseteq \bip \cap \qua\cC$.

  Next we show $\cC \subseteq \bip \cap \qua\cC$. Let $G = (X \cup Y,E)$
  be a graph in $\cC$ and let $\LR$ be a bipartition of $X \cup Y$.
  Now $G[\LR]$ is the disjoint union of $G_1 = G[(X \cap L)\cup(Y \cap R)]$
  and $G_2 = G[(X \cap R)\cup(Y \cap L)]$. The graphs $G_1$ and $G_2$ belong
  to $\cC$ since the class is hereditary, and hence $G[\LR]$ is in $\cC$ because
  $\cC$ is closed under disjoint union. Thus $G \in \qua\cC$.
\end{proof}
A hereditary graph class can equally well be characterised by a set $\cF$ of forbidden subgraphs.
The set $\cF$ is minimal if no graph in $\cF$ contains any other as an induced subgraph.

For a bipartite graph $H$, a graph $G=(V,E)$ is a \emph{pre-$H$} if there is a
bipartition $L,R$ of $V$ such that $G[\LR]=H$. In this case $H$ is a spanning subgraph
of $G$. Clearly any bipartite $H$ is itself a pre-$H$.

\begin{lemma}
If $\cC \subseteq \bip$ is characterised by a set $\cF$ of forbidden
induced subgraphs, let pre-$\cF = \{\text{pre-}H \mid H \in \cF\}$.
Then $\qua\cC$ is characterised by the set of forbidden induced subgraphs
pre-$\cF$.
\end{lemma}

\begin{proof}
Suppose $G=(V,E)$ contains $H'=(V',E')$, a pre-$H$ for some $H\in\cF$.
Then $V'$ has a bipartition $L',R'$ such that $H'[L'$:$R']=H$.
Extending $L',R'$ to a bipartition $L,R$ of $V$, $G[\LR]$ contains $H$.
Then $G[\LR]\notin\cC$, so $G\notin\qua\cC$.  Conversely, if $G\in\qua\cC$,
every $G[\LR]\in\cC$, so no $G[\LR]$ contains $H$, for any $H\in\cF$. Thus
$G$ contains no pre-$H$, for any $H\in\cF$, that is, no $H'\in\textrm{pre-}\cF$.
\end{proof}
Note, however, that pre-$\cF$ may not be minimal for $\qua\cC$ when $\cF$ is minimal for $\cC$.
\subsection{Examples}\label{ss:examples}

The class $\qua\bip$ is clearly the set of all graphs.

If $\cC$ is the class of complete bipartite graphs, it is easy to see that
$\qua\cC$ is the class of complete graphs. Note however, that this class
is not closed under disjoint union. Now, if $\cC$ becomes the class of graphs
for which every component is complete bipartite, then $\qua\cC$ is the class of
graphs without $P_4$, paw or diamond. These three graphs are the pre-$P_4$'s,
see Fig.~\ref{fig:preP4}. To see this we observe the following:
\begin{itemize}
\item If a graph $G$ contains a pre-$P_4$ then there is a bipartition $G[\LR]$
  that contains a $P_4$ as induced subgraph. A connected component of $G[\LR]$
  containing such a $P_4$ is not complete bipartite.
\item Now $G$ does not contain a pre-$P_4$. If a connected component $H$ of a
  bipartition of $G$ is not complete bipartite, then $H$ contains a $P_4$,
  contradicting the fact that $G$ does not contain a pre-$P_4$.
\end{itemize}

\begin{figure}[htbp]
  \hspace*{\fill}
  \begin{tikzpicture}[scale=0.5]
    \node[i] (0) at (0,0) {}; \node[i] (2) at (0,2) {};
    \node[b] (a) at (0,1) {}; \node[w] (b) at (2,1) {};
    \node[b] (c) at (4,1) {}; \node[w] (d) at (6,1) {};
    \draw (a)--(b)--(c)--(d);
  \end{tikzpicture}
  \hfill
  \begin{tikzpicture}[scale=0.5]
    \node[b] (a) at (0,0) {}; \node[w] (b) at (0,2) {};
    \node[b] (c) at ($(a)+(30:2)$) {};
    \node[w] (d) at ($(c)+( 2,0)$) {};
    \draw (c)--(a)--(b)--(c)--(d);
  \end{tikzpicture}
  \hfill
  \begin{tikzpicture}[scale=0.5]
    \node[w] (b) at (2,0) {};
    \node[b] (c) at (2,2) {};
    \node[b] (a) at ($(b)+(150:2)$) {};
    \node[w] (d) at ($(b)+( 30:2)$) {};
    \draw (c)--(a)--(b)--(c)--(d)--(b);
  \end{tikzpicture}
  \hspace*{\fill}
  \caption{The pre-$P_4$'s: the path $P_4$, the paw and the diamond}
  \label{fig:preP4}
\end{figure}
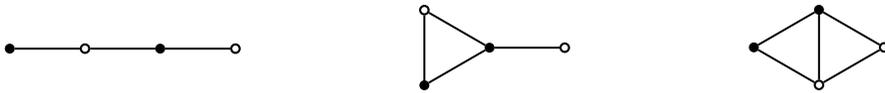

Bounded-degree graphs give another example. If $\cC_d$ is the class of bipartite graphs with degree at most $d$, for a fixed integer $d>0$, then $\qua\cC_d$ is the class of all graphs with degree at most $d$. To see this note that, if $v$ has degree at most $d$ in $G$, then $v$ has degree at most $d$ in any $G[\LR]$. Conversely, if $v$ has degree $d'>d$, then $v$ also has degree $d'$ in any $G[\LR]$ such that $v\in L$, $N(v)\subseteq R$. The unique forbidden subgraph for $\cC_d$ is clearly the star $K_{1,d+1}$. Therefore, the class $\qua\cC_d$ is characterised by  forbidding pre-$K_{1,d+1}$'s, a set with size $O(d^2)$. Hence $\qua\cC_d$ can be recognised in polynomial time, for fixed $d$.

A less obvious example is for the class $\cC$ of \emph{linear forests}, which
are disjoint unions of paths.
Its quasi-class contains all graphs with connected components that are either
a path or an odd cycle.

\chbip is the class of bipartite graphs in which every cycle of length 6 or more has a chord.
\och is the class of graphs in which every even cycle of length 6 or more has an odd chord.
We show in~\cite{DyeMu17a} that quasi-$\chbip=\och$.
However, the complexity of the recognition problem for the class \och is open,
even though the class \chbip can be recognised in linear time.
More generally, polynomial time recognition of $\cC$ does not
directly imply the same property for $\qua\cC$. All we can assert is membership in co-\NP,
by guessing a bipartition $L,R$ and showing in polynomial time that $G[\LR]\notin\cC$.

Finally, as remarked above, if \mon is the class of monotone graphs and
\qmon is the class of quasimonotone graphs, then $\qua\mon = \qmon$, by definition.
Most of this paper examines the structure and polynomial time recognition of graphs in this class.
As we have remarked, the linear-time recognition of \mon has little relevance to this issue.

A further example, \textsc{Quasichains}, is discussed in~\cite{DyeMu17a}. It is the quasi-class arising from unions of \emph{chain graphs} (defined in~\ref{ss:splitting} below), so $\textsc{Quasichains}\subset\qmon$.

Note that $\qua\cC$ does not necessarily inherit the properties as $\cC$. For example, consider \emph{perfection}. Since $\cC\subseteq\bip$, $\cC\subset\textsc{Perfect}$, the class of perfect graphs. But every bipartition of an odd hole is a linear forest. Thus, for any bipartite superclass of linear
forests, the quasi-class contains odd holes, which are imperfect. In particular, this holds for
the class \qmon. However, if $\cC$ is hereditary, closed under disjoint union,  since $\cC\subseteq\qua\cC$ and fails to have some property,then clearly $\qua\cC$ cannot have the property. Thus \qmon is not closed under edge deletions, since \mon is not.

However, if $\cC$ is closed under edge deletions, so is $\qua\cC$. To see this not that, if $G'=G\setminus e$ is deleted, then either $G'[\LR]=G[\LR]$ or $G'[\LR]=G[\LR]\setminus e$, so $G'[\LR]\in \cC$, and hence $G'\in\qua\cC$.

On the other hand, if $\cC$ is closed under edge contraction, $\qua\cC$ is unlikely to have this property. Any class which includes the cycle of length $\ell$, but excludes the cycle of length $\ell'<\ell$, is clearly not closed under edge contraction. Thus $\qua\linfor$ is not closed under edge contraction, even though \linfor is, since it includes all odd cycles, but no even cycle. In particular, $\qua\cC$ is unlikely to be \emph{minor-closed}, even when $\cC$ has this property, since this requires closure under edge contractions.

\section{The structure of quasimonotone graphs} \label{sec:qmg}
\subsection{Flaws and preholes} \label{ss:fp}
A bipartite graph is \emph{monotone} if and only if the rows and columns of
its biadjacency matrix can be permuted such that the ones appear consecutively
and the boundaries of these intervals are monotonic functions of the row or
column index. That is, all the ones are in a staircase-shaped region
in 
the biadjacency matrix.
Equivalent characterisations exist. For instance, a bipartite graph is monotone
if and only if it does not contain a hole, tripod, stirrer or armchair as
induced subgraph. The tripod, stirrer and armchair are depicted in Fig.~\ref{fig:tsa}. Monotone
graphs are also called \emph{bipartite permutation graphs}~\cite{SpBrSt87} and
\emph{proper interval bigraphs}~\cite{HelHua04}.

\begin{figure}[htbp]
  \hspace*{\fill}
  \begin{tikzpicture}[xscale=0.4,yscale=0.3]
    \node[b] (hh) at (3,6) {};                            
    \node[w] (l1) at (2,3) {}; \node[b] (l2) at (1,0) {}; 
    \node[w] (m1) at (3,3) {}; \node[b] (m2) at (3,0) {}; 
    \node[w] (r1) at (4,3) {}; \node[b] (r2) at (5,0) {}; 
    \foreach \leg in {l,m,r} \draw (hh)--(\leg1)--(\leg2);
  \end{tikzpicture}
  \hspace*{\fill}
  \begin{tikzpicture}[scale=0.3]
    \node[b] (hh) at (3,6) {};
    \foreach \x in {1,5} {
      \draw (\x,2) node[b] (\x!2) {} (\x,0) node[w] (\x!0) {};
    };
    \draw (3,0) node[b] (3!0) {} (3,2) node[w] (3!2) {} ;
    \draw (hh)--(3!2)--(3!0)--(1!0)--(1!2)--(3!2)--(5!2)--(5!0)--(3!0);
  \end{tikzpicture}
  \hspace*{\fill}
  \begin{tikzpicture}[xscale=0.3,yscale=0.26]
    \foreach \y in {0,4} \node[w] (0!\y) at (0,\y) {};
    \foreach \y in {2,7} \node[b] (0!\y) at (0,\y) {};
    \foreach \y in {0,4} \node[b] (3!\y) at (3,\y) {};
    \node[w] (3!2) at (3,2) {};
    \draw (0!0)--(0!2)--(0!4)--(0!7)  (3!0)--(3!2)--(3!4)--(0!4)  (0!2)--(3!2);
  \end{tikzpicture}
  \hspace*{\fill}
  \caption{The tripod, the stirrer and the armchair.}
  \label{fig:tsa}
\end{figure}
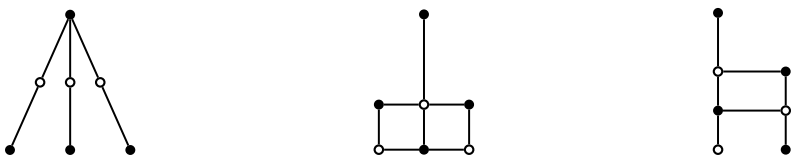

Let $\mon$ denote the class of monotone graphs, then the $\qmon$ will denote the class $\qua\mon$. Two example graphs are shown in Fig.~\ref{fig:quasiexample}.

Let $\flaw$ be the class containing all pre-tripods, pre-stirrers and pre-armchairs. We will say that any graph in $\flaw$ is a \emph{flaw}. A \emph{flawless} graph $G$ will be one which contains no flaw as an induced subgraph. Since all flaws have seven vertices, we can test in $O(n^7)$ time whether an input graph $G$ on $n$ vertices is flawless. Let $\noflaw$ denote the class of flawless graphs.

Therefore, quasimonotone graphs are characterised by the absence of \pre
holes, \pre tripods, \pre stirrers and \pre-armchairs. Let $\qmon$ be the
class of quasimonotone graphs. Clearly $\qmon\subseteq \ehf\,\cap\,\noflaw$,
but equality does not hold, as we now discuss.

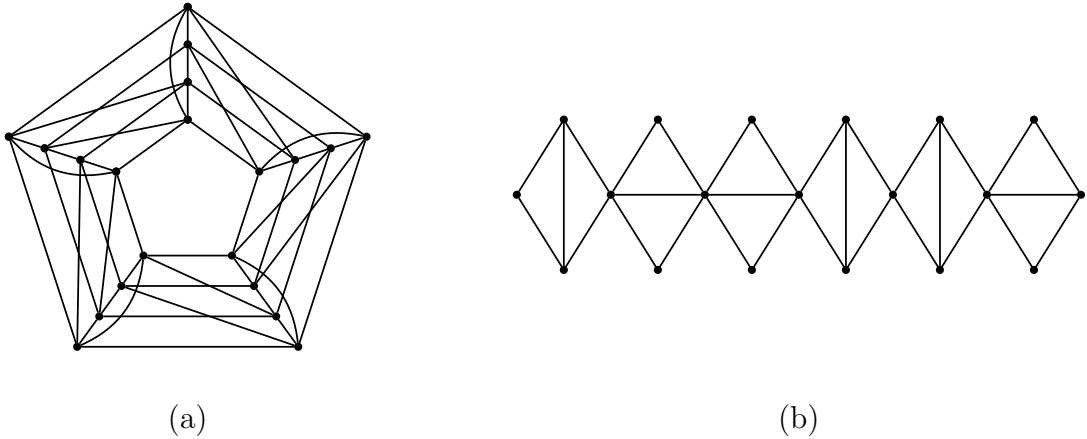
\begin{figure}[ht]
\begin{center}
  \tikzset{every node/.style={circle,draw,thick,fill=black,inner sep=0.8pt}}
  \begin{tikzpicture}[scale=0.5,line width=0.67pt]
    \foreach[count=\r from 2] \n in {11,12,13,14} \node (\n) at ( 90:\r) {};
    \foreach[count=\r from 2] \n in {15,16,17,18} \node (\n) at (162:\r) {};
    \foreach[count=\r from 2] \n in {19,20,21,22} \node (\n) at (234:\r) {};
    \foreach[count=\r from 2] \n in {23,24,25,26} \node (\n) at (306:\r) {};
    \foreach[count=\r from 2] \n in {27,28,29,30} \node (\n) at ( 18:\r) {};
    \foreach[count=\r from 2] \n in {31,32,33,34} \node (\n) at ( 90:\r) {};
    \foreach[count=\r from 2] \n in {35,36,37,38} \node (\n) at (162:\r) {};
    \foreach[count=\m from 15] \n in {11,12,...,30} \draw (\n)--(\m);
    \draw (11)--(12)--(13)--(14) to[bend right] (11)
          (15)--(16)--(17)--(18) to[bend right] (15)
          (19)--(20)--(21)--(22) to[bend right] (19)
          (23)--(24)--(25)--(26) to[bend right] (23)
          (27)--(28)--(29)--(30) to[bend right] (27);
    \draw (11)--(17)  (12)--(18)
          (15)--(21)  (16)--(22)
          (19)--(25)  (20)--(26)
          (23)--(29)  (24)--(30)
          (27)--(33)  (28)--(34);
    \draw (0,-6) node[empty] {(a)} ;
  \begin{scope}[xscale=1.25,xshift=6cm]
    \foreach \n in {2,4,6,8,10,12} \node (b\n) at (\n,-2) {} ;
    \foreach \n in {2,4,6,8,10,12} \node (t\n) at (\n,2) {} ;
    \foreach \n in {1,3,5,7,9,11,13} \node (m\n) at (\n,0) {} ;
    \foreach \n in {1,3,5,7,9,11} {\draw (m\n) -- ++(1,2) ;};
    \foreach \n in {3,5,7,9,11,13} {\draw (m\n) -- ++(-1,2);};
    \foreach \n in {1,3,5,7,9,11} {\draw (m\n) -- ++(1,-2) ;};
    \foreach \n in {3,5,7,9,11,13} {\draw (m\n) -- ++(-1,-2);};
    \draw (t2)--(b2) (t8)--(b8) (m3)--(m5)--(m7) (t10)--(b10) (m11)--(m13) ;
    \draw (7,-6) node[empty] {(b)} ;
  \end{scope}
  \end{tikzpicture}
\end{center}
\caption{Two quasimonotone graphs}\label{fig:quasiexample}
\end{figure}
Let $P=(p_1,p_2,\dots,p_\ell)$ be a path in $G$. The \emph{alternating bipartition} $L,R$ of $P$ assigns $L=\{p_1,p_3,\ldots\}$ and $R=\{p_2,p_4,\ldots\}$. We will say that $P$ is \emph{prechordless} if it is an induced path in $G[\LR]$. In particular, any induced path in $G$ is prechordless. Similarly, let $C=(p_1,p_2,\ldots,p_\ell)$ be an even cycle in $G$. Then $C$ is a prehole if it is a hole in $G[\LR]$. Thus $C$ must be an even cycle, and all chords must run between $L$ and $L$ or $R$ and $R$ in an alternating bipartition $L,R$ of $C$. This is equivalent to requiring that $C$ has no odd chord. The alternating partition is inconsistent for an odd cycle, so an odd cycle $C$ cannot be a prehole.

From this discussion, it is clear that $G$ contains no prehole if and only if it is odd-chordal, as defined in section~\ref{ss:examples}. Thus it follows that $\qmon=\noflaw\cap\och$.
Given an input graph $G$
, we wish to test whether or not $G\in \qmon$. We can test whether $G\in\noflaw$ in polynomial time, but we do not know how to determine whether $G\in\och$. Thus it is not clear that this \qmon can be recognised in polynomial time, since preholes can be of arbitrary size in \noflaw. See Fig,~\ref{prehole:example} for a family of such preholes.

The main contribution of this paper will be to show that the recognition problem for \qmon is indeed in polynomial time. We will not be too concerned with the efficiency of our algorithm beyond polynomiality, so the bounds we prove will often be far from optimal.

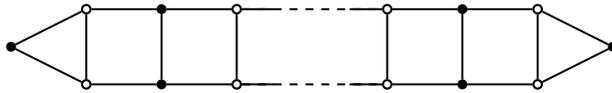
\begin{figure}[H]
\tikzset{every node/.style={circle,draw,fill=none,inner sep=0pt,minimum size=1.25mm}}
\centerline{\begin{tikzpicture}[xscale=1,yscale=0.5,font=\small]
\draw (0,0) node[b] (0) {}
 (1,1) node[w] (1) {} (1,-1) node[w] (1') {}
 (2,1) node[b] (2) {} (2,-1) node[b] (2') {}
 (3,1) node[w] (3) {} (3,-1) node[w] (3') {}
 (5,1) node[w] (5) {} (5,-1) node[w] (5') {}
 (6,1) node[b] (6) {} (6,-1) node[b] (6') {}
 (7,1) node[w] (7) {} (7,-1) node[w] (7') {}
 (8,0) node[b] (8) {} ;
 \draw (3')--(2')--(1')--(0)--(1)--(2)--(3) (5')--(6')--(7')--(8)--(7)--(6)--(5) ;
 \draw (3)--++(0.5,0) (5)--++(-0.5,0) (3')--++(0.5,0) (5')--++(-0.5,0) ;
 \draw[dashed] (3)--(5) (3')--(5') ;
 \draw (1)--(1') (2)--(2') (3)--(3') (5)--(5') (6)--(6') (7)--(7') ;
\end{tikzpicture}}
\caption{An infinite family of preholes}\label{prehole:example}
\end{figure}

\subsection{Properties of flawless graphs}

First we prove a useful lemma about graphs in $\noflaw$.

\begin{lemma}\label{lem:quasi10}
Let $G\in\noflaw$, $P=(p_1,p_2,p_3,p_4,p_5,p_6,p_7)$ be a prechordless path in $G$, $(p_2,p_3,p_4,p_5,p_6)$ be a hole in $G$, or $(p_1,p_2,p_3,p_4,p_5,p_6)$ be a prehole in $G$. If $v\notin P$ is such that $\dist(v,P)=\dist(v,p_4)$, then $\dist(v,p_4)= 1$.
\end{lemma}
\begin{proof}
Clearly, the shortest path from $v$ to $p_4$ cannot use any edge of $P$. Therefore,
suppose, without loss of generality, that $\dist(v,P)=2$, and $(v,u,p_4)$ is the shortest path from $v$ to $P$. Consider the alternating bipartition of $P$ extended to $u\in R$ and $v\in L$, as shown by the black ($L$) and white ($R$) vertices in Fig.~\ref{fig:dist2}. There are no edges from $v$ to $P$, since $\dist(v,P)>1$, and there are no edges between vertices of $P$ in $G[\LR]$, since $P$ is prechordless.

\begin{figure}[H]
\centering{
\begin{tikzpicture}[line width=0.5pt,minimum size=3pt,inner sep=1.5pt,scale=0.9,font=\small]
\draw (0,0) node[circle,b,label=below:$p_2$] (1){} (1,0) node[circle,draw,fill=none,label=below:$p_3$] (2){}
(2,0) node[circle,b,label=below:$p_4$] (3){} (3,0) node[circle,draw,fill=none,label=below:$p_5$] (4){}
(4,0) node[circle,b,label=below:$p_6$] (5){};
\draw (1)--(2)--(3)--(4)--(5);
\draw (2,1) node[circle,draw,fill=none,label=left:$u$] (6){} (2,2) node[circle,b,label=above:$v$] (7){} (3)--(6)--(7);
\draw (2,-0.75) node {$P$};
\end{tikzpicture}
\hspace*{1.5cm}
\begin{tikzpicture}[line width=0.5pt,minimum size=3pt,inner sep=1.5pt,scale=0.9,font=\small]
\draw (0,0) node[circle,draw,fill=none,label=below:$p_3$] (1){} (1,0) node[circle,b,label=below:$p_4$] (2){} (2,0) node[circle,draw,fill=none,label=below:$p_5$] (3){} (3,0) node[circle,b,label=below:$p_6$] (4){}
(4,0) node[circle,draw,fill=none,label=below:$p_7$] (5){};
\draw (1)--(2)--(3)--(4)--(5);
\draw (1,1) node[circle,draw,fill=none,label=left:$u$] (6){} (1,2) node[circle,b,label=above:$v$] (7){} (2)--(6)--(7) (6)--(4);
\draw (2,-0.75) node {$P$};
\end{tikzpicture}
\hspace*{1.5cm}
\begin{tikzpicture}[line width=0.5pt,minimum size=3pt,inner sep=1.5pt,scale=0.9,font=\small]
\draw (0,0) node[circle,b,label=below:$p_2$] (1){} (1,0) node[circle,draw,fill=none,label=below:$p_3$] (2){}
(2,0) node[circle,b,label=below:$p_4$] (3){} (3,0) node[circle,draw,fill=none,label=below:$p_5$] (4){}
(4,0) node[circle,b,label=below:$p_6$] (5){};
\draw (1)--(2)--(3)--(4)--(5);
\draw (2,1) node[circle,draw,fill=none,label=north west:$u$] (6){} (2,2) node[circle,b,label=above:$v$] (7){} (3)--(6)--(7) (5)--(6)--(1);
\draw (2,-0.75) node {$P$};
\end{tikzpicture}
}
\caption{}\label{fig:dist2}
\end{figure}
Thus the only possible edges in $G[\LR]$, other than in $P$ and the 2-path $(v,u,p_3)$, are those joining $u$ to a vertex in $L$. There are three cases, where none, one or both of these edges are present, as shown in Fig.~\ref{fig:dist2}. Note that the ``one'' case has a symmetric version, where the  vertices in $P$ are $p_1,p_2,p_3,p_4,p_5$, and the edge $up_2$ is present. But these graphs are the tripod, armchair  and stirrer, respectively, contradicting $G$ being flawless.

If $(p_2,p_3,p_4,p_5,p_6)$ is a hole, or $(p_1,p_2,p_3,p_4,p_5,p_6)$ is a prehole, we need only observe that the configurations in  Fig.~\ref{fig:dist2} exist, if $P$ is allowed to ``wrap around'' $C$. That is, if $p_1,p_7$ are interpreted as $p_6,p_2$ respectively.
\end{proof}
Note that any subpath of a prehole or odd hole $C$ is prechordless.
\begin{lemma}\label{lem:quasi20}
  Every odd hole or prehole in a connected flawless graph is dominating.
\end{lemma}
\begin{proof}
  Let $C$ be an odd hole or prehole in the connected flawless graph $G$. We show $\dist(v,C)\le1$ for every vertex $v$ of $G$.

If $v\in C$, this is obvious. Otherwise, let $w$ be a vertex such that $\dist(v,C)=\dist(v,w)$.  Consider the subpath $P=(p_1,p_2,\ldots,p_7)$ of $C$ such that $w=p_4$, where this path wraps around $C$ if $|C|<7$.
Since $C$ is a hole or a prehole, $P$ is prechordless. The result then follows from Lemma~\ref{lem:quasi10}.
\end{proof}
If $C$ is an odd hole we will call $n(C)=\{v\in V:\dist(v,C)\leq 1\}$, the neighbourhood of $C$. Thus, if $G$ is connected, then $G=N(C)$ for any odd hole $C\subseteq G$.
\begin{lemma}\label{lem:quasi30}
 Suppose $G\in\noflaw\cap\thsp\ehf$, and that $C$ is an odd hole in $G$, of length at least seven. Then every vertex $v\in V$ has at most three neighbours in $C$. If there are two neighbours, $w,x$, then $\dist_C(w,x)=2$. If there are three neighbours, $w,x,y$, then $\dist_C(w,x)=\dist_C(x,y)=2$. If $C$ is a short odd hole in $G$, then $v$ has at most two neighbours on $C$.
\end{lemma}
\begin{proof}
If $v\in C$, $v$ has exactly two neighbours in $C$, so the lemma is true. Thus suppose $v\notin C$.
If $w$ is the only neighbour of $v\in C$, then $G[\{v\}\cup C]$ is the graph shown in Fig.~\ref{fig:dist1}. Note that $w$ may have several such neighbours $v_1,v_2,\ldots$, but no two can be connected by an edge, since otherwise there is a bipartition containing a tripod, see~Fig.~\ref{fig:dist1}. If $v_1,v_2$ are leaves in $G$, then they vertices are false twin.
\begin{figure}[H]
\centering{
\begin{tikzpicture}[line width=0.5pt,minimum size=3pt,inner sep=1.5pt,scale=0.8,font=\small]
\draw (0,0) node[circle,draw,w] (1){} (1,0.5) node[circle,draw,b] (2){}
(2,0.75) node[circle,draw,w,label=below:$w$] (3){} (3,0.5) node[circle,draw,b] (4){}
(4,0) node[circle,draw,w] (5){};
\draw (1)--(2)--(3)--(4)--(5);
\draw (2,1.4) node[circle,draw,b,label=above:$v$] (6){} (3)--(6);
\draw (2,-0.25) node {$C$}; \draw (-0.5,-0.5)--(1) (4.5,-0.5)--(5) ;
\end{tikzpicture}\hspace*{1.5cm}
\begin{tikzpicture}[line width=0.5pt,minimum size=3pt,inner sep=1.5pt,scale=0.8,font=\small]
\draw (0,0) node[circle,draw,w] (1){} (1,0.5) node[circle,draw,b] (2){}
(2,0.75) node[circle,draw,w,label=below:$w$] (3){} (3,0.5) node[circle,draw,b] (4){}
(4,0) node[circle,draw,w] (5){};
\draw (1)--(2)--(3)--(4)--(5);
\draw (1.6,1.4) node[circle,draw,b,label=above:$v_1$] (6){}
(2.4,1.4) node[circle,draw,b,label=above:$v_2$] (7){} (7)--(3)--(6);
\draw (2,-0.25) node {$C$}; \draw (-0.5,-0.5)--(1) (4.5,-0.5)--(5) ;
\end{tikzpicture}\hspace*{1.5cm}
\begin{tikzpicture}[line width=0.5pt,minimum size=3pt,inner sep=1.5pt,scale=0.8,font=\small]
\draw (0,0) node[circle,draw,w] (1){} (1,0.5) node[circle,draw,b] (2){}
(2,0.75) node[circle,draw,w,label=below:$w$] (3){} (3,0.5) node[circle,draw,b] (4){}
(4,0) node[circle,draw,w] (5){};
\draw (1)--(2)--(3)--(4)--(5);
\draw (1.6,1.4) node[circle,draw,b,label=above:$v_1$] (6){}
(2.4,1.4) node[circle,draw,w,label=above:$v_2$] (7){} (7)--(6)--(3) (7)edge[thick,densely dotted](3);
\draw (2,-0.25) node {$C$}; \draw (-0.5,-0.5)--(1) (4.5,-0.5)--(5) ;
\end{tikzpicture}
}
\caption{}\label{fig:dist1}
\end{figure}
Now suppose $v$ has two neighbours $w,x$ on $C$. Since $\ell=|C|\geq 7$ is odd, we may assume $\nu=\dist_C(w,x)\leq(\ell-1)/2$. Suppose first that $\nu$ is odd. Then we obtain an even cycle $C'$ by omitting $\nu-1$ vertices of $C$ and adding $v$.  Consider the alternating bipartition $\LR$ of $C'$. See Fig.~\ref{fig:dist1a}, with black nodes $L$, and white nodes $R$. Now $C'$ has even length $\ell'=\ell-\nu+2\geq (\ell+5)/2\geq 6$. Then, since $C$ is chordless, there is a an even hole in $G[\LR]$, unless $v$ is adjacent to every vertex of $L$.   However,  since $\ell\geq 7$, $G[\LR]$ contains a stirrer,  as shown in Fig.~\ref{fig:dist1a}, contradicting $G\in\noflaw$. Thus $\dist_C(w,x)$ cannot be odd.
\begin{figure}[H]
\centering{
\begin{tikzpicture}[line width=0.5pt,minimum size=3pt,inner sep=1.5pt,xscale=0.65,yscale=0.6,font=\small]
\foreach \x in {1,3,5,7} {\draw (90-180/7+360*\x/7:3cm) node[circle,draw,w](\x) {} ;};
\draw (90-180/7+360/7:3.4cm) node {$w$} (90-180/7+360:3.4cm) node {$x$} ;
\foreach \x in {2,4,6} {\draw (90-180/7+360*\x/7:3cm) node[circle,draw,b] (\x) {} ;};
\draw (0,3.5) node[circle,draw,b,label=above:$v$] (0) {} ;
\draw[densely dashed] (7)--(6)--(5) ; \draw[densely dotted] (1)--(7) ;
\draw[thick] (0)--(1)--(2)--(3)--(0)--(5)--(4)--(3) (0)--(7) ;
\end{tikzpicture}
}
\caption{}\label{fig:dist1a}
\end{figure}
Thus suppose $\dist_C(w,x)=\nu\leq (\ell-1)/2$ is even and $\nu>2$, so $\nu\geq 4$. Then $v,w,x$ lie on a chordless cycle in $G$ of even length $\nu+2\geq6$. This is an even hole, contradicting $G\in\ehf$, so we must have $\dist_C(w,x)=2$. Then $G[\{v\}\cup C]$ is the graph shown in Fig.~\ref{fig:dist1b}, and $G$ has another odd hole of length $\ell$, passing through $v$.

There can be several vertices $v_1,v_2,\ldots$ with neighbours $w$ and $x$, but there can be no edge between any pair of these vertices. Otherwise there is a bipartition containing an armchair. See Fig.~\ref{fig:dist1b}. If these vertices have neighbours only in $C$, then they are all false twin.
\begin{figure}[H]
\centering{
\begin{tikzpicture}[line width=0.5pt,minimum size=3pt,inner sep=1.5pt,scale=0.9,font=\small]
\draw (0,0) node[circle,draw,b] (1){} (1,0.5) node[circle,draw,w,label=below:$w$] (2){}
(2,0.75) node[circle,draw,b] (3){} (3,0.5) node[circle,draw,w,label=below:$x$] (4){}
(4,0) node[circle,draw,b] (5){};
\draw (1)--(2)--(3)--(4)--(5);
\draw (2,1.8) node[circle,draw,w,label=above:$v$] (6){} (2)--(6)--(4) ;
\draw (2,-0.25) node {$C$}; \draw (-0.5,-0.5)--(1) (4.5,-0.5)--(5) ;
\end{tikzpicture}\hspace*{3cm}
\begin{tikzpicture}[line width=0.5pt,minimum size=3pt,inner sep=1.5pt,scale=0.9,font=\small]
\draw (0,0) node[circle,draw,b] (1){} (1,0.5) node[circle,draw,w,label=below:$w$] (2){}
(2,0.75) node[circle,draw,b] (3){} (3,0.5) node[circle,draw,w,label=below:$x$] (4){}
(4,0) node[circle,draw,b] (5){};
\draw (1)--(2)--(3)--(4)--(5);
\draw (2,1.8) node[circle,draw,b,label=below:$\strut v_1$] (6){}
(2,2.7) node[circle,draw,w,label=above:$v_2$] (7){} (2)--(6)--(4) (2)--(7)--(4) (7)--(6) ;
\draw (2,-0.25) node {$C$}; \draw (-0.5,-0.5)--(1) (4.5,-0.5)--(5) ;
\end{tikzpicture}
}
\caption{}\label{fig:dist1b}
\end{figure}
Now suppose $v$ has at least three neighbours $w,x,y$ on $C$, where $w,y$ are such that $\dist_C(w,y)$ is maximised. Consider the alternating bipartition $\LR$ of the \Path{w}{y} in $C$, extended to $v\in R$. Then $v$ must be adjacent to every vertex in $R$ between $w$ and $y$, since otherwise there is an even hole in $G[\LR]$. If $v$ has exactly three neighbours, $G[\LR]$ contains the subgraph  shown in Fig.~\ref{fig:dist1c}, with $L$, $R$ are the white and black vertices, respectively. Now $v$ cannot have a fourth neighbour $z$ on $C$. Otherwise, $G$ has a stirrer, involving $v$, the \Path{w}{z}, and $y$, as shown in Fig.~\ref{fig:dist1c}.
\begin{figure}[H]
\centering{
\begin{tikzpicture}[line width=0.5pt,minimum size=3pt,inner sep=1.5pt,scale=0.9,font=\small]
\draw (0,0) node[circle,draw,b,label=below:$w$] (1){} (1,0.5) node[circle,draw,w] (2){}
(2,0.75) node[circle,draw,b,label=below:$x$] (3){} (3,0.5) node[circle,draw,w] (4){}
(4,0) node[circle,draw,b,label=below:$y$] (5){};
\draw (1)--(2)--(3)--(4)--(5);
\draw (2,1.75) node[circle,draw,w,label=above:$v$] (6){} (1)edge[bend left=20](6) (6)edge[bend left=20](5) (6)--(3);
\draw (2,-0.25) node {$C$}; \draw (-0.5,-0.5)--(1) (4.5,-0.5)--(5) ;
\end{tikzpicture}
\hspace*{2cm}
\begin{tikzpicture}[line width=0.5pt,minimum size=3pt,inner sep=1.5pt,yscale=0.7,xscale=0.9,font=\small]
\draw  (0,0) node[circle,draw,b,label=below:$w$] (1){} (1,0.5) node[circle,draw,w] (2){} (2,0.75) node[circle,draw,b,label=below:$x$] (3){} (3,0.75) node[circle,draw,w] (4){}
(4,0.67) node[circle,draw,b,label=below:$z$] (5){} (5,0.45) node[circle,draw,w] (6){} (6,0) node[circle,draw,b,label=below:$y$] (7){} ;
\draw (1)--(2)--(3)--(4)--(5)--(6)--(7) ;
\draw (2,1.75) node[circle,draw,w,label=above:$v$] (9){} (1)edge[bend left=10](9) (9)edge[bend left=10](5) (9)edge[bend left=25](7);
\draw (2,-0.5) node {$C$}; \draw  (9)--(3) ;
\end{tikzpicture}
}
\caption{}\label{fig:dist1c}
\end{figure}
Finally, if $v$ has three neighbours in $C$, as shown in Fig.~\ref{fig:dist1c}, then $\dist(w,y)=2$, and $\dist_C(w,y)=4$. Thus $C$ is not a short odd hole, a contradiction.
\end{proof}
The following is similar to Lemma~\ref{lem:quasi30}, but the details of the proof are slightly different.
\begin{lemma}\label{lem:quasi35}
  Let $C$ be a prehole in $G\in\noflaw$. Then every vertex $v\in C$ has at most five neighbours in $C$. Two of these are via edges of $C$, so $v$ is incident to at most three chords.
If there are two chords, $vw,vx$, then $\dist_C(w,x)=2$. If there are three chords, $vw,vx,vy$, then $\dist_C(w,x)=\dist_C(x,y)=2$.
\end{lemma}
\begin{proof}
  Otherwise, $v$ must have at least four chords. These must be even chords to $c_0,c_2,c_4,c_6$, where $P=(c_0,c_1,\ldots,c_6,c_7)$ is a subpath of $C$, since $C$ is a prehole and $G$ has no even holes. We now move $v$ from $L$ to $R$. The only new edges which appear in $G[\LR]$ are those adjacent to $v$. But now $c_0,v,c_3,c_4,c_5,c_6,c_7$ induce an armchair in $G[\LR]$, contradicting $G\in\noflaw$.   See Fig.~\ref{prehole:fig20}.
\end{proof}
\begin{figure}[H]
\tikzset{every node/.style={circle,draw,fill=none,inner sep=0pt,minimum size=1.25mm}}
\centerline{\begin{tikzpicture}[xscale=0.4,yscale=0.5,font=\small]
\draw
(-2,1) node[label=above:$c_0$] (0) {} (0,1) node[b,label=above:$c_1$]  (1) {}  (2,1)  node[label=above:$c_2$] (2) {} (4,1)  node[b,label=above:$c_3$] (3) {}  (6,1)  node[label=above:$c_4$] (4) {}
(8,1)  node[b,label=above:$c_5$] (5) {}  (10,1)  node[label=above:$c_6$] (6) {}   (12,1)  node[b,label=above:$c_7$] (7) {}
(3,-3.5)  node[label=below:$v$] (v) {} (1,-3.5)  ;
\draw (0)--(1)--(2)--(3)--(4)--(5)--(6)--(7) ;
\draw[densely dashed] (v)--(0) (v)--(2) (v)--(4) (v)--(6) ;
\draw[thin] (0)edge[bend right=40](v) (7)edge[bend left=25](v) ;
\end{tikzpicture}\hspace*{0.75in}
\begin{tikzpicture}[xscale=0.4,yscale=0.5,font=\small]
\draw
(-2,1) node[label=above:$c_0$] (0) {} (0,1) node[b,label=above:$c_1$]  (1) {}  (2,1)  node[label=above:$c_2$] (2) {} (4,1)  node[b,label=above:$c_3$] (3) {}  (6,1)  node[label=above:$c_4$] (4) {}
(8,1)  node[b,label=above:$c_5$] (5) {}  (10,1)  node[label=above:$c_6$] (6) {}  {} (12,1)  node[b,label=above:$c_7$] (7) {} (3,-3.5)  node[b,label=below:$v$] (v) {} ;
\draw[very thin] (0)--(1)--(2)--(3) (v)--(2)  ;
\draw[thick] (v)--(0) (v)--(4) (v)--(6) (3)--(4)--(5)--(6)--(7) ;
\draw[very thin] (0)edge[bend right=40](v) (7)edge[bend left=25](v) ;
\end{tikzpicture}}
\caption{An armchair}\label{prehole:fig20}
\end{figure}
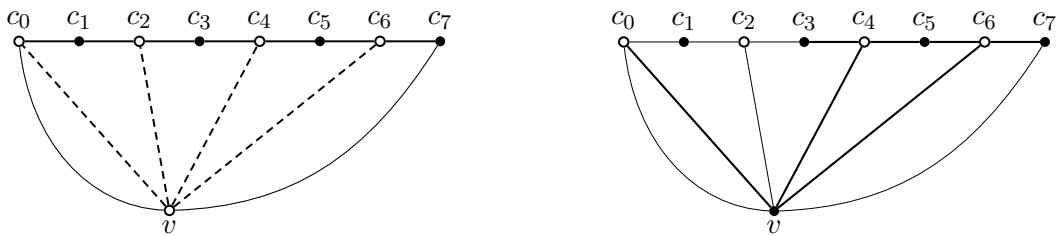
The degree bound of Lemma~\ref{lem:quasi35} is tight. See Fig.~\ref{fig:prehole}.
\begin{figure}[H]
\tikzset{every node/.style={circle,draw,fill=none,inner sep=0pt,minimum size=1.25mm}}
\centerline{\begin{tikzpicture}[xscale=1,yscale=0.5,font=\small]
\draw (0,0) node[b] (0) {}
 (1,1) node[w] (1) {} (1,-1) node[w] (1') {}
 (2,1) node[b] (2) {} (2,-1) node[b] (2') {}
 (3,1) node[w] (3) {} (3,-1) node[w] (3') {}
 (4,1) node[b] (4) {} (4,-1) node[b] (4') {}
 (5,1) node[w] (5) {} (5,-1) node[w] (5') {}
 (6,1) node[b] (6) {} (6,-1) node[b] (6') {}
 (7,1) node[w] (7) {} (7,-1) node[w] (7') {}
 (8,0) node[b] (8) {} ;
 \draw (3')--(2')--(1')--(0)--(1)--(2)--(3) (5')--(6')--(7')--(8)--(7)--(6)--(5)  (3)--(5) (3')--(5') ;
 \draw (1)--(1') (2)--(2') (3)--(3') (4)--(4') (5)--(5') (6)--(6') (7)--(7') (2')--(4)--(6');
\end{tikzpicture}}
\caption{A prehole with a vertex of degree 5}\label{fig:prehole}
\end{figure}
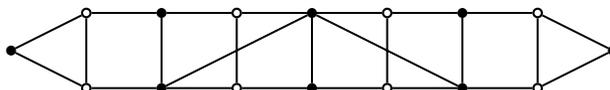

\begin{lemma}\label{lem:quasi40}
 Let $C$ be an odd hole in $G\in\noflaw$. Suppose that $v\notin C$ has a neighbour $x\notin C$. Then there are vertices $w,y\in C$ such that $(v,x,y,w)$ is a quadrangle. \end{lemma}
\begin{proof}

From Lemmas~\ref{lem:quasi20} and~\ref{lem:quasi30}, $v$ and $x$ have at least one, and at most three, neighbours on $C$.  If either has two neighbours then these are at distance 2 on $C$. Observe that it suffices to prove that $w,y$ exist so that $vxyw$ is simply a 4-cycle. If it is not a quadrangle, either $v$ or $x$ has two adjacent neighbours on $C$, contradicting Lemma~\ref{lem:quasi30}.

Suppose $v$ has exactly one neighbour $c$ on $C$. Then we have the first configuration shown in Fig.~\ref{fig:dist1d}, with the bipartition indicated by the black and white vertices on $G[S]$, where $S=\{a,b,c,d,e,v,x\}$. This contains a tripod, unless $x$ is adjacent to one of either $b,d$ or both. If is adjacent to $b$, then the Lemma follows with $w=c$ and $y=b$. If is adjacent to $d$, then the Lemma follows with $w=c$ and $y=d$.

\begin{figure}[ht]
\centering{
\begin{tikzpicture}[line width=0.5pt,minimum size=3pt,inner sep=1.5pt,scale=0.9,font=\small]
\draw (0,0) node[circle,draw,b,label=below:$a$] (1){} (1,0.5) node[circle,draw,w,label=below:$b$] (2){}(2,0.75) node[circle,draw,b,label=below:$c$] (3){} (3,0.5) node[circle,draw,w,label=below:$d$] (4){} (4,0) node[circle,draw,b,label=below:$e$] (5){};
\draw (1)--(2)--(3)--(4)--(5);
\draw (2,1.8) node[circle,draw,w,label=above:$v$] (6){}
(3.1,1.8) node[circle,draw,b,label=above:$x$] (7){} (6)--(3) (6)--(7) ;
\draw (2,-0.25) node {$C$}; \draw (-0.5,-0.5)--(1) (4.5,-0.5)--(5)
(7)edge[densely dotted,bend right=15](2)
(7)edge[densely dotted](4);
\end{tikzpicture}\hspace*{2cm}
\begin{tikzpicture}[line width=0.5pt,minimum size=3pt,inner sep=1.5pt,scale=0.9,font=\small]
\draw (0,0) node[circle,draw,b,label=below:$a$] (1){} (1,0.5) node[circle,draw,w,label=below:$b$] (2){}(2,0.75) node[circle,draw,b,label=below:$c$] (3){} (3,0.5) node[circle,draw,w,label=below:$d$] (4){} (4,0) node[circle,draw,b,label=below:$e$] (5){};
\draw (1)--(2)--(3)--(4)--(5);
\draw (2,1.8) node[circle,draw,b,label=above:$v$] (6){}
(3.1,1.8) node[circle,draw,w,label=above:$x$] (7){} (2)--(6)--(4) (6)--(7) ;
\draw (2,-0.25) node {$C$}; \draw (-0.5,-0.5)--(1) (4.5,-0.5)--(5)
(7)edge[densely dotted,bend right=15](1)
(7)edge[densely dotted](3)
(7)edge[densely dotted](5);
\end{tikzpicture}
}
\caption{}\label{fig:dist1d}
\end{figure}

Now suppose $v$ has two neighbours $b,d$ on $C$. Then, by Lemma~\ref{lem:quasi30},we have the second configuration in Fig.~\ref{fig:dist1d}, with the bipartition on $G[S]$ shown in black and white. If $v$ has a third neighbour $z$ on $C$, then $z\notin S$ so we may ignore it.  Then $G[S]$ contains an armchair, unless $x$ is adjacent to at least one of the black vertices $a,c,e$. If $xa\in E$, we take $w=b,y=a$, if $xc\in E$,  we take $w=b,y=c$, and if $xe\in E$, we take $w=d,y=e$.
\end{proof}

\subsection{Determining a short odd hole}\label{ss:short}

We can test whether $G$ contains a hole in time $O(|E|^2)$, using the algorithm of~\cite{NikPal07}. Moreover, the algorithm returns a hole if one exists. If the hole is even, we can conclude $G\notin\qmon$. If $G\in\noflaw$, we will show that it has a well-defined structure, so it is possible that there is a faster algorithm than~\cite{NikPal07} for detecting a hole. However, we will not pursue this here.

We begin with a simple result.
\begin{lemma}\label{lem:oddcycle}
If $C$ is an odd cycle in a graph $G$, there is a triangle or an odd hole $C'$ in $G$.
\end{lemma}
\begin{proof}
The claim is clearly true if $|C|\leq 3$. Otherwise, assume by induction that it is true for all cycles shorter than $C$. If $C$ is not already a hole, it has a chord that divides it into a smaller odd cycle $C_1$, and an even cycle $C'_1$. The lemma now follows by induction on $C_1$.
\end{proof}
The proof of Lemma~\ref{lem:oddcycle} can easily be turned into an efficient algorithm to find $C'$.
Let $C$ be an odd hole in a  graph $G$. Then $C$ will be called a \emph{short} odd hole in $G$ if $\dist(v,w)=\dist_C(v,w)$ for all pairs $v,w\in C$.
\begin{lemma}\label{lem:shorthole}
If $G$ is a triangle-free graph containing an odd hole $C$, then $G$ contains a short odd hole.
\end{lemma}
\begin{proof}
  Clearly $\dist(v,w)\leq \dist_C(v,w)$ for all pairs $v,w\in C$. Thus,  suppose that $C$ is an odd hole in $G$, but there is a pair $v,w$ such that $\dist(v,w)=d< \ell=\dist_C(v,w)$, and let $\ell'=|C|-\ell\geq\ell> d$. Thus one of $\ell,\ell'$ is odd and the other even. See Fig.~\ref{fig:short}.
  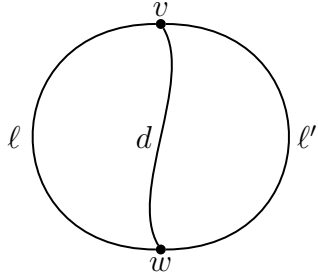
\begin{figure}[H]
\centering{
\begin{tikzpicture}[line width=0.75pt,minimum size=2pt,inner sep=0.9pt,scale=1.5]
\draw (2,0) node[b,label=below:$w$] (w) {} (2,2) node[b,label=above:$v$] (v) {};
\draw (w) .. controls (0.5,0) and (0.5,2) .. (v)  (w) .. controls (3.5,0) and (3.5,2) .. (v)
(w) .. controls (1.7,0.5) and (2.3,1.5) .. (v);
\draw (0.7,1) node {$\ell$} (3.3,1) node {$\ell'$}(1.85,1) node {$d$} ;
\end{tikzpicture}
}
\caption{Odd hole $C$ with a shorter $vw$ path}\label{fig:short}
\end{figure}
We may assume that the shortest $vw$ path $P$ has no internal vertex in common with $C$. Otherwise, we may choose a different pair $v,w$ for which this is true. Thus we can form two cycles $C_1,C_2$ of lengths $\ell+d,\ell'+d$. Now one of $\ell+d,\ell'+d$ is odd and the other even. Also $\max\{\ell+d,\ell'+d\}=\ell'+d=|C|-\ell+d<|C|$. Thus we have an odd cycle, $C_1$ say, with $|C_1|<|C|$. Now, by Lemmas~\ref{lem:triangle10} and~\ref{lem:oddcycle}, $C_1$ contains an odd hole $H$, and we have $|H|\leq |C_1|<|C|$. We can now check whether $H$ is a short hole. This process must clearly terminate with a short hole, since the hole becomes progressively shorter.
\end{proof}
Note that the proof of Lemma~\ref{lem:shorthole} gives  an efficient algorithm for finding a short odd hole $H$, given any odd hole $C$. Clearly the shortest hole in $G$ is a short hole, but the converse need not be true in general, even for quasimonotone graphs. See Fig.~\ref{fig:twoholes}, which has a short 5-hole and a short 7-hole.

\begin{figure}[H]
  \centering
  \begin{tikzpicture}[
                      scale=0.9]
    \foreach \x [evaluate = \x as \y using \x*360/7+90/7] in {0,1,...,7} \node[b] (\x) at (\y:1.67) {};
    \draw (1)--(2)--(3)--(4)--(5)--(6)--(7)--(1) ;
    \foreach \x [evaluate = \x as \y using (\x-8)*72-18] in {8,9,10,11} \node[b] (\x) at (\y:1) {};
    \draw (8)--(9)--(10)--(11) ;
    \draw (11)--(5)--(8) (0)--(8) (1)--(9) (2)--(10) (3)--(11) ;
  \end{tikzpicture}
  \caption{Short odd holes of unequal size in a quasimonotone graph}\label{fig:twoholes}
\end{figure}
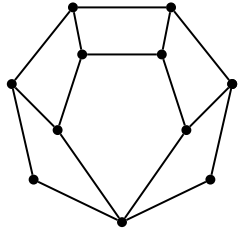

 We will also use the following simple corollary.
 \begin{corollary}\label{cor:diameter}
   If $C$ is a short odd hole in a graph $G$, $\diam(G)\geq\diam(C)=(|C|-1)/2$.\qed
 \end{corollary}

 We can make similar definitions for preholes. Thus, if $C$ is a prehole,  $G'=G[C]$, and \LR\ is the alternating bipartition of $C$, then $G'[\LR]$ contains no edge other than those of $C$. A \emph{minimal} prehole $C$ is such that $G[C]$ contains no prehole with fewer than $|C|$ vertices. Clearly, any graph which contains a prehole contains a minimal prehole.

\section{Flawless graphs containing a long hole}\label{sec:long}
\subsection{Triangles}
\begin{lemma}\label{lem:triangle10}
Let $G$ be a quasimonotone graph containing an odd hole $C$ of size at least 7. Then $G$ contains no triangle that has a vertex in $C$.
\end{lemma}

\begin{proof}
  Since $C$ is an odd hole, there are only two cases.
\begin{enumerate}[label=(\roman*)]
  \item If the triangle has two vertices on $C$, then we have the situation of Fig.~\ref{fig:dist1a} on page \pageref{fig:dist1a}, which cannot occur.
  \item If the triangle $T=(v,w,x)$ has one vertex $v$ on $C$, then consider the graph $G'$ induced by $C\cup T$, so $C$ is the shortest hole in $G'$. There are two subcases.
    \begin{enumerate}[label=(\alph*)]
    \item If neither $w$ nor $x$ has another neighbour on $v$,  we have the situation of Fig.~\ref{fig:dist1} above, which cannot occur.
    \item If $w$ has another neighbour $z$ on $C$,  but $x$ does not, we have the situation in Fig.~\ref{fig:triangle20}. This is a pre-armchair, since the edge $xv$ is not in the bipartition shown. If $x$ has another neighbour on $C$, this must be either $y$ or $z$. Neither appears in the bipartition shown in Fig.~\ref{fig:triangle20}.\qedhere
    \end{enumerate}

    \begin{figure}[H]
      \centering{
        \begin{tikzpicture}[line width=0.5pt,minimum size=3pt,inner sep=1.25pt,font=\small]
          \draw  (-1,0) node[b,label=below:$a$] (1) {}
                 ( 0,0) node[w,label=below:$z$] (2) {}
                 ( 1,0) node[b]                 (3) {}
                 ( 2,0) node[w,label=below:$v$] (4) {}
                 ( 3,0) node[b]                 (5) {}
                 ( 4,0) node[w,label=below:$y$] (6) {};
          \draw (1)--(2)--(3)--(4)--(5)--(6);
          \draw (1.5,1) node[b,label=above:$w$] (8) {}
                (2.5,1) node[w,label=above:$x$] (7) {}
                (7)--(8)--(4)
                (4)edge[dashed](7)
                (6)edge[densely dotted](7)
                (2)edge[densely dotted](7)
                (8)--(2) ;
          \draw (1,-0.7) node {$C$};
        \end{tikzpicture}
      }
      \caption{}\label{fig:triangle20}
    \end{figure}
  \end{enumerate}
\end{proof}

However, if $C$ is a 5-hole, there are quasimonotone graphs which contain a triangle with one vertex in $C$. See Fig.~\ref{fig:triangle10}.
\begin{figure}[H]
\centering{
\begin{tikzpicture}[line width=0.5pt,minimum size=3pt,inner sep=1.1pt,xscale=0.8,yscale=0.8,font=\small]
\draw (0,0) node[b] (1) {} (-1.2,1) node[b] (2) {} (1.2,1) node[b] (3) {} (-1,2.5) node[b] (4) {} (1,2.5) node[b] (5) {} (0,1.5) node[b] (6) {} ;
\draw (1)--(2)--(4)--(6)--(5)--(3)--(1) (4)--(5) (6)--(1) ;
\end{tikzpicture}
\hspace*{1in}
\begin{tikzpicture}[line width=0.5pt,minimum size=3pt,inner sep=1.25pt,xscale=0.8,yscale=0.8,font=\small]
\draw (0,0) node[b] (1) {} (-1.2,1) node[b] (2) {} (1.2,1) node[b] (3) {} (-1,2.5) node[b] (4) {} (1,2.5) node[b] (5) {} ;
\draw (1)--(2)--(4)--(5)--(3)--(1) (4)--(5)  ;
\draw (1)--(4) node[b,pos=0.6] (6) {} ;
\draw (1)--(5) node[b,pos=0.6] (7) {} ;
\draw (6)--(7) ;
\end{tikzpicture}
\hspace*{1in}
\begin{tikzpicture}[line width=0.5pt,minimum size=3pt,inner sep=1.25pt,xscale=0.8,yscale=0.8,font=\small]
\draw (0,0) node[b] (1) {} (-1.2,1) node[b] (2) {} (1.2,1) node[b] (3) {} (-1,2.5) node[b] (4) {} (1,2.5) node[b] (5) {} (0,1.1) node[b] (1') {};
\draw (1)--(2)--(4)--(5)--(3)--(1) (4)--(5)  ;
\draw (1')--(4) node[b,pos=0.6] (6) {} ;
\draw (1')--(5) node[b,pos=0.6] (7) {} ;
\draw (6)--(7) (2)--(1')--(3) (4)--(7) (5)--(6) (7)--(1)--(6) ;
\end{tikzpicture}
}
\caption{}\label{fig:triangle10}
\end{figure}

\begin{lemma}\label{lem:triangle20}
Let $G$ be a quasimonotone graph containing an odd hole $C$ of size at least $7$. Then $G$ contains no triangle which is vertex-disjoint from $C$.
\end{lemma}

\begin{proof}
Suppose $T=(vwx)$ is such a triangle, and consider the subgraph $G'$ induced by $C\cup T$. By Lemma \ref{lem:quasi40}, $C$ is a shortest hole in $G'$, so the vertices of $T$ have degree at least one and at most two in $C$. Now $T$ is the only triangle in $G'$, since any other triangle would have a vertex in $C$, contradicting Lemma~\ref{lem:triangle10}. Thus all vertices of $C$ have degree at most one in $T$, since a vertex of degree two or more would induce a triangle, using an edge of $T$. Now suppose some vertex of $T$, $v$ say, has two neighbours $a,b$ in $C$, see Fig.~\ref{fig:triangle30}. Then $G'$ has a hole with $|C'|=|C|$ through $a,v$ and $b$, with $v\in C'$. Since $|C'|\geq 7$, this  contradicts Lemma~\ref{lem:triangle10}.

Thus $T$ has at most three neighbours in $C$. Hence there are at least four vertices in $C$ which have no neighbour in $T$. Let $\nu\geq 2$ be the maximum number of consecutive vertices in $C$ with no neighbour in $T$. Suppose these are bordered by vertices $a,b\in C$, where $wa,xb\in E$. See Fig.~\ref{fig:triangle30}. Thus, if $\nu$ is odd, there is a prehole through $a,w,v,x$ and $b$, and, if $\nu$ is even, there is an even hole through $a,w,x$ and $b$. See Fig.~\ref{fig:triangle30}.\qedhere
\begin{figure}[H]
\centering{
\begin{tikzpicture}[line width=0.5pt,minimum size=3pt,inner sep=1.25pt,xscale=0.8,yscale=0.8,font=\small]
\draw (0,0) node[b,label=above:$v$] (0') {} +(150:1) node[b,label=left:$w$] (1') {} +(210:1) node[b,label=left:$x$] (2') {} (0:2.6) node {$C$} ;
\draw (-75:2) node[w] (0) {}  (-50:2) node[b] (1) {}  (-25:2) node[w,label=right:$a$] (2) {}
(0:2) node[b] (3) {}  (25:2) node[w,label=right:$b$] (4) {}  (50:2) node[b] (5) {}  (75:2) node[w] (6) {} ;
\draw (0)--(1)--(2)--(3)--(4)--(5)--(6) (0')--(1')--(2')--(0') (0')--(2) (0')--(4) ;
\end{tikzpicture}
\hspace*{1in}
\begin{tikzpicture}[line width=0.5pt,minimum size=3pt,inner sep=1.25pt,xscale=0.8,yscale=0.8,font=\small]
\draw (-1,0) node[w,label=left:$v$] (0') {} +(-30:1) node[b,label=below:$w$] (1') {} +(30:1) node[w,label=above:$x$] (2') {}  (0:2.6) node {$C$} ;
\draw (-75:2) node[w] (0) {}  (-50:2) node[b] (1) {}  (-25:2) node[w,label=right:$a$] (2) {}
(0:2) node[b] (3) {}  (25:2) node[w] (4) {}  (50:2) node[b,label=right:$b$] (5) {}  (75:2) node[w] (6) {} ;
\draw (0)--(1)--(2)--(3)--(4)--(5)--(6) (0')--(1')--(2')--(0') (1')--(2) (2')--(5) ;
\end{tikzpicture}}
\caption{}\label{fig:triangle30}
\end{figure}
\end{proof}
Again, if $|C|=5$, it is possible to have a triangle which is vertex-disjoint from $C$.
Fig.~\ref{fig:triangle10} shows an example, but note that this also contains
triangles which share a vertex or an edge  with the 5-cycle. It is not difficult to show that, if $G$ contains a 5-cycle and a vertex-disjoint triangle, then $G$ must contain a triangle which shares at least one vertex with the 5-cycle. However, we will not prove this because we make no use of it here.

\subsection{Long odd holes}

\begin{lemma}\label{lem:nooddhole}
Let $C,C'$ be odd holes in a quasimonotone graph $G$ such that $C'\cap C\neq \es$, and $|C|,|C'|\geq 7$. Let $G'=G[(C'\cup C)\sm (C'\cap C)]$, Then $G'$ has no odd hole or prehole.
\end{lemma}

\begin{proof}
Without loss of generality, we will assume $|C|\leq|C'|$.

If $G'$ has a prehole, $G$ is not quasimonotone. So suppose there is an odd hole $H$ in $G'$. Clearly $H$ must contain edges from both $C$ and $C'$. Let $P$ be the path $H\cap C$, and $P'$ the path $H\cap C'$. We choose $H$ so that $|P'|$ is minimised. See Fig.~\ref{fig:noholes1}, where $H$ is bounded by the edges $3'\,3$ and $5'\,6$, and $|P'|=2$.
\begin{figure}[H]
  \centering
  \begin{tikzpicture}[yscale=1.4,xscale=1.4,font=\small]
    \foreach \x in {2,4,6,8} {
      \node[b,label=below:$\x$]  (\x)  at (\x,0) {};
      \node[w,label=above:$\x'$] (\x') at (\x,1) {};
    };
    \foreach \x in {1,3,5,7} {
      \node[w,label=below:$\x$]  (\x)  at (\x,0) {};
      \node[b,label=above:$\x'$] (\x') at (\x,1) {};
    };
    \foreach \x in {0,9} {
      \node[i] (\x)  at (\x,0) {};
      \node[i] (\x') at (\x,1) {};
    };
    \foreach[count=\y] \x in {0,1,...,8} {
      \draw (\x)--(\y);
      \draw[dashed] (\x')--(\y');
    };
    \draw (3)--(3')--(4')--(5')--(6);
    \draw[densely dotted] (1)--(4')--(2);
    \draw (9.25,0) node {$C$} (9.25,1) node {$C'$} (4,0.5) node {$H$};
  \end{tikzpicture}
  \caption{}\label{fig:noholes1}
\end{figure}
Suppose any vertex in $v$ the interior of  path $P'$ has an edge to a vertex $w\in C$. Clearly $w\notin P$, or $H$ is not a hole. Let $v$ and $w$ be chosen so that $\dist_C(w,P)$ is minimised. Then there is either an odd hole $H'$ with $|C'\cap H'|<|P'|$, contradicting the choice of $H$, or an even hole, contradicting quasimonotonicity. For example, consider $v=4'$ in Fig.~\ref{fig:noholes1}. If $w=2$, then $H'$ is $(4',5',6,5,4,3,2,4')$ and $|C'\cap H'|=1$. If $w=1$, then $(4',5',6,5,4,3,2,1,4')$ is an even hole, and $G$ is not quasimonotone.
\begin{figure}[H]
  \centering
  \begin{tikzpicture}[yscale=1.4,xscale=1.4,font=\small]
    \foreach \x in {2,4,6,8} {
      \node[b,label=below:$\x$]  (\x)  at (\x,0) {};
      \node[w,label=above:$\x'$] (\x') at (\x,1) {};
    };
    \foreach \x in {1,3,5,7} {
      \node[w,label=below:$\x$]  (\x)  at (\x,0) {};
      \node[b,label=above:$\x'$] (\x') at (\x,1) {};
    };
    \foreach \x in {0,9} {
      \node[i] (\x)  at (\x,0) {};
      \node[i] (\x') at (\x,1) {};
    };
    \foreach[count=\y] \x in {0,1,...,8} {
      \draw (\x)--(\y);
      \draw[dashed] (\x')--(\y');
    };
    \draw (3')--(4')--(5');
    \draw[densely dashed] (3')--(3) (5')--(6);
    \draw[densely dotted] (3')--(1) (5')--(8);
    \draw (9.25,0) node {$C$} (9.25,1) node {$C'$} (4,0.5) node {$H'$} ;
  \end{tikzpicture}
  \caption{}\label{fig:noholes2}
\end{figure}
Thus we may assume that no vertex in $P'$ has an edge to $C$, excepting possible the extreme vertices of $P'$.
Now $(C\cup H)\sm(C\cap H)$ is an even hole, with length $|C|+|H|-2|C\cap H|$, unless the extreme vertices of $P'$ have edges to $C\sm P$. If not, $G'$ cannot be quasimonotone. So suppose one of the two extreme vertices $v_i$ has an edge to $w_i\in C\sm P$ ($i=1,2)$. Then, by Lemma~\ref{lem:quasi30}, $w_i$ is unique and $\dist(w_i,P)=2$ ($i=1,2)$. Now we can construct an even hole in $G'$. For example, consider $v_1=3',v_2=5'$ in Fig.~\ref{fig:noholes2}. The only possibilities for $v_iw_i$ are $3'\,1$ and/or $5'\,8$. Then we can use $3'\,1$ in place of $3'\,3$ and/or $5'\,8$ in place of $5'\,6$ to form an odd cycle $H'$, as shown in Fig.~\ref{fig:noholes2}.  Since $H'$ has no edge to a vertex in $C\sm H'$,  we have an even hole $(C\cup H')\sm(C\cap H')$. So $G$  is not quasimonotone, a contradiction.
\end{proof}

\begin{corollary}\label{cor:qm}
Let $C,C'$ be odd holes in a quasimonotone graph $G$, such that $C'\cap C\neq \es$. Let $G'=G[(C'\cup C)\sm (C'\cap C)]$. Then $G'$ is a monotone graph.
\end{corollary}
\begin{proof}
$G'$ is flawless, and has no holes or preholes from Lemma~\ref{lem:nooddhole}, so it is monotone.
\end{proof}

Note that the holes $C,C'$ in Corollary~\ref{cor:qm} can have different size. See Fig.~\ref{fig:twoholes}, where $G'$ is a \emph{ladder} (see~\cite{DyJeMu17}) with two pendant edges.
However, if we have vertex-disjoint odd holes they cannot have different lengths.

A \emph{prism} is the graph given by joining corresponding vertices in two cycles of the same length. It is an \emph{$n$-prism} if the cycles have length $n$~\cite{HlMaPi02}. See Fig.~\ref{fig:prism} for an example.

\begin{figure}[ht]
\centering{
\begin{tikzpicture}[line width=0.5pt,minimum size=1pt,xscale=1,yscale=1]
\draw (-90:1.3) node[b] (0) {} -- (-50:1.3) node[b] (1) {} -- (-10:1.3) node[b] (2) {} -- (30:1.3) node[b] (3) {}  -- (70:1.3) node[b] (4) {} -- (110:1.3) node[b] (5) {} -- (150:1.3) node[b] (6) {} -- (190:1.3) node[b] (7) {} -- (230:1.3) node[b] (8) {} -- (0) ;
\draw (-90:2) node[b] (0') {} -- (-50:2) node[b] (1') {} -- (-10:2) node[b] (2') {} -- (30:2) node[b] (3') {} -- (70:2) node[b] (4') {} -- (110:2) node[b] (5') {} -- (150:2) node[b] (6') {}  -- (190:2) node[b] (7') {} -- (230:2) node[b] (8') {} --(0') ;
\foreach \x in {0,...,8} {\draw (\x)--(\x');};
\end{tikzpicture}}
\caption{A 9-prism}\label{fig:prism}
\end{figure}
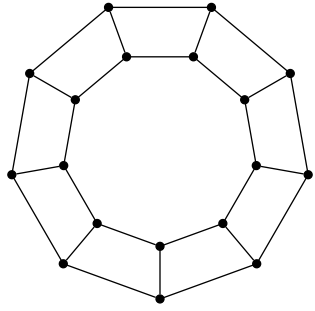

\begin{lemma}\label{lem:holes20}
Let $G$ be a quasimonotone graph containing an odd hole $C$. Then $G$ contains no vertex-disjoint hole $C'$ with $|C'|\neq |C|$. Moreover, if $|C|\geq 7$, any two vertex-disjoint holes with $|C'|=|C|$ induce a prism in $G$.
\end{lemma}

\begin{proof}
Let $G'=G[C\cup C']$, and suppose $|C'|\neq |C|$.
Let $C_1$ denote the shorter of $C,C'$, and $C_2$ the longer, so $|C_1|\geq 5$ and $|C_2|\geq 7$. Then every vertex of $C_1$ has degree in $\{1,2,3\}$ in $C_2$ and every vertex of $C_2$ has degree in $\{1,2\}$ in $C_1$. Since $|C_2|>|C_1$, there must be a vertex $v\in C_1$ with degree 2 or 3 in $C_2$, by simple counting. Let $a,b$ be two of these neighbours, such that $acb$ is a subpath of $C_2$. See Fig.~\ref{fig:crossover1}. Since $G'$ is flawless, by Lemma~\ref{lem:quasi40} every edge of $C_1$ is in a quadrangle with some edge of $C_2$, and vice versa. Thus $c$ must be adjacent to a neighbour $w$ of $v$ on $C_1$, and $G'$ must contain the configuration of Fig.~\ref{fig:crossover1}. Now $G'$ must contain either $ax$ or $dw$ or both. Otherwise $(d,a,v,w,x)$ is a chordless path of length $4$, and must be a subpath of a hole, since $d$ has some edge to $C_1$, contradicting Lemma~\ref{lem:nooddhole}. Note that $dx\notin E$, since otherwise $(d,x,a)$ or $(d,x,w)$ would be a triangle, contradicting Lemma~\ref{lem:triangle10}.  Whether $ax$ or $dv$ is an edge, this argument can be repeated for the vertices to the left of $d$ and $x$, or to the right of $c$ and $v$. Thus there cannot be any more edges than those indicated in Fig.~\ref{fig:crossover1}, since every vertex of $C_2$ has degree at most 2 in $C_1$.

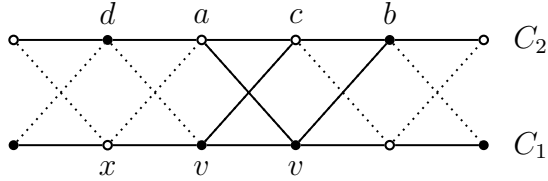
\begin{figure}[H]
  \centering
  \begin{tikzpicture}[xscale=1.25,yscale=1.4]
    \foreach[count=\x] \c/\n\d/\m in {w//b/,b/d/w/x,w/a/b/v,w/c/b/v,b/b/w/,w//b/}{
      \node[\c,label=above:$\n$] (t\x) at (\x,1) {}; 
      \node[\d,label=below:$\m$] (b\x) at (\x,0) {}; 
    };
    \foreach[count=\x] \y in {2,3,...,6}{
      \draw         (t\x)--(t\y)  (b\x)--(b\y);
      \draw[dotted] (t\x)--(b\y)  (b\x)--(t\y);
    };
    \draw (t3)--(b4)--(t5)  (b3)--(t4);
    \draw (6.5,0) node {$C_1$} (6.5,1) node {$C_2$};
  \end{tikzpicture}
  \caption{Crossover}
  \label{fig:crossover1}
\end{figure}

We give $C_1$ the alternating bipartition except, since $|C_1|$ is odd, both $w,v\in L$. Similarly we give $C_1$ the alternating bipartition with both $a,c\in R$. This is the alternating partition of an even cycle $C_1\sm wv$, $wc$, $C_2\sm ac$, $av$. We call $av$, $wc$ a \emph{crossover}, and a prehole formed in this way a \emph{crossover prehole}. Now we observe that all possible edges other than $av,cw$ have both endpoints in $L$ or both endpoints in $R$. Thus $G'$ is a prehole.

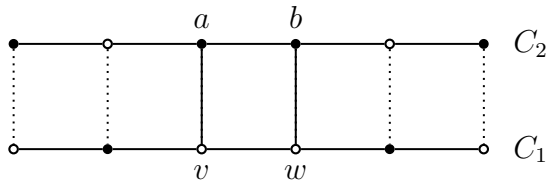
\begin{figure}[H]
  \centering
  \begin{tikzpicture}[xscale=1.25,yscale=1.4]
    \foreach[count=\x] \c/\n\d/\m in {b//w/,w//b/,b/a/w/v,b/b/w/w, w//b/,b//w/}{
      \node[\c,label=above:$\n$] (t\x) at (\x,1) {}; 
      \node[\d,label=below:$\m$] (b\x) at (\x,0) {}; 
      \draw[dotted] (t\x)--(b\x);
    };
    \foreach[count=\x] \y in {2,3,...,6} \draw (t\x)--(t\y)  (b\x)--(b\y);
    \draw (t3)--(b3)  (t4)--(b4);
    \draw (6.5,0) node {$C_1$} (6.5,1) node {$C_2$};
  \end{tikzpicture}
  \caption{Matching between $C_1$ and $C_2$}
  \label{fig:crossover2}
\end{figure}

Thus all vertices in $C_1$ must have only one edge to $C_2$. (See Fig.~\ref{fig:crossover2}.) Since these edges form a matching between $C_1$ and $C_2$, we must have $|C_1|=|C_2|$, and $G[C\cup C']$ must be a prism.
\end{proof}

\section{Preholes in flawless graphs}\label{sec:prehole}

\begin{lemma}\label{lem:preholetypes}
 If $G\in\noflaw$ and has an odd hole of size $\ell\geq 7$, any minimal prehole $C$ in $G$ is either an even hole or
 \begin{enumerate*}[label=(\alph*)]
 \item two odd holes intersecting in an edge \label{prehole:casea} or
 \item two disjoint odd holes connected by a quadrangle. \label{prehole:caseb}
 \end{enumerate*}
 See Fig.~\ref{fig:long_h}.
\end{lemma}
\begin{figure}[H]
\centering{
\begin{tikzpicture}[line width=0.75pt,minimum size=2pt,inner sep=0.9pt,xscale=1.25,yscale=1.4]
\draw[rotate=25.5] (0:1.1) coordinate (0) -- (51.4:1.1) node[w] (1) {} -- (102.9:1.1) node[b] (2) {} -- (154.3:1.1) node[w] (3) {} -- (205.7:1.1) node[b] (4) {} -- (257.1:1.1) node[w] (5) {} -- (308.6:1.1) coordinate (6) edge[dashed]  (0);
\draw[xshift=1.65cm] (0:0.8) node[b] (0') {} -- (72:0.8) node[w] (1') {} -- (144:0.8) node[b] (2') {}  (216:0.8) node[b] (3') {} -- (288:0.8) node[w] (4') {} -- (0')  ;
\draw (0)--(2') (6)--(3') ;
\draw (0,0) node{$C_1$} (1.65,0) node{$C_2$} ;
\draw (1,-1.5) node {\ref{prehole:casea}} ;
\end{tikzpicture}
\hspace*{2cm}
\begin{tikzpicture}[line width=0.75pt,minimum size=2pt,inner sep=0.9pt,xscale=1.25,yscale=1.4]
\draw[rotate=25.5] (0:1.1) node[w] (0) {} -- (51.4:1.1) node[b] (1) {} -- (102.9:1.1) node[w] (2) {} -- (154.3:1.1) node[b] (3) {} -- (205.7:1.1) node[w] (4) {} -- (257.1:1.1) node[b] (5) {} -- (308.6:1.1) node[w] (6) {} edge[dashed]  (0);
\draw[xshift=2.25cm] (0:0.8) node[b] (0') {} -- (72:0.8) node[w] (1') {} -- (144:0.8) node[b] (2') {}  (216:0.8) node[b] (3') {} -- (288:0.8) node[w] (4') {} -- (0') (2') edge[dashed] (3') ;
\draw (0)--(2') (6)--(3') ;
\draw (0,0) node{$C_1$} (2.25,0) node{$C_2$} ;
\draw (1.3,-1.5) node {\ref{prehole:caseb}} ;
\end{tikzpicture}}
\caption{Preholes with odd holes $C_1$, $C_2$.}\label{fig:long_h}
\end{figure}
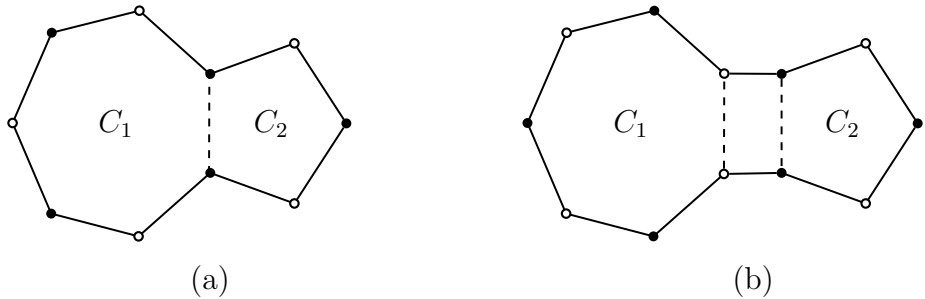
\begin{proof}
We may assume that $G$ has no triangles, from Lemmas~\ref{lem:triangle10} and~\ref{lem:triangle20}.  Clearly $C$ has at least one even chord $e$ which divides it into two smaller odd cycles $C_1$ and $C_2$. By Lemma~\ref{lem:oddcycle},   $C_1$ and $C_2$ contain odd holes $C'_1,C'_2$. If $C'_1=C_1$, $C'_2=C_2$, then we are in case  \ref{prehole:casea}. Otherwise, we can use Lemma~\ref{lem:oddcycle} to arrive at two odd holes  $A$ in $C_1$ and $B$ in $C_2$. Now $A$ and $B$ can have only one chord of $C$ in their boundaries, since $C$ has only even chords. Thus the structure of $C$ is either as shown in Fig.~\ref{fig:long_h}, or as shown in Fig.~\ref{fig:possible0}. We must show that $C$ cannot be a minimal prehole in the latter case.
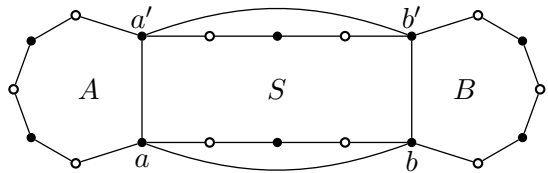
\begin{figure}[H]
\begin{center}
\begin{tikzpicture}[line width=0.5pt,inner sep=5pt,font=\small]
\foreach \x in {100,180,260} {\draw (\x:1) node[w] (A\x) {};};
\foreach \x in {45,140,220,315} {\draw (\x:1) node[b] (A\x) {};};
\draw (A45)--(A100)--(A140)--(A180)--(A220)--(A260)--(A315) ;
\draw (0,0) node[empty] {$A$} ;
\begin{scope}[xshift=5cm,rotate=180]
\foreach \x in {100,180,260} {\draw (\x:1) node[w] (B\x) {};};
\foreach \x in {45,140,220,315} {\draw (\x:1) node[b] (B\x) {};};
\draw (B45)--(B100)--(B140)--(B180)--(B220)--(B260)--(B315) ;
\draw (0,0) node[empty] {$B$} ;
\end{scope}
\draw (2.5cm,0) node[empty] {$S$} ;
\draw (A45)--(A315) (B45)--(B315) (A45)--(B315) (B45)--(A315) ;
\draw (A45)++(.9,0) node[w] {} ++(.9,0) node[b] {}  ++(.9,0) node[w] {}  ;
\draw (A315)++(.9,0) node[w] {} ++(.9,0) node[b] (c) {}  ++(.9,0) node[w] {}  ;
\draw (A45)++(0,.25) node[empty] {$a'$} (A315)++(0,-.25) node[empty] {$a$} ;
\draw (B45)++(0,-.25) node[empty] {$b$} (B315)++(0,.25) node[empty] {$b'$} ;
\draw (A45)edge[bend left=20](B315) (B45)edge[bend left=20](A315) ;
\end{tikzpicture}
\end{center}
\caption{A possible prehole}\label{fig:possible0}
\end{figure}
In Fig.~\ref{fig:possible0}, $A$ and $B$ are odd holes and $S$ joins them, and is not a single edge or quadrangle. Thus $A$ and $B$ both have size $\ell$, and $A, B$ must induce a prism. Otherwise $G[A\cup B$] contains a smaller crossover prehole, by Lemma~\ref{lem:holes20}. Thus, in particular, the edge $a'b'$ must be present, and $b'$ has no other edge to $A$. Also, we must have $a,b,a',b'\in L$, or we would have a smaller case~\ref{prehole:caseb} prehole using $ab,a'b'$.

\begin{figure}[ht]
\begin{center}
\begin{tikzpicture}[line width=0.5pt,inner sep=5pt,font=\small]
\foreach \x in {100,180,260} {\draw (\x:1) node[w] (A\x) {};};
\foreach \x in {45,140,220,315} {\draw (\x:1) node[b] (A\x) {};};
\draw (A45)--(A100)--(A140)--(A180)--(A220)--(A260)--(A315) ;
\draw (0,0) node[empty] {$A$} ;
\begin{scope}[xshift=5cm,rotate=180]
\foreach \x in {100,180,260} {\draw (\x:1) node[w] (B\x) {};};
\foreach \x in {45,140,220,315} {\draw (\x:1) node[b] (B\x) {};};
\draw (B45)--(B100)--(B140)--(B180)--(B220)--(B260)--(B315) ;
\draw (0,0) node[empty] {$B$} ;
\end{scope}
\draw (2.5cm,0) node[empty] {$S$} ;
\draw (A45)--(A315) (B45)--(B315) (A45)--(B315) (B45)--(A315) ;
\draw (A45)++(.9,0) node[w] {} ++(.9,0) node[b] {}  ++(.9,0) node[w] {}  ;
\draw (A315)++(.9,0) node[w] {} ++(.9,0) node[b] (c) {}  ++(.9,0) node[w] {}  ;
\draw (A45)++(0,.25) node[empty] {$a'$} (A315)++(0,-.25) node[empty] {$a$} (c)++(0,-.2) node[empty] {$c$};
\draw (B45)++(0,-.25) node[empty] {$b$} (B315)++(0,.25) node[empty] {$b'$} (B315)--(c) ;
\draw (A45)edge[bend left=20](B315) (B45)edge[bend left=20](A315) ;
\end{tikzpicture}
\end{center}
\caption{A possible prehole}\label{fig:possible1}
\end{figure}
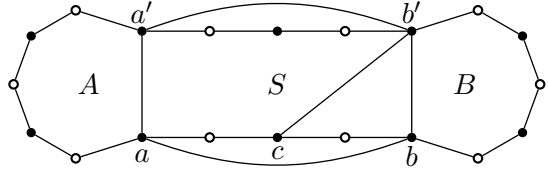

Let $c$ be the vertex nearest $a$ on the path from $a$ to $b$ such that $b'c$ is an even chord of $C$, as shown in Fig.~\ref{fig:possible1}. Note that $c$ is well determined, since $c=b$ is possible.

Consider the cycle $C'$ in $G$ shown in Fig.~\ref{fig:possible2}. It is easy to see that $C'$ is a prehole, where the certifying bipartition simply moves $b'$ from $L$ to $R$ from that of Fig.~\ref{fig:possible1}, as shown.
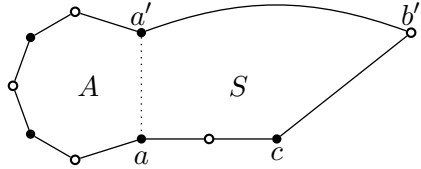
\begin{figure}[H]
\begin{center}
\begin{tikzpicture}[line width=0.5pt,inner sep=5pt,font=\small]
\foreach \x in {100,180,260} {\draw (\x:1) node[w] (A\x) {};};
\foreach \x in {45,140,220,315} {\draw (\x:1) node[b] (A\x) {};};
\draw (A45)--(A100)--(A140)--(A180)--(A220)--(A260)--(A315) ;
\draw (0,0) node[empty] {$A$} ;
\begin{scope}[xshift=5cm,rotate=180]
\foreach \x in {315} {\draw (\x:1) node[w] (B\x) {};};
\end{scope}
\draw (2cm,0) node[empty] {$S$} ;
\draw[dotted] (A45)--(A315)  ;
\draw (A315)++(.9,0) node[w] {} ++(.9,0) node[b] (c) {}  (A315)--(c);
\draw (A45)++(0,.25) node[empty] {$a'$} (A315)++(0,-.25) node[empty] {$a$} (c)++(0,-.2) node[empty] {$c$};
\draw (B315)++(0,.25) node[empty] {$b'$} (B315)--(c) ;
\draw (A45)edge[bend left=20](B315) ;
\end{tikzpicture}
\end{center}
\caption{Smaller prehole}\label{fig:possible2}
\end{figure}
Thus $C'$ is a prehole with size $|C'| < |C|$, so $C$ was not a minimal prehole. Thus any minimal prehole involving two vertex-disjoint holes of size at least 7 must be as in case~\ref{prehole:caseb}.

Finally suppose $C'_1,C'_2$ are edge- but not vertex-disjoint, so they intersect in a single vertex~$c$. Since $|C'_1 \cup C'_2|$ is odd, there must be a vertex $v\in C\sm(C'_1\cup C'_2)$. Since $C$ is minimal, and $v$ is adjacent to both $C'_1,C'_2$, by Lemma~\ref{lem:quasi40}, $v$ must be unique. Since $G$ has no even hole, $c,v$ are the opposite vertices of a quadrangle $(c,x,v,w)$, with $v\in L$, $c,x,w\in R$. Thus we have the configuration shown in Fig.~\ref{fig:commonvertex}.

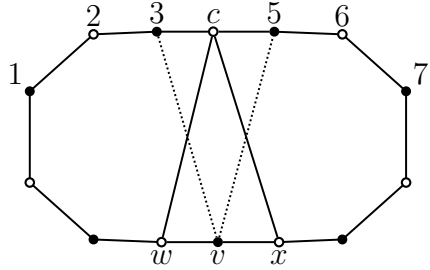
\begin{figure}[hb]
\centering{
\begin{tikzpicture}[line width=0.75pt,minimum size=2pt,inner sep=0.9pt,xscale=1.25,yscale=1.4]
\draw (1.05,1) node[w,label=above:$c$] (c) {} (1.1,-1) node[b,label=below:$v$] (v) {} ;
\draw (c) -- (0.45,1) node[b,label=above:$3$] (3) {} -- (102.9:1) node[w,label=above:$2$] (2) {} -- (154.3:1) node[b,label=above left:$1$] (1) {} -- (205.7:1) node[w] (4) {} -- (257.1:1) node[b] (5') {} -- (0.5,-1) node[w,label=below:$w$] (w) {} ;
\draw[xshift=2.2cm,rotate=180] (v) -- (0.45,1) node[w,label=below:$x$] (x) {} -- (102.9:1) node[b] (2') {} -- (154.3:1) node[w] (3') {} -- (205.7:1) node[b,label=above right:$7$] (7) {} -- (257.1:1) node[w,label=above:$6$] (6) {} -- (0.5,-1) node[b,label=above:$5$] (5) {} ;
\draw (c)--(5) (w)--(v) (w)--(c)--(x) ;
\draw[densely dotted] (3)--(v)--(5)  ;
\end{tikzpicture}}
\caption{Preholes with odd holes $C_1$, $C_2$.}\label{fig:commonvertex}
\end{figure}

The vertices such that $\dist_C(c)\leq 3$ form a prechordless path $P=(1,2,3,c,5,6,7)$, as shown in Fig.~\ref{fig:commonvertex}. Since $|C'_1|,|C'_2|\geq 5$, $\dist_C(v,P)\geq 2$. Suppose $v$ is not adjacent to any vertex of $P$. Then $\dist(v,P)=\dist(v,c)=2$, so $vc\in E$ by Lemma~\ref{lem:quasi10}. This is a contradiction, since $C$ is a prehole. Thus $v$ has an edge to $P$, thus to $1,3,5$ or $7$. By symmetry, suppose either $v1\in E$ or $v3\in E$. If $v1\in E$, $v3\notin E$, the cycle $(v,1,2,3,c,w)$ is a 6-hole in $G$, contradicting the minimality of $C$. Thus $v3\in E$. But now the quadrangle $(v,3,c,x)$ separates $C$ into two vertex-disjoint odd cycles $H_1=(v,3,2,1, \ldots, w)$ and $H_2=(c,5,6,7,\ldots x)$. Thus $C$ contains two vertex-disjoint odd holes $H'_1\subseteq H_1$, $H'_2\subseteq H_2$, by Lemma~\ref{lem:oddcycle}. By the minimality of $C$, we must have $H'_1=H_1$, $H'_2=H_2$, and we are in case~\ref{prehole:caseb} again.
\end{proof}
Thus, if $G$ contains an odd hole of size at least 7, minimal preholes have only two types, case~\ref{prehole:casea} and case~\ref{prehole:caseb}. From Lemma~\ref{lem:holes20}, case~\ref{prehole:caseb} are crossover preholes. Examples are shown in Fig.~\ref{fig:long_f}.
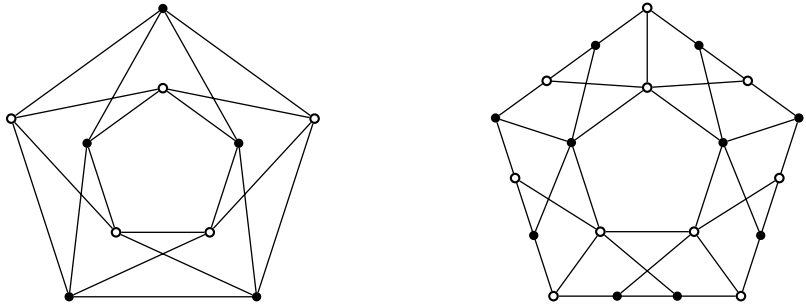
\begin{figure}[H]
\centering{
\begin{tikzpicture}[line width=0.5pt,minimum size=2pt,inner sep=0.9pt,scale=0.53]
\draw (18:2) node[b] (0) {}  (90:2) node[w] (1) {}  (162:2) node[b] (2) {}
(234:2) node[w] (3) {}  (306:2) node[w] (4) {} ;
\draw (18:4) node[w] (0') {}  (90:4) node[b] (1') {}  (162:4) node[w] (2') {}
(234:4) node[b] (3') {}  (306:4) node[b] (4') {} ;
\draw (0)--(1)--(2)--(3)--(4)--(0) ;
\draw (4)--(3') (3)--(4') (4')--(0) (4)--(0') (0')--(1) (0)--(1') (1')--(2) (1)--(2') (2')--(3) (2)--(3')  ;
\draw (0')--(1')--(2')--(3')--(4')--(0') ;
\end{tikzpicture}
\hspace*{2cm}
\begin{tikzpicture}[line width=0.5pt,minimum size=2pt,inner sep=0.9pt,scale=0.53]
\draw (18:2) node[b] (0) {}  (90:2) node[w] (1) {}  (162:2) node[b] (2) {}
(234:2) node[w] (3) {}  (306:2) node[w] (4) {} ;
\draw (18:4) node[b] (0') {}  (90:4) node[w] (1') {}  (162:4) node[b] (2') {}
(234:4) node[w] (3') {}  (306:4) node[w] (4') {} ;
\draw (0)--(0') (1)--(1') (2)--(2') (3)--(3') (4)--(4') ;
\draw (0)--(1)--(2)--(3)--(4)--(0) ;
\draw (0')--(1') node[w, pos=0.33] (00) {} node[b, pos=0.67] (01) {} ;
\draw (1')--(2') node[b, pos=0.33] (10) {} node[w, pos=0.67] (11) {} ;
\draw (2')--(3') node[w, pos=0.33] (20) {} node[b, pos=0.67] (21) {} ;
\draw (3')--(4') node[b, pos=0.33] (30) {} node[b, pos=0.67] (31) {} ;
\draw (4')--(0') node[b, pos=0.33] (40) {} node[w, pos=0.67] (41) {} ;
\draw (00)--(1) (01)--(0) (10)--(2) (11)--(1) (20)--(3) (21)--(2) (30)--(4) (31)--(3) (40)--(0) (41)--(4) ;
\end{tikzpicture}
}\vspace{5mm}
\caption{Flawless crossover preholes.}\label{fig:long_f}
\end{figure}

So let us consider the case \ref{prehole:casea} preholes. We will call these \emph{M\"obius preholes}, since we will show that such a prehole must be a \emph{M\"obius ladder}~\cite{GuyHar67,HlMaPi02}. See Fig.~\ref{fig:long_i1} for two different drawings of a M\"obius ladder. As the name suggests, this is a ladder
with a crossover.

\begin{lemma}\label{lem:prehole20}
  If $C$ is a M\"obius prehole in a flawless graph $G$, then $C$ is a \emph{M\"obius ladder}.
\end{lemma}
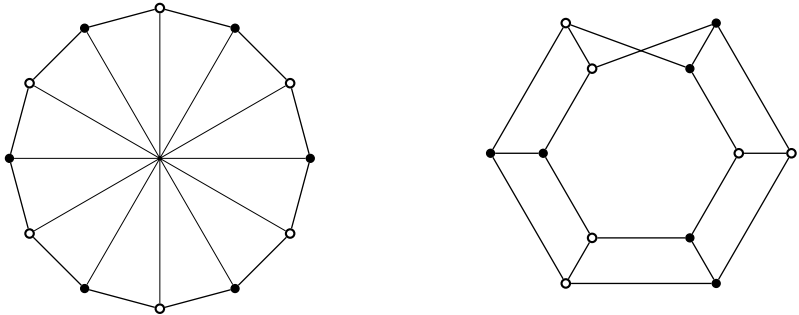
\begin{figure}[H]
\centering{
\begin{tikzpicture}[line width=0.5pt,minimum size=1.5pt,xscale=1,yscale=1]
\draw (-90:2) node[w] (0) {}  (-30:2) node[w] (2) {}  (30:2) node[w] (4) {} (90:2) node[w] (6) {}  (150:2) node[w] (8) {}  (210:2) node[w] (10) {}   ;
\draw (-60:2) node[b] (1) {}  (0:2) node[b] (3) {}  (60:2) node[b] (5) {}   (120:2) node[b] (7) {} (180:2) node[b] (9) {} (240:2) node[b] (11) {}  ;
\draw (0)--(1)--(2)--(3)--(4)--(5)--(6)--(7)--(8)--(9)--(10)--(11)--(0) ;
\draw[line width=0.25pt] (0)--(6) (2)--(8) (4)--(10) (1)--(7) (3)--(9) (5)--(11)  ;
\end{tikzpicture}
\hspace*{2cm}
\begin{tikzpicture}[line width=0.5pt,minimum size=1.5pt,xscale=1,yscale=1,rotate=30]
\draw (-90:1.3) node[b] (0) {}  (-30:2) node[w] (2) {}  (30:1.3) node[b] (4) {} (90:2) node[w] (6) {}  (150:1.3) node[b] (8) {}  (210:2) node[w] (10) {}   ;
\draw (-90:2) node[b] (1) {}  (-30:1.3) node[w] (3) {}  (30:2) node[b] (5) {}   (90:1.3) node[w] (7) {} (150:2) node[b] (9) {} (210:1.3) node[w] (11) {}  ;
\draw (1)--(2)--(5) (6)-- (9)--(10)--(1) (0)--(3)--(4) (7)--(8)--(11)--(0)
(0)--(1) (2)--(3) (4)--(5) (6)--(7) (8)--(9) (10)--(11) (5)--(7) (4)--(6) ;
\draw (-120:2) node {} ;  ;
\end{tikzpicture}
}
\caption{A M\"obius ladder}\label{fig:long_i1}
\end{figure}
\begin{proof}
Let $C$ have the alternating bipartition such that $a,b\in R$ and $ab$ divides $C$ into odd holes $C_1,C_2$, with $|C_1|\leq |C_2|$. Note, since $|C_1|,|C_2|\geq 5$, that $|C|\geq 8$. See~Fig.~\ref{fig:long_i2}. If $a$ is incident to more than one chord, then one of $C_1,C_2$ is not a hole, so $C$ is not a case~\ref{prehole:casea} prehole. So $ab$ is the only chord incident to $a$ and, similarly, $b$.

Let $v\in R\cap C_1$ have distance 2 from $a$. Then $v$ must have an edge to some $w\in C_2$. Since $C$ is a prehole, we must have $w\in R$. Then $(a,x,v,w,b,a)$ is an even cycle, so must have a chord. The only possible chords are from $x$ to the vertices on $C_2\cap L$ between $w$ and $b$. Thus, in particular, $xy$ must be an edge, where $y$ is adjacent to $b$ on $C_2$. Now $xy$ divides $C$ into odd holes $C'_1,C'_2$, so we can repeat the argument to show that $vw$ is an edge, where $w$ is adjacent to $y$ on $C_2$. We can iterate the argument for all vertices between $a$ and $b$ on $C$. If $|C_1|<|C_2|$, we will be left with an even hole on $|C_2|-|C_1|+4$ vertices. So we must have $|C_1|=|C_2|$, and all edges between diametral pairs on $C$, as in Fig.~\ref{fig:long_i1}.\end{proof}

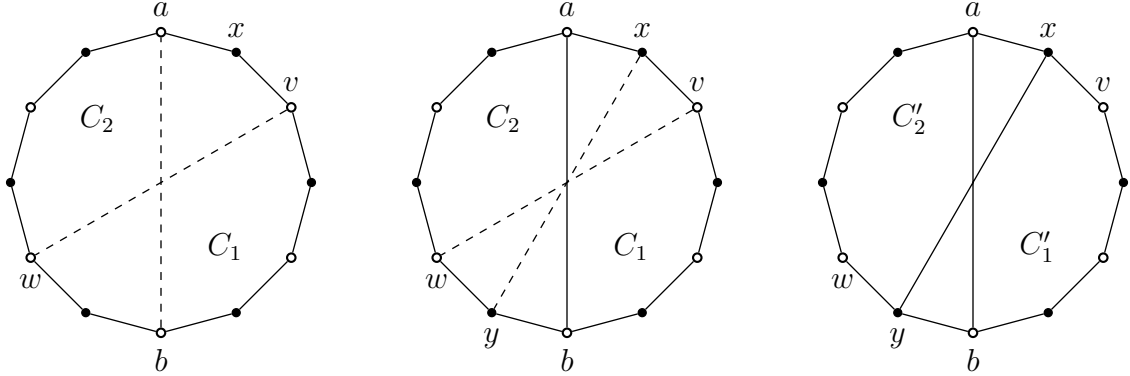
\begin{figure}[H]
\centering{
\begin{tikzpicture}[line width=0.5pt,minimum size=1.5pt,xscale=1,yscale=1]
\draw (-90:2) node[w,label=below:$b$] (0) {}  (-30:2) node[w] (2) {}  (30:2) node[w,label=above:$v$] (4) {} (90:2) node[w,label=above:$a$] (6) {}  (150:2) node[w] (8) {}  (210:2) node[w,label=below:$w$] (10) {}   ;
\draw (-60:2) node[b] (1) {}  (0:2) node[b] (3) {}  (60:2) node[b,label=above:$x$] (5) {}   (120:2) node[b] (7) {} (180:2) node[b] (9) {} (240:2) node[b] (11) {}  ;
\draw (0)--(1)--(2)--(3)--(4)--(5)--(6)--(7)--(8)--(9)--(10)--(11)--(0) ;
\draw (0)edge[dashed](6) (4)edge[dashed](10) (135:1.2) node {$C_2$} (-45:1.2) node {$C_1$};
\end{tikzpicture}
\hspace*{1cm}
\begin{tikzpicture}[line width=0.5pt,minimum size=1.5pt,xscale=1,yscale=1]
\draw (-90:2) node[w,label=below:$b$] (0) {}  (-30:2) node[w] (2) {}  (30:2) node[w,label=above:$v$] (4) {} (90:2) node[w,label=above:$a$] (6) {}  (150:2) node[w] (8) {}  (210:2) node[w,label=below:$w$] (10) {}   ;
\draw (-60:2) node[b] (1) {}  (0:2) node[b] (3) {}  (60:2) node[b,label=above:$x$] (5) {}   (120:2) node[b] (7) {} (180:2) node[b] (9) {} (240:2) node[b,label=below:$y$] (11) {}  ;
\draw (0)--(1)--(2)--(3)--(4)--(5)--(6)--(7)--(8)--(9)--(10)--(11)--(0) ;
\draw (0)--(6) (4)edge[dashed](10) (5)edge[dashed](11) (135:1.2) node {$C_2$} (-45:1.2) node {$C_1$} ;
\end{tikzpicture}
\hspace*{1cm}
\begin{tikzpicture}[line width=0.5pt,minimum size=1.5pt,xscale=1,yscale=1]
\draw (-90:2) node[w,label=below:$b$] (0) {}  (-30:2) node[w] (2) {}  (30:2) node[w,label=above:$v$] (4) {} (90:2) node[w,label=above:$a$] (6) {}  (150:2) node[w] (8) {}  (210:2) node[w,label=below:$w$] (10) {}   ;
\draw (-60:2) node[b] (1) {}  (0:2) node[b] (3) {}  (60:2) node[b,label=above:$x$] (5) {}   (120:2) node[b] (7) {} (180:2) node[b] (9) {} (240:2) node[b,label=below:$y$] (11) {}  ;
\draw (0)--(1)--(2)--(3)--(4)--(5)--(6)--(7)--(8)--(9)--(10)--(11)--(0) ;
\draw (0)--(6) (5)--(11) (135:1.2) node {$C'_2$} (-45:1.2) node {$C'_1$} ;
\end{tikzpicture}}
\caption{Implied diameter}\label{fig:long_i2}
\end{figure}

\subsection{Splitting}\label{ss:splitting}
Let $G$ be a flawless graph with a hole $C$ of length $|C|\geq 6$.  If $|C|$ is even, we conclude $G\notin\qmon$, so $|C|\geq 7$ is odd. Thus $G$ does not contain a triangle, from Lemmas~\ref{lem:triangle10} and~\ref{lem:triangle20}. We will assume that this has been tested. We will now show that $G$ must have the annular structure referred to in section~\ref{ss:stucture}, rather like a monotone graph with its ends identified.

Now suppose $G$ has a short odd hole $C$ with $C\geq 7$, determined by the procedure of Lemma~\ref{lem:shorthole}. Thus, by Corollary~\ref{cor:diameter}, $\diam(G)\geq \tfrac12(|C|-1)\geq 3$. Choose any $v\in C$, and consider the graph $\cG_v=G[V\sm N[v]]$. Then $\cG_v$ contains no holes, since any hole $H$ in $\cG_v$ must be a hole in $G$. But any hole $H$ in $G$ either contains $v$, or has a vertex $w$ adjacent to $v$, by Lemma~\ref{lem:quasi20}. Since $v,w\notin \cG_v$, $H\nsubseteq \cG_v$. Neither can $\cG_v$ contain a prehole, since any prehole must contain two holes. Thus $\cG_v$ is flawless and contains no holes or preholes, so is a monotone graph. Now $\diam(G)$ is at least $\diam(C)=(|C|-1)/2\geq 3$. Thus there exists a $w\in C$ such that $N(v)\cap N(w)=\es$.

By definition~\cite{DyJeMu17}, a graph is monotone if and only if it is bipartite, and its biadjacency matrix has an ordering of rows ($L$) and columns ($R$) so that it has the ``staircase'' structure indicated in Fig.~\ref{fig:mono}. This is symmetrical with respect to rows and columns~\cite{DyJeMu17}. That is, the transpose of the biadjacency matrix represents the same monotone graph with $L$ and $R$ interchanged.

A chain graph is a monotone graph in which each vertex in $L$ (resp.~$R$) has an edge to the first vertex in $R$ (resp.~$L$), in this ordering. Thus the biadjacency matrix has the form indicated in Fig.~\ref{fig:chain}. (See~\cite{DyJeMu17} for details.)
\begin{figure}[hb]
\centering{%
\begin{tikzpicture}[scale=0.35]
\draw[densely dotted] (0,5) coordinate (1){} --(8,5)--(8,0) coordinate (2) --(0,0) coordinate (0)--cycle
(3,3.75) coordinate (3) (6,1) coordinate (4) ;
\draw[color=gray!70!black,fill=gray!20!white] (0)--(1)--(3)--(4)--(2)--cycle  (8/3,5/3) node {$\cC$};
\draw (0,2.5) node[empty,left] {$L$} (4,5) node[empty,above] {$R$} ;
\end{tikzpicture}}\caption{Chain graph structure}\label{fig:chain}
\end{figure}
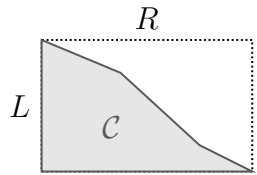
In the monotone representation, it is an easy observation that the graph has a decomposition into chain graphs, as indicated in Fig.~\ref{fig:mono}, where $L$ is partitioned in $D_1,D_3,\ldots$ and $R$ into $D_2,D_4,\ldots$. This partition was given previously by Brandst\"{a}dt and Lozin~\cite{BraLoz03}, though their proof is based on the representation of a monotone graph as a bipartite permutation graph, and is nontrivial.
\begin{figure}[ht]
\centering{%
\begin{tikzpicture}[scale=0.4,font=\small]
\draw[densely dotted] (0,11) coordinate (0){} --(11,11)--(11,0) coordinate (8) --(0,0)--cycle ;
\draw[densely dotted] (0,11) coordinate (1) {}
(4,11) coordinate (2) {} (4,8) coordinate (3) {}
(8,8)coordinate (4) {} (8,3)coordinate (5) {}
(11,3)coordinate (6) {} (11,0)coordinate (7) {} ;
\draw[densely dotted] (1)--(2)--(3)--(4)--(5)--(6)--(7) ;
\draw[line width=0.5pt,fill=gray!15!white] (0)--(1)--(2)--cycle
(1)--(3)--(2)--cycle (2)--(4)--(3)--cycle (4)--(3)--(5)--cycle
(4)--(6)--(5)--cycle (5)--(6)--(7)--cycle (8)--(6)--(7)--cycle ;
\draw[color=gray!70!black] (2.5,10) node[empty] {$\cC_1$}
(5.5,9) node[empty] {$\cC_2$}  (6.5,6.5) node[empty] {$\cC_3$}
(9,4.5) node[empty] {$\cC_4$}  (10,2) node[empty] {$\cC_5$}  ;
\draw (2,11) node[empty,above] {$D_2$} (4,11)  node[empty,above] {$|$} (6,11) node[empty,above] {$D_4$}
(8,11)  node[empty,above] {$|$}  (9.5,11) node[empty,above] {$D_6$} ;
\draw (0,1.5) node[empty,left] {$D_5$} (0,3) node[empty,left] {|}
(0,5.5) node[empty,left] {$D_3$} (0,8) node[empty,left] {|}
(0,9.5) node[empty,left] {$D_1$} ;
\end{tikzpicture}}\caption{Decomposition of a monotone graph}\label{fig:mono}
\end{figure}
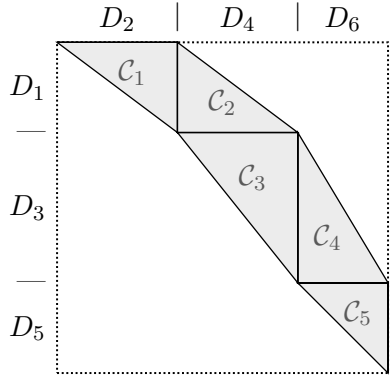
Now we have shown that $\cG_v$ is monotone, and there is a $w$ such that $N(v)\cap N(w)=\es$. Thus $N(w)$ and its neighbours induce a monotone subgraph $\cN_w$ of $G$, as indicated in Fig.~\ref{fig:neighbours}. It is easy to see that the vertex set of $\cN_w$ is $\{x\in L\cup R:\dist(w,x)\leq 2\}$. Clearly $\cN_w$ is the union of two chain graphs $C_w,C'_w$, with $C_w$ lying in the rows below and including $w$, and $C'_w$ in the rows above.
\begin{figure}[H]
\centering{%
\begin{tikzpicture}[scale=0.3,font=\small]
\draw[densely dotted] (0,11) coordinate (0){} --(11,11)--(11,0) coordinate (8) --(0,0)--cycle ;
\draw[densely dotted] (0,11) coordinate (1) {}
(4,11) coordinate (2) {} (4,8) coordinate (3) {}
(8,8)coordinate (4) {} (8,3)coordinate (5) {}
(11,3)coordinate (6) {} (11,0)coordinate (7) {} ;
\draw[densely dotted] (1)--(2)--(3)--(4)--(5) ;
\draw[line width=0.5pt,fill=gray!20!white] (2)--(4)--(3)--cycle (4)--(3)--(5)--cycle ;
\draw[line width=0.5pt,fill=gray!10!white] (1)--(3)--(2)--cycle  (5)--(4)--(6)--(8)--cycle  ;
\draw[color=gray!70!black]
(5.5,9) node[empty] {$\cC'_w$}  (6.5,6.5) node[empty] {$\cC_w$};
\draw (4,11)  node[empty,above] {$|$} (6,11) node[empty,above] {$N(w)$} (8,11)  node[empty,above] {$|$}  ;
\draw[dotted] (0,11) node[empty,left] {$v$} (0,8) node[empty,left] {$w$}--(3);
\end{tikzpicture}}\caption{Neighbourhood of $w$ in $\cG_v$}\label{fig:neighbours}
\end{figure}
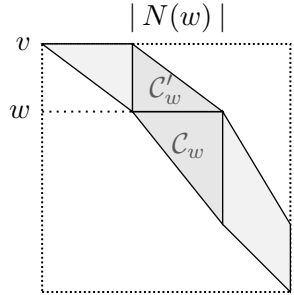
We can determine this split using the monotone representation of $\cG_v$, with the algorithm of~\cite{SpBrSt87}. Then we can construct a representation of the adjacency matrix $A(G)$ of $G$ as indicated in the first diagram in Fig.~\ref{fig:quasimono}, where $D_2=N(w)$, $\cC_1=\cC_w$ (transposed), and $\cC_7=\cC'_w$. The chain graphs $\cC_2,\ldots,\cC_6$ are a decomposition of the monotone graph $\cG_w$. Note that the ordering of the chain graphs in the decomposition is circular, and the second diagram in Fig.~\ref{fig:quasimono} gives an equivalent representation to the first, where $\cC_1$ (transposed) is moved from the first to the last position.
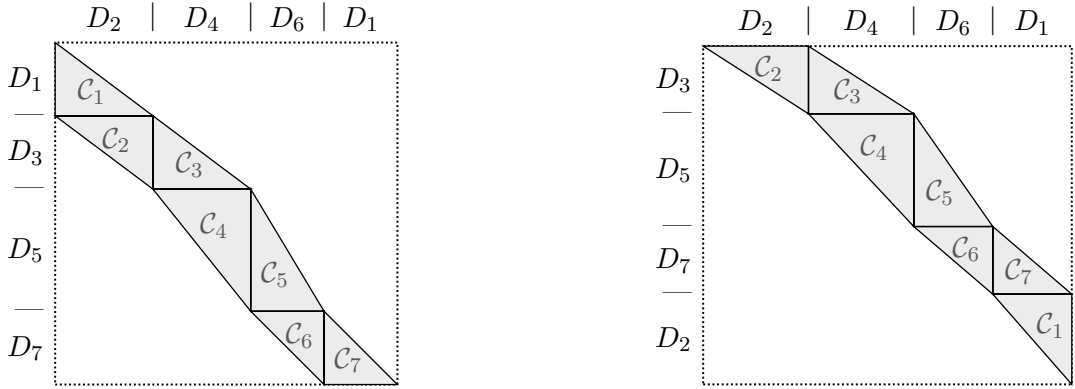
\begin{figure}[H]
\centering{%
\begin{tikzpicture}[scale=0.325,font=\small]
\draw[densely dotted] (0,14) coordinate (0){} --(14,14)--(14,0) coordinate (8) --(0,0)--cycle ;
\draw[densely dotted] (0,11) coordinate (1) {}
(4,11) coordinate (2) {} (4,8) coordinate (3) {}
(8,8)coordinate (4) {} (8,3)coordinate (5) {}
(11,3)coordinate (6) {} (11,0)coordinate (7) {} ;
\draw[densely dotted] (1)--(2)--(3)--(4)--(5)--(6)--(7) ;
\draw[line width=0.5pt,fill=gray!15!white] (0)--(1)--(2)--cycle
(1)--(3)--(2)--cycle (2)--(4)--(3)--cycle (4)--(3)--(5)--cycle
(4)--(6)--(5)--cycle (5)--(6)--(7)--cycle (8)--(6)--(7)--cycle ;
\draw[color=gray!70!black] (1.5,12) node[empty] {$\cC_1$}  (2.5,10) node[empty] {$\cC_2$}
(5.5,9) node[empty] {$\cC_3$}  (6.5,6.5) node[empty] {$\cC_4$}
(9,4.5) node[empty] {$\cC_5$}  (10,2) node[empty] {$\cC_6$}  (12,1) node[empty] {$\cC_7$}  ;
\draw (2,14) node[empty,above] {$D_2$} (4,14)  node[empty,above] {$|$} (6,14) node[empty,above] {$D_4$}
(8,14)  node[empty,above] {$|$}  (9.5,14) node[empty,above] {$D_6$} (11,14)  node[empty,above] {$|$}
(12.5,14) node[empty,above] {$D_1$};
\draw (0,1.5) node[empty,left] {$D_7$} (0,3) node[empty,left] {|}
(0,5.5) node[empty,left] {$D_5$} (0,8) node[empty,left] {|}
(0,9.5) node[empty,left] {$D_3$} (0,11) node[empty,left] {|}
 (0,12.5) node[empty,left] {$D_1$};
\end{tikzpicture}
\hspace*{3cm}
\begin{tikzpicture}[yscale=0.3,xscale=0.35,font=\small]
\draw[densely dotted] (0,11) coordinate (0){} --(14,11)--(14,-4) coordinate (9) --(0,-4)--cycle ;
\draw[densely dotted] (0,11) coordinate (1) {}
(4,11) coordinate (2) {} (4,8) coordinate (3) {}
(8,8)coordinate (4) {} (8,3)coordinate (5) {}
(11,3)coordinate (6) {} (11,0)coordinate (7) {} (14,0)coordinate (8) {} ;
\draw[densely dotted] (1)--(2)--(3)--(4)--(5)--(6)--(7) ;
\draw[line width=0.5pt,fill=gray!15!white]
(1)--(3)--(2)--cycle (2)--(4)--(3)--cycle (4)--(3)--(5)--cycle
(4)--(6)--(5)--cycle (5)--(6)--(7)--cycle (8)--(6)--(7)--cycle
(9)--(8)--(7)--cycle;
\draw[color=gray!70!black] (13.2,-1.4) node[empty] {$\cC_1$}  (2.5,10) node[empty] {$\cC_2$}
(5.5,9) node[empty] {$\cC_3$}  (6.5,6.5) node[empty] {$\cC_4$}
(9,4.5) node[empty] {$\cC_5$}  (10,2) node[empty] {$\cC_6$}  (12,1) node[empty] {$\cC_7$}  ;
\draw (2,11) node[empty,above] {$D_2$} (4,11)  node[empty,above] {$|$} (6,11) node[empty,above] {$D_4$}
(8,11)  node[empty,above] {$|$}  (9.5,11) node[empty,above] {$D_6$} (11,11)  node[empty,above] {$|$}
(12.5,11) node[empty,above] {$D_1$};
\draw (0,-2) node[empty,left] {$D_2$} (0,0) node[empty,left] {|} (0,1.5) node[empty,left] {$D_7$} (0,3) node[empty,left] {|}
(0,5.5) node[empty,left] {$D_5$} (0,8) node[empty,left] {|} (0,9.5) node[empty,left] {$D_3$} ;
\end{tikzpicture}}
\caption{Decomposition of $A(G)$ for a quasimonotone graph $G$}\label{fig:quasimono}
\end{figure}
Suppose there are $k$ chains graphs in the decomposition. In our illustration, Fig.~\ref{fig:quasimono}, $k=7$.
\begin{lemma}\label{lem:splitting}
  A flawless graph $G$ which  has an odd hole of size at least 7 is quasimonotone if and only if it has such a decomposition and does not contain a prehole. If there are $k$ chain graphs in the decomposition, then $k$ is odd, and the shortest hole in $G$ has $k$ vertices.
\end{lemma}
\begin{proof}
  It is clear that $k$ must be odd, since $D_1,D_3,\ldots,D_k$ are the sets of rows.

  The only reason that $G$ could fail to be quasimonotone is that it has an even hole or a prehole. But any hole $H$ must have at least one edge in each of the chain graphs $\cC_1,\cC_2,\ldots,\cC_k$. Otherwise, suppose $H$ has no edge in $\cC_i$. If $i=1$, then $H$ is entirely contained in a monotone graph with decomposition $\cC_{2},\ldots,\cC_k$. If $i=k$, then $H$ is entirely contained in a monotone graph with decomposition $\cC_{1},\ldots,\cC_{k-1}$. Otherwise, $H$ is entirely contained in a monotone graph with decomposition $\cC_{i+1},\ldots,\cC_k,\cC_1,\ldots,\cC_{i-1}$. This contradicts monotonicity.

Now we observe that $H$ must have an odd number of edges in each chain graph $\cC$. This is because the path through $\cC$ comprises alternating horizontal and vertices line segments, representing vertices, which meet at nodes representing edges. See Fig.~\ref{fig:chainpath}. The path must be monotonic within $\cC$, or $H$ would have a chord (see~Fig.~\ref{fig:chainpath}), a contradiction. For the same reason, the path cannot pass through $\cC$ more than once. So the path must enter $\cC$ through a horizontal segment and leave through a vertical segment, or vice versa. Therefore, there must be an odd number of edges of $H$ in $\cC$.  Hence $H$ has an odd number ($k$) of odd numbers of edges, and so the total number of edges of $H$ must be odd. Thus any hole $H$ in $G$ must be an odd hole.
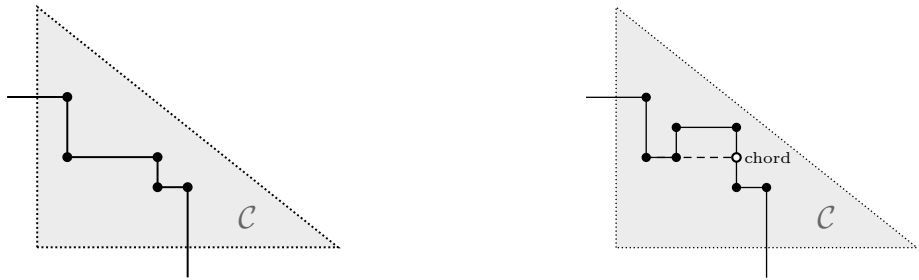
\begin{figure}[H]
\centering{
\begin{tikzpicture}[line width=0.75pt,minimum size=2pt,inner sep=0.8pt,scale=0.4]
\draw[densely dotted,fill=gray!15!white] (0,0)--(0,8)--(10,0)--cycle;
\draw (-1,5)--(1,5) node[b] {} --(1,3) node[b] {} --(4,3) node[b] {} --(4,2) node[b] {} --(5,2) node[b] {} -- (5,-1) ;
\draw (7,1) node[color=gray!80!black] {$\cC$} ;
\end{tikzpicture}
\hspace*{3cm}
\begin{tikzpicture}[line width=0.5pt,minimum size=2pt,inner sep=0.8pt,scale=0.4]
\draw[densely dotted,fill=gray!15!white] (0,0)--(0,8)--(10,0)--cycle;
\draw (-1,5)-- (1,5) node[b] {} --(1,3) node[b] (b) {} --(2,3) node[b] {} --(2,4) node[b] {} --(4,4) node[b] {} --(4,2) node[b] {} --(5,2) node[b] {} -- (5,-1) ;
\draw (4,3) node[inner sep=0.9pt,w,label=right:{\tiny chord}] (a) {}  (a)edge[densely dashed](b) (7,1) node[color=gray!80!black] {$\cC$} ;
\end{tikzpicture}}
\caption{Path through a chain graph}\label{fig:chainpath}
\end{figure}\vspace{-1ex}
Now we observe that the number $k$ of chain graphs in the decomposition of $G$ is the length of a shortest hole. We take exactly one edge in each of the $\cC_i$ ($i=1,2,\ldots,k$), connected by a zig-zag path of vertices, as indicated in Fig.~\ref{fig:quasimono}.
\end{proof}
However, it is still possible that $G$ contains a prehole. However, the decomposition of Lemma~\ref{lem:splitting} implies that any prehole must wind around the annular structure of $G$.  We consider the question of detecting such preholes in section~\ref{ss:findprehole} below.

The decomposition of $G$ can clearly be carried out in polynomial time. We check that $G$ is flawless, and has no triangle, or hole of size smaller than 7. We check that $G$ is not a prehole. If so, we determine a short odd hole. Then we use the algorithm of~\cite{SpBrSt87} to determine the monotone structure of $\cG_v$, the monotone structure of $\cG_w$, and the split of $\cN_w$. If any  of these steps fails, $G$ is not quasimonotone. If all succeed, $G$ is quasimonotone, unless it contains a prehole. As an example, consider the graph $G$ shown in Fig.~\ref{fig:qmexample}.\vspace{-1ex}
\begin{figure}[ht]
\centering{
\begin{tikzpicture}[line width=0.5pt,minimum size=2pt,inner sep=0.9pt,scale=1,font=\footnotesize]
\draw (90:1) node[b] (0){} -- (90+360/7:1) node[b] (1) {} -- (90+2*360/7:1) node[b] (2) {} -- (90+3*360/7:1) node[b] (3) {} -- (90+4*360/7:1) node[b] (4) {} -- (90+5*360/7:1) node[b] (5) {} -- (90+6*360/7:1) node[b] (6) {} -- (0) ;
\draw (90:1.67) node[b] (0'){} -- (90+360/7:1.67) node[b] (1') {} -- (90+2*360/7:1.67) node[b] (2') {} -- (90+3*360/7:1.67) node[b] (3') {} -- (90+4*360/7:1.67) node[b] (4') {} -- (90+5*360/7:1.67) node[b] (5') {} -- (90+6*360/7:1.67) node[b] (6') {} -- (0') ;
\foreach \x in {0,...,6} {\draw (\x)--(\x');};
\foreach \x [evaluate = \x as \y using \x-1] in {1,...,7} {\draw (90-\y*360/7:0.8) node {$\x$};};
\foreach \x [evaluate = \x as \y using \x-1] in {1,...,7} {\draw (90-\y*360/7:1.85) node {$\x'$};};
\draw (0,0) node {\large$G$};
\end{tikzpicture}\hspace*{1cm}
\begin{tikzpicture}[line width=0.5pt,minimum size=2pt,inner sep=0.9pt,scale=1,font=\footnotesize]
\draw  (90-2*360/7:1) node[w] (3) {} -- (90-3*360/7:1) node[b] (4) {} -- (90-4*360/7:1) node[w] (5) {} -- (90-5*360/7:1) node[b] (6) {} ;
\draw (90-360/7:1.67) node[w] (2') {} -- (90-2*360/7:1.67) node[b] (3') {} -- (90-3*360/7:1.67) node[w] (4') {} -- (90-4*360/7:1.67) node[b] (5') {} -- (90-5*360/7:1.67) node[w] (6') {} -- (90-6*360/7:1.67) node[b] (7') {} ;
\foreach \x in {3,...,6} {\draw (\x)--(\x');};
\foreach \x [evaluate = \x as \y using \x-1] in {3,...,6} {\draw (90-\y*360/7:0.8) node {$\x$};};
\foreach \x [evaluate = \x as \y using \x-1] in {2,...,7} {\draw (90-\y*360/7:1.85) node {$\x'$};};
\draw (0,0) node {\large$\cG_1$};
\end{tikzpicture}\hspace*{1cm}
\begin{tikzpicture}[line width=0.5pt,minimum size=2pt,inner sep=0.9pt,scale=1,font=\footnotesize]
\draw (90:1) node[w] (0){} -- (90-360/7:1) node[b] (1) {} -- (90-2*360/7:1) node[w] (2) {} (90-5*360/7:1) node[w] (5) {} -- (90-6*360/7:1) node[b] (6) {} -- (0) ;
\draw (90:1.67) node[b] (0'){} -- (90-360/7:1.67) node[w] (1'){} (90-6*360/7:1.67) node[w] (6') {} -- (0') ;
\draw (0)--(0') (1)--(1') (6)--(6') ;
\foreach \x [evaluate = \x as \y using \x-1] in {1,2,3,6,7} {\draw (90-\y*360/7:0.8) node {$\x$};};
\foreach \x [evaluate = \x as \y using \x-1] in {1,2,7} {\draw (90-\y*360/7:1.85) node {$\x'$};};
\draw (0,-1.67) node {};
\draw (0,0) node {\large$\cN_1$};
\end{tikzpicture}\\[2ex]
\scalebox{0.85}{$\kbordermatrix{ & 7 & 1' & 2' & 3' & 4' & 5' & 6' & 7'\\
1 & 1 & 1 & 1 & 0 & 0 & 0 & 0 & 0\\
2 & 0 & 1 & 1 & 1 & 0 & 0 & 0 & 0\\
3 & 0 & 0 & 1 & 1 & 1 & 0 & 0 & 0\\
4 & 0 & 0 & 0 & 1 & 1 & 1 & 0 & 0\\
5 & 0 & 0 & 0 & 0 & 1 & 1 & 1 & 0\\
6 & 0 & 0 & 0 & 0 & 0 & 1 & 1 & 1\\
7 & 0 & 0 & 0 & 0 & 0 & 0 & 1 & 1\\
1' & 0 & 0 & 0 & 0 & 0 & 0 & 0 & 1\\
}$}
}
\caption{$G$, $\cG_1$, $\cN_1$ and the derived $A(G)$.}\label{fig:qmexample}
\end{figure}
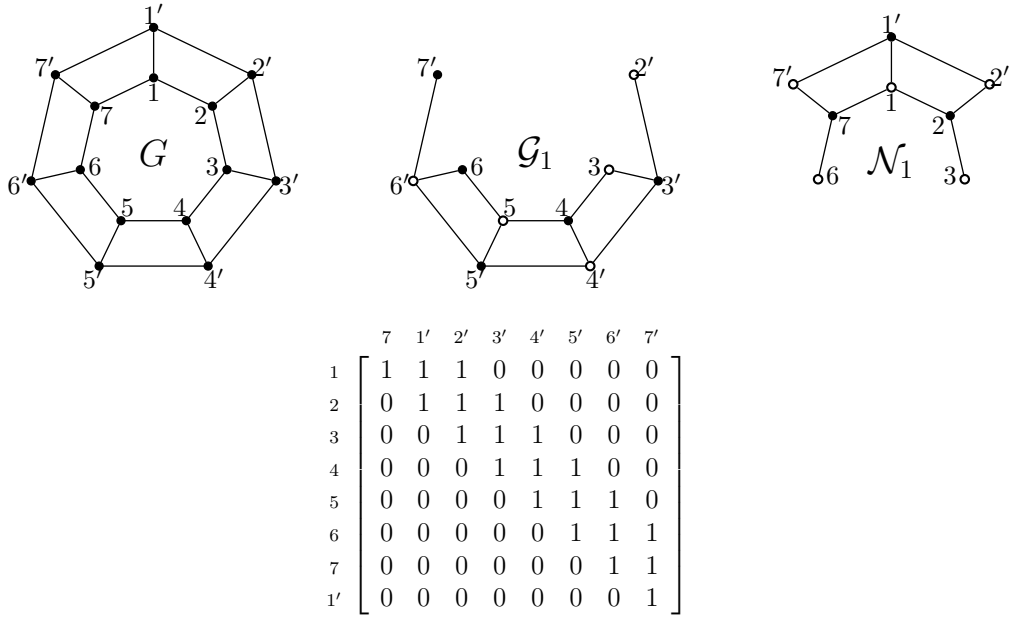

Observe that this procedure really only requires that $G$ be triangle-free, and have diameter at least 3. Thus it can be applied to test quasimonotonicity of some graphs with 5-holes, for example that in Fig.~\ref{fig:twoholes}.

\subsection{Recognising preholes}\label{ss:findprehole}
Let $G=(V,E)$ be a flawless graph with a hole of size $\ell\geq 7$. Lemma~\ref{lem:splitting} can determine whether or not $G$ is quasimonotone provided it does not contain a prehole. We now consider recognition of a prehole in such a graph.

We use the partition of $V$ from section~\ref{ss:splitting} into independent sets $D_1,D_2,\ldots,D_\ell$, where $D_{\ell+1}\equiv D_1$. All edges in $E$ run between $D_i$ and $D_{i+1}$ ($i\in[\ell]$). Let $G_i=G[D_i\cup D_{i+1}]$, with edge set $E_i$, and let $\overline{G}_i=(V,E\setminus E_i)$. Note that $G_i$ is a chain graph and $\overline{G}_i$ is a monotone graph. Thus $\overline{G}_i$ is bipartition, with bipartition $\LR$, say, with $D_i,D_{i+1}\in L$.

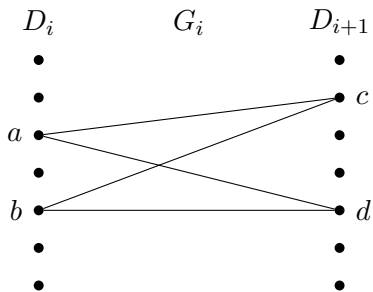
\begin{figure}[ht]
\centerline{
\begin{tikzpicture}[xscale=2,yscale=0.5,font=\small];
\draw
 (0,0) node[b] (1) {} (2,0) node[b] (1') {}
 (0,1) node[b] (2) {} (2,1) node[b] (2') {}
 (0,2) node[b,label=left:$b$] (3) {} (2,2) node[b,label=right:$d$] (3') {}
 (0,3) node[b] (4) {} (2,3) node[b] (4') {}
 (0,4) node[b,label=left:$a$] (5) {} (2,4) node[b] (5') {}
 (0,5) node[b] (6) {} (2,5) node[b,label=right:$c$] (6') {}
 (0,6) node[b] (7) {} (2,6) node[b] (7') {};
 \draw[thin] (5)--(3')--(3)--(6')--(5);
 \draw (1,7) node[empty]{$G_i$} (0,7) node[empty]{$D_i$} (2,7) node[empty]{$D_{i+1}$} ;
 \end{tikzpicture} }
\caption{Possible crossover}\label{fig:findprehole1}
\end{figure}

We search for possible crossovers in $G_i$. These are pairs $a,b\in D_{i+1}$, $c,d\in D_i$, such that $ac,ad,bc,bd\in E$. We list all such quadruples $a,b,c,d$, $O(n^4)$ in total, see Fig.~\ref{fig:findprehole1}.
Given any quadruple, we attempt to determine vertex disjoint paths $P_{ac},P_{bd}$ in $\overline{G}_i$ between $a,c$ and $b,d$ or between $a,d$ and $b,c$. See Fig.~\ref{fig:findprehole2}, cases (a) and (b). We can do this in $O(n|E|)=O(n^3)$ time by network flow. Both paths are even length, since $G_i$ is bipartite and $a,b,c,d\in L$.

 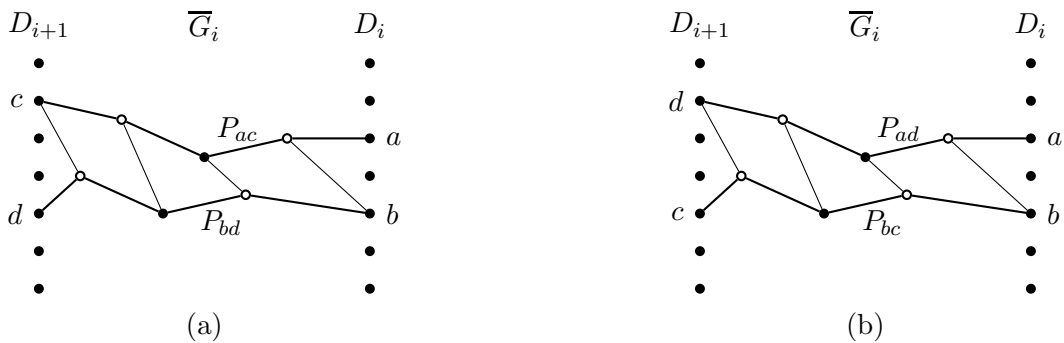
\begin{figure}[H]
 \centerline{
\begin{tikzpicture}[xscale=2.2,yscale=0.5,font=\small]
\draw
 (2,0) node[b] (1) {} (0,0) node[b] (1') {}
 (2,1) node[b] (2) {} (0,1) node[b] (2') {}
 (2,2) node[b,label=right:$b$] (3) {} (0,2) node[b,label=left:$d$] (3') {}
 (2,3) node[b] (4) {} (0,3) node[b] (4') {}
 (2,4) node[b,label=right:$a$] (5) {} (0,4) node[b] (5') {}
 (2,5) node[b] (6) {} (0,5) node[b,label=left:$c$] (6') {}
 (2,6) node[b] (7) {} (0,6) node[b] (7') {};
\draw  (0.5,4.5) node[w] (x) {} (1,3.5) node[b] (y') {} (1.5,4) node[w] (z) {} ;
\draw[thick] (6')--(x)--(y')--(z)--(5) ;
  \draw  (0.25,3) node[w] (x') {} (0.75,2) node[b] (y) {} (1.25,2.5) node[w] (z') {} ;
\draw[thick] (3')--(x')--(y)--(z')--(3) ;
\draw[thin] (6')--(x') (x)--(y) (y')--(z') (z)--(3);
  \draw (2,7) node[empty]{$D_i$} (0,7) node[empty]{$D_{i+1}$}
 (1,7) node[empty]{$\overline{G}_i$} (1.2,4.25) node[empty]{$P_{ac}$} (1.1,1.75) node[empty]{$P_{bd}$} ;
 \draw (1,-1) node {(a)};
 \end{tikzpicture}
 \hspace*{3cm}
 \begin{tikzpicture}[xscale=2.2,yscale=0.5,font=\small]
\draw
 (2,0) node[b] (1) {} (0,0) node[b] (1') {}
 (2,1) node[b] (2) {} (0,1) node[b] (2') {}
 (2,2) node[b,label=right:$b$] (3) {} (0,2) node[b,label=left:$c$] (3') {}
 (2,3) node[b] (4) {} (0,3) node[b] (4') {}
 (2,4) node[b,label=right:$a$] (5) {} (0,4) node[b] (5') {}
 (2,5) node[b] (6) {} (0,5) node[b,label=left:$d$] (6') {}
 (2,6) node[b] (7) {} (0,6) node[b] (7') {};
\draw  (0.5,4.5) node[w] (x) {} (1,3.5) node[b] (y') {} (1.5,4) node[w] (z) {} ;
\draw[thick] (6')--(x)--(y')--(z)--(5) ;
  \draw  (0.25,3) node[w] (x') {} (0.75,2) node[b] (y) {} (1.25,2.5) node[w] (z') {} ;
\draw[thick] (3')--(x')--(y)--(z')--(3) ;
\draw[thin] (6')--(x') (x)--(y) (y')--(z') (z)--(3);
  \draw (2,7) node[empty]{$D_i$} (0,7) node[empty]{$D_{i+1}$}
 (1,7) node[empty]{$\overline{G}_i$} (1.2,4.25) node[empty]{$P_{ad}$} (1.1,1.75) node[empty]{$P_{bc}$} ;
 \draw (1,-1) node {(b)};
 \end{tikzpicture} }
\caption{Vertex-disjoint paths}\label{fig:findprehole2}
\end{figure}
 If these paths do not exist, we discard this quadruple and consider the next in the list. If these paths do exist, in case (a) we have found a crossover prehole $P_{ac},ad,P_{bd},bc$, in case (b) we have found a M\"obius prehole $P_{ad},bd,P_{bc},ac$. This is clearly a cycle with even length. That it is a prehole is certified by reversing the bipartition on $P_{ac}$ in case(a), $P_{ad}$ in case (b), as shown in Fig.~\ref{fig:findprehole3}.

 \begin{figure}[H]
 \centerline{
\begin{tikzpicture}[xscale=2,yscale=0.5,font=\small]
\draw
 (2,2) node[w,label=right:$b$] (3) {} (0,2) node[w,label=left:$d$] (3') {}
 (2,4) node[b,label=right:$a$] (5) {}  (0,5) node[b,label=left:$c$] (6') {} ;
\draw  (0.5,4.5) node[w] (x) {} (1,3.5) node[b] (y') {} (1.5,4) node[w] (z) {} ;
\draw[thick] (6')--(x)--(y')--(z)--(5) ;
  \draw  (0.25,3) node[b] (x') {} (0.75,2) node[w] (y) {} (1.25,2.5) node[b] (z') {} ;
\draw[thick] (3')--(x')--(y)--(z')--(3) ;
\draw[thin] (6')--(x') (x)--(y) (y')--(z') (z)--(3);
  \draw[thick] (5)..controls (3,0) and (2,-1)..(3') (3)..controls (3,7) and (2,7)..(6') ;
  \draw[thin] (6')edge[bend left=40](5) (3')edge[bend right=55](3) ;
  \draw (1.2,4.2) node[empty]{\scriptsize$P_{ac}$} (0.85,2.6) node[empty]{\scriptsize$P_{bd}$};
\draw (1,-1) node {(a) Crossover};
  \end{tikzpicture}
\hspace*{1.5cm}
\begin{tikzpicture}[xscale=2,yscale=0.5,font=\small]
\draw
 (2,2) node[w,label=right:$b$] (3) {} (0,2) node[w,label=left:$c$] (3') {}
 (2,4) node[b,label=right:$a$] (5) {}  (0,5) node[b,label=left:$d$] (6') {} ;
\draw  (0.5,4.5) node[w] (x) {} (1,3.5) node[b] (y') {} (1.5,4) node[w] (z) {} ;
\draw[thick] (6')--(x)--(y')--(z)--(5) ;
  \draw  (0.25,3) node[b] (x') {} (0.75,2) node[w] (y) {} (1.25,2.5) node[b] (z') {} ;
\draw[thick] (3')--(x')--(y)--(z')--(3) ;
\draw[thin] (6')--(x') (x)--(y) (y')--(z') (z)--(3);
  \draw[thick] (5)..controls (3,0) and (2,-1)..(3') (3)..controls (3,7) and (2,7)..(6') ;
  \draw[thin] (6')edge[bend left=40](5) (3')edge[bend right=55](3) ;
  \draw (1.2,4.2) node[empty]{\scriptsize$P_{ad}$} (0.85,2.6) node[empty]{\scriptsize$P_{bc}$};
  \draw (1,-1) node {(b) M\"obius};
  \end{tikzpicture} }
\caption{Preholes}\label{fig:findprehole3}
\end{figure}
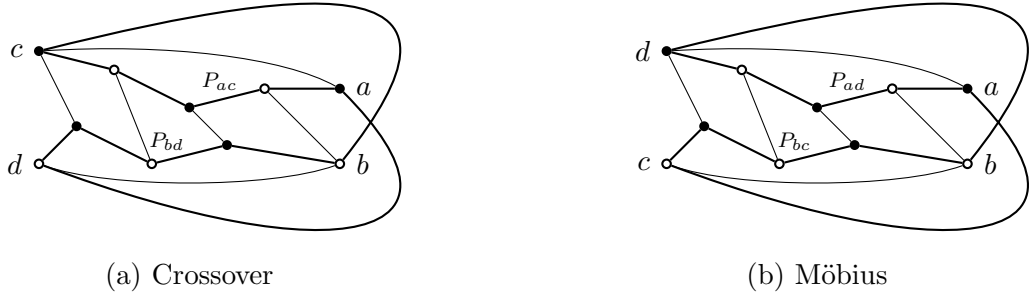

Thus we can detect a prehole, or show that none exists, in $O(n^4\times n^3)=O(n^7)$ time.
If a prehole exists, we may stop. We have shown that $G$ is not quasimonotone.

\section{Flawless graphs without long holes} \label{sec:nolongholes}

\subsection{Minimal preholes in hole-free graphs}\label{ss:noholes}

The main problem here is to recognise preholes. Let $C$ be any minimal prehole in a flawless hole-free graph $G$. A triangle in $G[C]$ will be called an \emph{interior} triangle of $C$ if it has no edge in common with $C$,  a \emph{crossing} triangle if it has one edge in common with $C$, and a \emph{cap} of $C$ if it has two edges in common with $C$.
\begin{lemma}\label{prehole:lem30}
  If $C$ is a minimal prehole in a flawless graph with $|C|>12$, then $G[C]$ has no interior or crossing triangles, and $C$ is determined by two edge-disjoint caps.
\end{lemma}
\begin{proof}
Suppose $C$ has an interior triangle, as shown in Fig.~\ref{fig:internal triangle}. The vertices $v,w,x$ partition $C$ into three segments. From Lemma~\ref{lem:quasi35}, we must have $\dist_C(v,w)\leq 4$, $\dist_C(w,x)\leq 4$ and $\dist_C(x,v)\leq 4$, and hence $|C|\leq 12$.
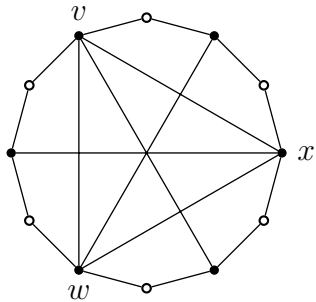
\begin{figure}[H]
\centering{
\begin{tikzpicture}[line width=0.5pt,minimum size=1.5pt,xscale=0.9,yscale=0.9]
\draw (-90:2) node[w] (0) {}  (-30:2) node[w] (2) {}  (30:2) node[w] (4) {} (90:2) node[w] (6) {}  (150:2) node[w] (8) {}  (210:2) node[w] (10) {}   ;
\draw (-60:2) node[b] (1) {}  (0:2) node[b,label=right:$x$] (3) {}  (60:2) node[b] (5) {}   (120:2) node[b,label=above:$v$] (7) {} (180:2) node[b] (9) {} (240:2) node[b,label=below:$w$] (11) {}  ;
\draw (0)--(1)--(2)--(3)--(4)--(5)--(6)--(7)--(8)--(9)--(10)--(11)--(0) ;
\draw (7)--(11)--(3)--(7) (7)--(1) (3)--(9) (11)--(5);
\end{tikzpicture}}
\caption{An internal triangle}\label{fig:internal triangle}
\end{figure}
  If $G[C]$ has a crossing triangle, then $C$ has an odd chord, a contradiction. Now, since $C$ is a prehole, there must be an even chord $u_0v_0$ which partitions $C$ into two odd cycles $C_1$, $C_2$ with common edge $u_0v_0$. Suppose that $C_1$ is not a triangle. Since $G\in\hf$, $C_1$ must have a chord $u_1v_1$ partitioning it into an odd cycle $C'_1$ and an even cycle $C'_2$, with $|C'_1|<|C|$. If $C'_1$ is not a triangle we repeat the process until we reach a triangle $T_1$, which must be a cap, and must be unique. Otherwise we have discovered a crossing or internal triangle of $C$, a contradiction. This must occur after at most $|C_1|$ repetitions. We then apply the same procedure to $C_2$, obtaining the second cap $T_2$.  Clearly $T_1$ and $T_2$ are edge-disjoint, since they are separated by the chord $u_0v_0$. Then $|C|\geq 6$ implies that they can share at most one of the vertices $u_0,v_0$.
\end{proof}
Note that preholes with fewer than 12 vertices may contain an interior triangle. See Fig.~\ref{fig:degree2} for an example with 6 vertices. However, the bound 12 is probably far from tight.

Let $T_1$, $T_2$ be caps of $C$, such that $v_i\in T_i$ is adjacent to two edges of $C$ $(i=1,2)$. Then there are two edge-disjoint \Path{v_1}{v_2}s $P_1,P_2$ in $C$.  See  Fig.~\ref{prehole:fig30}.

\begin{figure}[H]
\tikzset{every node/.style={circle,draw,fill=none,inner sep=0pt,minimum size=1.25mm}}
\centerline{\begin{tikzpicture}[xscale=0.8,yscale=0.6,font=\small]
\draw (0,0) node[b,label=left:$v_1$] (0) {}
 (1,1) node[w] (1) {} (1,-1) node[w] (1') {}
 (2,1) node[b] (2) {} (2,-1) node[b] (2') {}
 (3,1) node[w] (3) {} (3,-1) node[w] (3') {}
 (5,1) node[w] (5) {} (5,-1) node[w] (5') {}
 (6,1) node[b] (6) {} (6,-1) node[b] (6') {}
 (7,1) node[w] (7) {} (7,-1) node[w] (7') {}
 (8,0) node[b,label=right:$v_2$] (8) {} ;
 \draw (3')--(2')--(1')--(0)--(1)--(2)--(3) (5')--(6')--(7')--(8)--(7)--(6)--(5) ;
 \draw[dashed] (3)--(5) (3')--(5') ;
 \draw[dotted] (1)--(1') (7)--(7') ;
 \draw (0.6,0) node[empty] {$T_1$}  (7.4,0) node[empty] {$T_2$}
 (2.5,1.5) node[empty] {$P_1$}  (5.5,-1.5) node[empty] {$P_2$}  ;
\end{tikzpicture}\hspace*{1cm}
\begin{tikzpicture}[xscale=0.8,yscale=0.6,font=\small]
\draw (0,0) node[w,label=left:$v_1$] (0) {}
 (1,1) node[b] (1) {} (1,-1) node[w] (1') {}
 (2,1) node[w] (2) {} (2,-1) node[b] (2') {}
 (3,1) node[b] (3) {} (3,-1) node[w] (3') {}
 (5,1) node[b] (5) {} (5,-1) node[w] (5') {}
 (6,1) node[w] (6) {} (6,-1) node[b] (6') {}
 (7,1) node[b] (7) {} (7,-1) node[w] (7') {}
 (8,0) node[w,label=right:$v_2$] (8) {} ;
 \draw (3')--(2')--(1')--(1)--(2)--(3) (5')--(6')--(7')--(7)--(6)--(5) ;
 \draw[dashed] (3)--(5) (3')--(5') ;
 \draw (1)--(1') (7)--(7') ;
 \draw (2.5,1.5) node[empty] {$P_1$}  (5.5,-1.5) node[empty] {$P_2$}  ;
\end{tikzpicture}}
\caption{A prehole and its Hamilton subgraph}\label{prehole:fig30}
\end{figure}
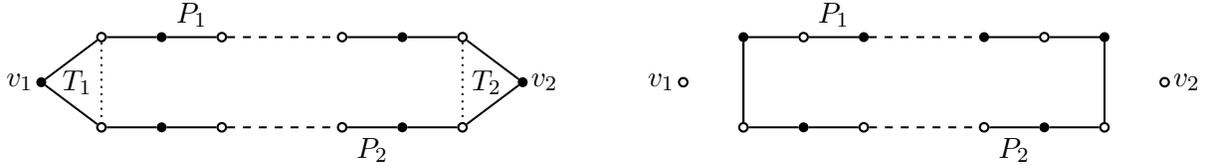

\begin{lemma}\label{prehole:lem35}
 Let $C$, with $|C|>12$, be a minimal prehole in a flawless hole-free graph determined by $v_1,v_2$, and let $C'=C\sm\{v_1,v_2\}$. Then $G[C']$ is a Hamilton monotone graph, and all chords of $C'$ connect $P_1$ to $P_2$.
\end{lemma}
\begin{proof}
  Clearly $G[C']$ is Hamilton, since $G[C]$ is Hamilton. Now $C'$ cannot be a prehole, since it is strictly smaller than $C$. So $G[C']$ cannot contain a triangle, by Lemma~\ref{prehole:lem30}. It cannot contain a larger odd cycle, since then it would contain a triangle, by the argument of Lemma~\ref{prehole:lem30}. Therefore, $G[C']$ is bipartite and, since $G\in\hf$,  contains no hole. So, since $G\in\noflaw$, $G[C']$ is a monotone graph. Suppose $uv$ is an edge of $G[C']$ with $u,v\in P_1$. Then, since $G[C]$ has only even chords, the even chord $uv$ and the segment of $P_1$ between $u$ and $v$ forms an odd cycle, giving a contradiction.
\end{proof}

Thus any minimal prehole $C$ comprises a Hamilton monotone graph $G[C']$, to which we add two caps $T_1$, $T_2$. We may also add edges from $v_1$ and $v_2$ to $C'$, as long as they are even chords in $C$.

\begin{lemma}\label{prehole:lem37}
  Let $C$ be a minimal prehole with a cap at $v\in\{v_1,v_2\}$. Then there are at most two chords
  from $v$, and both must be connected to either $P_1$ or $P_2$.
\end{lemma}
\begin{proof}
The chords must be as shown in  Fig.~\ref{fig:v chords}(a), since otherwise there is an even hole or a pre-armchair, similarly to Lemma~\ref{lem:quasi30}. Note that $vb$ must be present if $vd$ is in $G$. If there is a chord to both $P_1$ and $P_2$, we can find a smaller prehole by moving $v$ from $L$ to $R$ and using the longer chords from $v$. See Fig.~\ref{fig:v chords}(b).
\end{proof}
\begin{figure}[H]
\tikzset{every node/.style={circle,draw,fill=none,inner sep=0pt,minimum size=1.25mm}}
\centerline{
\begin{tikzpicture}[xscale=0.8,yscale=0.7,font=\small]
\draw  (3,1) node[w,label=above:$a$] (a) {} (2,0) node[b,label=left:$v$] (v) {}
 (4,1) node[b,label=above:$b$] (b) {} (3,-1) node[w,label=below:$w$] (w) {} (5,1) node[w,label=above:$c$] (c) {} (4,-1) node[b,label=below:$x$] (x) {}  (5,-1) node[w,label=below:$y$] (y) {} (6,-1) node[b,label=below:$z$] (z) {}  (6,1) node[b,label=above:$d$] (d) {} ;
\draw (a)--(v)--(w)--(x)--(y)--(z)--(7,-1) (v)--(a)--(b)--(c)--(d)--(7,1) (a)--(w);
\draw[dashed] (v)--(b) (v)--(d) ;
\draw (5,-2.5) node[empty]{(a)};
\end{tikzpicture}
\hspace*{3cm}
\begin{tikzpicture}[xscale=0.8,yscale=0.7,font=\small]
\draw  (3,1) node[w,label=above:$a$] (a) {} (2,0) node[w,label=left:$v$] (v) {}
 (4,1) node[b,label=above:$b$] (b) {} (3,-1) node[w,label=below:$w$] (w) {} (5,1) node[w,label=above:$c$] (c) {} (4,-1) node[b,label=below:$x$] (x) {}  (5,-1) node[w,label=below:$y$] (y) {} (6,-1) node[b,label=below:$z$] (z) {}  (6,1) node[b,label=above:$d$] (d) {} ;
\draw (a)--(v)--(w)--(x)--(y)--(z)--(7,-1) (v)--(a)--(b)--(c)--(d)--(7,1) (a)--(w);
\draw[dashed] (v)--(b) (v)--(d) (v)--(x) (v)--(z) ;
\draw (5,-2.5) node[empty]{(b)};
\end{tikzpicture}}
\caption{}\label{fig:v chords}
\end{figure}

\begin{lemma}\label{prehole:lem40}
  Let $C$ be a minimal prehole with $|C|\geq 8$. Then all vertices in $P_1$ have a chord to $P_2$ and vice versa.
\end{lemma}
\begin{proof}
Suppose first that there are no chords from $v_1$ or $v_2$.
\begin{figure}[H]
\tikzset{every node/.style={circle,draw,fill=none,inner sep=0pt,minimum size=1.25mm}}
\centerline{\begin{tikzpicture}[xscale=0.7,yscale=0.6,font=\small]
\draw (0,-1) node[b] (0') {} (1,-1) node[w] (1') {}
 (2,1) node[w] (2) {} (2,-1) node[b] (2') {}
 (3,1) node[b,label=above:$u$] (3) {} (3,-1) node[w,label=below:$w$] (3') {}  (4,1) node[w] (4) {} (4,-1) node[b] (4') {}
 (5,-1) node[w] (5') {} (6,-1) node[b] (6') {} ;
\draw (-1,-1)--(0')--(1')--(2')--(3')--(4')--(5')--(6')--(7,-1) (1,1)--(2)--(3)--(4)--(5,1) (2)--(3')--(4);
\draw (3,-2.5) node[empty]{(a)};
\end{tikzpicture}
\hspace*{1cm}
\begin{tikzpicture}[xscale=0.7,yscale=0.6,font=\small]
\draw  (3,1) node[w,label=above:$a$] (a) {} (2,0) node[b,label=left:$v$] (v) {}
 (4,1) node[b,label=above:$b$] (b) {} (3,-1) node[w,label=below:$w$] (w) {} (5,1) node[w,label=above:$c$] (c) {} (4,-1) node[b,label=below:$x$] (x) {}  (5,-1) node[w,label=below:$y$] (y) {} (6,-1) node[b,label=below:$z$] (z) {}  (6,1) node[b,label=above:$d$] (d) {} ;
\draw (a)--(v)--(w)--(x)--(y)--(z)--(7,-1) (v)--(a)--(b)--(c)--(d)--(7,1) (a)--(w)
(a)--(y);
\draw[dashed] (b)--(z) (c)--(y);
\draw (5,-2.5) node[empty]{(b)};
\end{tikzpicture}
\hspace*{1cm}
\begin{tikzpicture}[xscale=0.7,yscale=0.6,font=\small]
\draw  (3,1) node[w,label=above:$a$] (a) {} (2,0) node[b,label=left:$v$] (v) {}
 (4,1) node[b,label=above:$b$] (b) {} (3,-1) node[w,label=below:$w$] (w) {} (5,1) node[w,label=above:$c$] (c) {} (4,-1) node[b,label=below:$x$] (x) {}  (5,-1) node[w,label=below:$y$] (y) {} (6,-1) node[b,label=below:$z$] (z) {}  (6,1) node[b,label=above:$d$] (d) {} ;
\draw (a)--(v)--(w)--(x)--(y)--(z)--(7,-1) (v)--(a)--(b)--(c)--(d)--(7,1) (a)--(w)
(a)--(y) (c)--(w) ;
\draw (5,-2.5) node[empty]{(c)};
\end{tikzpicture}}
\caption{}\label{fig:crossing chord}
\end{figure}
Let $u\in P_1$ have no edge to $P_2$. Since $G[C']$ is monotone, $u$ must be in a quadrangle with its neighbours in $P_1$ and a vertex $w\in P_2$. See Fig.~\ref{fig:crossing chord}(a). Now $u$ is a distance at least 2 from $v_1$ and $v_2$, since otherwise it has an edge of $T_1$ or $T_2$ to $P_2$. Thus $u$ is at distance at least 2 from $P_2$. Now if $w$ is a distance at least 3 from both $v_1$ and $v_2$, Lemma~\ref{lem:quasi10} implies that $u$ must have an edge to $P_1$, a contradiction.

Otherwise, we have the situations shown in Fig.~\ref{fig:crossing chord}(b), where the edge $bx$ is absent. Thus $ay$ or $cw$, or both, must be in $G$. Suppose that only one is in $G$ and, without loss of generality, that it is $wb$, as shown in Fig.~\ref{fig:crossing chord}(b). Now either $bz$ or $cy$, or both, must be in $G$. If $bz\in G$, $v,a,b,z,y,x,w$  give a pre-stirrer, a contradiction. So suppose only $cy$ is in $G$. Then $v,w,x,y,z,a,b$  give a pre-stirrer, again a contradiction. cannot have a chord to $P_1$ or we would have an even hole $(d,w,a,u,b,\ldots,d)$. Thus the edge $bx\in G$ unless both $ay,cw\in G$, as shown in Fig.~\ref{fig:crossing chord}(c).  In this case we have a shorter prehole $(\ldots,d,c,w,v,a,y,z,\ldots)$, after interchanging $v,a,w$ between $L$ and $R$.
Thus we must have $bx\in G$, giving the conclusion.
Observe that the configuration of  case (a) requires at least 10 vertices, and those of cases (b) and (c) require at least 8. Therefore the conclusion holds only if $|C|\geq 8$.

Now we must consider the effect of chords from $v_1$ or $v_2$. These do not affect case (a), since $P_2$ remains prechordless. Also $u$ must remain at distance 2 from $P_2$, since otherwise we are in case (b) or (c).

In case (b), the only chord from $v$ that can break the flaws in Fig.~\ref{fig:crossing chord}(b) is the edge $vx$, as shown in Fig.~\ref{fig:bad v chord}(a). In this case, we simply give $w$ the role of $v$, as shown in  Fig.~\ref{fig:bad v chord}(b). Note that the chord $vz$, if it exists, will now connect $P_1$ to $P_2$.
\begin{figure}[ht]
\tikzset{every node/.style={circle,draw,fill=none,inner sep=0pt,minimum size=1.25mm}}
\centerline{
\begin{tikzpicture}[xscale=0.8,yscale=0.6,font=\small]
\draw  (3,1) node[w,label=above:$a$] (a) {} (2,0) node[b,label=left:$v$] (v) {}
 (4,1) node[b,label=above:$b$] (b) {} (3,-1) node[w,label=below:$w$] (w) {} (5,1) node[w,label=above:$c$] (c) {} (4,-1) node[b,label=below:$x$] (x) {}  (5,-1) node[w,label=below:$y$] (y) {} (6,-1) node[b,label=below:$z$] (z) {}  (6,1) node[b,label=above:$d$] (d) {} ;
\draw (a)--(v)--(w)--(x)--(y)--(z)--(7,-1) (v)--(a)--(b)--(c)--(d)--(7,1) (a)--(w)
(a)--(y) (v)--(x);
\draw[dashed] (b)--(z) (c)--(y);
\draw (5,-2.5) node[empty]{(a)};
\end{tikzpicture}
\hspace*{3cm}
\begin{tikzpicture}[xscale=0.8,yscale=0.6,font=\small]
\draw  (4,1) node[w,label=above:$a$] (a) {} (3,1) node[b,label=above:$v$] (v) {}
 (5,1) node[b,label=above:$b$] (b) {} (2,0) node[w,label=left:$w$] (w) {}
 (7,1) node[b,label=above:$d$] (d) {} (6,1) node[w,label=above:$c$] (c) {} (3,-1) node[b,label=below:$x$] (x) {}  (4,-1) node[w,label=below:$y$] (y) {} (5,-1) node[b,label=below:$z$] (z) {}--(6,-1) node[w,fill=white] {}--(7,-1) node[b] {}    ;
\draw (a)--(v)--(w)--(x)--(y)--(z) (v)--(a)--(b)--(c)--(d)--(7,1) (a)--(w)
(a)--(y) (v)--(x);
\draw[dashed] (b)--(z) (c)--(y);
\draw (5,-2.5) node[empty]{(b)};
\end{tikzpicture}}
\caption{}\label{fig:bad v chord}
\end{figure}
Finally, consider the configuration of Fig.~\ref{fig:crossing chord}(c). By symmetry, we can assume that the chords from $v$ are $vb,vd$, as shown in Fig.~\ref{fig:another v chord}(a). The edge $vb$ is absent from the shorter prehole in Fig.~\ref{fig:another v chord}(b), so we have only to consider the edge $vd$. However, if $vd\in G$, it now connects $L$ to $R$. Thus there is a shorter prehole $(\ldots,z,y,a,v,d,\ldots)$, a contradiction.
\end{proof}
\begin{figure}[H]
\tikzset{every node/.style={circle,draw,fill=none,inner sep=0pt,minimum size=1.25mm}}
\centerline{
\begin{tikzpicture}[xscale=0.8,yscale=0.6,font=\small]
\draw  (3,1) node[w,label=above:$a$] (a) {} (2,0) node[b,label=left:$v$] (v) {}
 (4,1) node[b,label=above:$b$] (b) {} (3,-1) node[w,label=below:$w$] (w) {} (5,1) node[w,label=above:$c$] (c) {} (4,-1) node[b,label=below:$x$] (x) {}  (5,-1) node[w,label=below:$y$] (y) {} (6,-1) node[b,label=below:$z$] (z) {}  (6,1) node[b,label=above:$d$] (d) {} ;
\draw (a)--(v)--(w)--(x)--(y)--(z)--(7,-1) (v)--(a)--(b)--(c)--(d)--(7,1) (a)--(w)
(a)--(y) (c)--(w) ;
\draw[dashed] (v)--(b) (v)--(d);
\draw (5,-2.5) node[empty]{(a)};
\end{tikzpicture}
\hspace*{3cm}
\begin{tikzpicture}[xscale=0.8,yscale=0.6,font=\small]
\draw  (3,1) node[b,label=above:$a$] (a) {} (2,0) node[w,label=left:$v$] (v) {}
 (4,1) node[w,label=above:$y$] (y) {} (3,-1) node[b,label=below:$w$] (w) {} (5,1) node[b,label=above:$z$] (z) {} (4,-1) node[w,label=below:$c$] (c) {}  (5,-1) node[b,label=below:$d$] (d) {} (6,-1) node[w] (d') {}  (6,1) node[w] (z') {} ;
\draw (v)--(a)--(y)--(z)--(z')--(6,1) (v)--(w)--(c)--(d)--(d') (a)--(w);
\draw[dashed]  (v)--(d);
\draw (5,-2.5) node[empty]{(b)};
\end{tikzpicture}}
\caption{}\label{fig:another v chord}
\end{figure}

\begin{corollary}\label{prehole:cor20}
 Let $C$ be a prehole with $|C|\geq 8$. Then, for $v\in C$, $3 \leq \deg_{G[C]}(v)\leq 5$ ($v\notin\{ v_1,v_2\}$), $2 \leq \deg_{G[C]}(v)\leq 4$ ($v\in\{ v_1,v_2\}$).
\end{corollary}
\begin{proof}
  Follows directly from Lemmas~\ref{lem:quasi35} and~\ref{prehole:lem40}.
\end{proof}
Note that there is a prehole with six vertices and three vertices of degree 2, and an interior triangle.
\begin{figure}[H]
\tikzset{every node/.style={circle,draw,fill=none,inner sep=0pt,minimum size=1.25mm}}
\centerline{\begin{tikzpicture}[xscale=0.8,yscale=0.8,font=\small]
\draw (1,1) node[b] (0) {} (0,0.5) node[w] (1) {}  (2,0.5) node[w] (2) {} (0,-0.5) node[b] (3) {} (2,-0.5) node[b] (4) {} (1,-1) node[w] (5) {} ;
\draw (4)--(2)--(0)--(1)--(3)--(5)--(4) (3)--(0)--(4)--(3) ;
\end{tikzpicture}}
\caption{Prehole with 3 vertices of degree 2}\label{fig:degree2}
\end{figure}
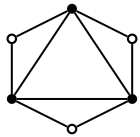
Let $T_1=\{v_1,u_1,w_1\}$, $T_2=\{v_2,u_2,w_2\}$ be any two edge-disjoint triangles in a flawless graph $G$. Let $M$ be the component of $G\sm\{v_1,v_2\}$ containing $u_1w_1$, $u_2w_2$, if such a component exists. If $M$ does not exist then $v_1,v_2$ clearly do not determine a prehole. Otherwise,
 \begin{lemma}\label{prehole:lem50}
$C=(v_1,u_1,\ldots,u_2,v_2,w_2,\ldots,w_1,v_1)$ determines a minimal prehole if and only if $M$ is a monotone graph containing two vertex-disjoint paths betweens $u_1,u_2$ and  $v_1,v_2$.
\end{lemma}
\begin{proof}
Since $C'\subseteq M$, the condition is certainly sufficient.  Now, if there are not vertex-disjoint paths $P_1=(u_1,\ldots,u_2)$, $P_1=(v_1,\ldots,v_2)$, then $C$ cannot be a prehole. So suppose that $M$ is not monotone, and hence must contain a triangle $T=\{u,v,w\}$. There are four cases depending on the number of vertices in $T\cap C'$.
\begin{enumerate}[topsep=2pt,itemsep=0pt,label=(\roman*)]
 \item  $|T\cap C'|=3$. Then $C'$ is not monotone, so this cannot occur, by Lemma~\ref{prehole:lem35}.
\item  $|T\cap C'|=2$. There are two subcases:
\begin{enumerate}[topsep=0pt,itemsep=0pt,label=(\alph*)]
 \item $|T\cap P_1|=2$. See Fig.~\ref{prehole:2vertices}(a). In this case, with $v$ to the left of $u$, let $x$ be the neighbour of $u$ in $P_2$ furthest to the left of $w_2$.      Then, on moving $u$ from $R$ to $L$, $(v_1,u_1,\ldots,w,v,w,x,\ldots,w_1,v_1)$ is a shorter prehole.
 \item $|T\cap P_1|=|T\cap P_2|=1$. See Fig.~\ref{prehole:2vertices}(b). In this case, one edge of the triangle is a chord of $C'$. Then  $(v,u,\ldots,u_2,v_2,w_2,\ldots,w,v)$ is a shorter prehole.
   \begin{figure}[H]
\tikzset{every node/.style={circle,draw,fill=none,inner sep=0pt,minimum size=1.25mm}}
\centerline{\begin{tikzpicture}[xscale=0.9,yscale=0.67,font=\small]
\draw
 (0,0) node[w,label=left:$v_1$] (0) {}
 (1,1) node[b,label=above:$u_1$] (1) {} (1,-1) node[b,label=below:$w_1$] (1') {}
 (2,1) node[w] (2) {} (2,-1) node[w] (2') {}
 (3,1) node[b,label=above left:$w$] (3) {} (3,-1) node[b] (3') {}
 (4,1) node[w,label=above right:$u$] (4) {} (4,-1) node[w,label=below:$x$] (4') {}
 (5,1) node[b] (5) {} (5,-1) node[b] (5') {}
 (6,1) node[w,label=above right:$u_2$] (6) {} (6,-1) node[w,label=below:$w_2$] (6') {}
 (7,0) node[b,label=right:$v_2$] (7) {}
 (3.5,2.5) node[w,label=above:$v$] (v) {}(3.5,1.5) node[empty]{$T$} ;
   \draw (4')--(3')--(2')--(1')--(0)--(1)--(2)--(3)--(4)--(5)--(6)--(7)--(6')--(5')--(4')
   (3)--(v)--(4) (3,-2) node[empty]{(a)};
 \draw (1)--(1') (6)--(6') (4')--(4);
\end{tikzpicture}
\hspace*{1cm}
\begin{tikzpicture}[xscale=0.9,yscale=0.67,font=\small]
\draw
 (0,0) node[w,label=left:$v_1$] (0) {}
 (1,1) node[b,label=above:$u_1$] (1) {} (1,-1) node[b,label=below:$w_1$] (1') {}
 (2,1) node[w] (2) {} (2,-1) node[w] (2') {}
 (3,1) node[b] (3) {} (3,-1) node[b] (3') {}
 (4,1) node[w,label=above:$u$] (4) {} (4,-1) node[w,label=below:$w$] (4') {}
 (5,1) node[b] (5) {} (5,-1) node[b] (5') {}
 (6,1) node[w,label=above right:$u_2$] (6) {} (6,-1) node[w,label=below:$w_2$] (6') {}
 (7,0) node[b,label=right:$v_2$] (7) {}
 (3.25,0) node[b,label=left:$v$] (v) {}(3.75,0) node[empty]{$T$} ;
   \draw (4')--(3')--(2')--(1')--(0)--(1)--(2)--(3)--(4)--(5)--(6)--(7)--(6')--(5')--(4')
   (4)--(v)--(4') (4')--(4) (3,-2) node[empty]{(b)};
 \draw (1)--(1') (6)--(6') ;
\end{tikzpicture}}
\caption{$|T\cap C'|=3$}\label{prehole:2vertices}
\end{figure}
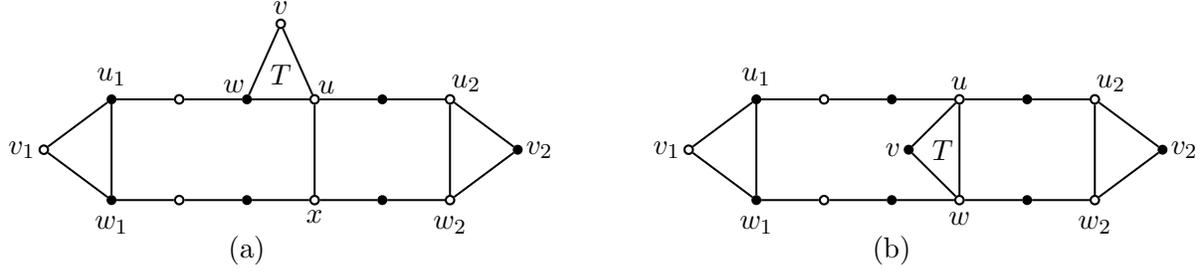

\end{enumerate}
 \item  $|T\cap C'|=1$. Then, using the arguments of Lemma~\ref{lem:triangle10}, the situation is as shown in Fig.~\ref{prehole:1vertex}. Here $x$ is the neighbour of $u$ in $P_2$ furthest to the left of $w_2$. Then, on moving $v$ from $R$ to $L$, there is a shorter prehole $(v_1,u_1,\ldots,y,w,v,u,x,\ldots,w_1,v_1)$.
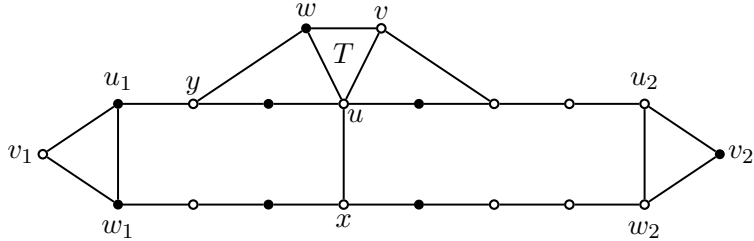
\begin{figure}[H]
\tikzset{every node/.style={circle,draw,fill=none,inner sep=0pt,minimum size=1.25mm}}
\centerline{\begin{tikzpicture}[xscale=1.0,yscale=0.67,font=\small]
\draw
 (0,0) node[w,label=left:$v_1$] (0) {}
 (1,1) node[b,label=above:$u_1$] (1) {} (1,-1) node[b,label=below:$w_1$] (1') {}
 (2,1) node[w,label=above:$y$] (2) {} (2,-1) node[w] (2') {}
 (3,1) node[b] (3) {} (3,-1) node[b] (3') {}
 (4,1) node[w,label=below right:$u$] (4) {} (4,-1) node[w,label=below:$x$] (4') {}
 (5,1) node[b] (5) {} (5,-1) node[b] (5') {}
 (6,1) node[w] (6) {} (6,-1) node[w] (6') {}
 (7,1) node[w] (7) {} (7,-1) node[w] (7') {}
 (8,1) node[w,label=above:$u_2$] (8) {} (8,-1) node[w,label=below:$w_2$] (8') {}
 (9,0) node[b,label=right:$v_2$] (9) {}  (4,2) node[empty]{$T$} ;
   \draw (4')--(3')--(2')--(1')--(0)--(1)--(2)--(3)--(4)--(5)--(6)
   --(7)--(8)--(9)--(8')--(7')--(6')--(5')--(4') ;
   \draw (3.5,2.5) node[b,label=above:$w$] (w) {} (4.5,2.5) node[w,label=above:$v$] (v) {}
   (v)--(w)--(4)--(v) (w)--(2) (v)--(6);
 \draw (1)--(1') (8)--(8') (4')--(4);
\end{tikzpicture}}
\caption{$|T\cap C'|=1$}\label{prehole:1vertex}
\end{figure}
\item  $|T\cap C'|=0$. Then, again using the arguments of Lemma~\ref{lem:triangle10}, the situation is as in Fig.~\ref{prehole:1vertex}. Here $x$ is the neighbour of $z$ in $P_2$ furthest to the left of $w_2$. Then, on moving $z$ from $R$ to $L$, there is a shorter prehole $(v_1,u_1,\ldots,y,w,v,u,z,x,\ldots,w_1,v_1)$.
\begin{figure}[H]
\tikzset{every node/.style={circle,draw,fill=none,inner sep=0pt,minimum size=1.25mm}}
\centerline{\begin{tikzpicture}[xscale=1.0,yscale=0.67,font=\small]
\draw
 (0,0) node[w,label=left:$v_1$] (0) {}
 (1,1) node[b,label=above:$u_1$] (1) {} (1,-1) node[b,label=below:$w_1$] (1') {}
 (2,1) node[w] (2) {} (2,-1) node[w] (2') {}
 (3,1) node[b,label=below:$y$] (3) {} (3,-1) node[b] (3') {}
 (4,1) node[w,label=below right:$z$] (4) {} (4,-1) node[w,label=below:$x$] (4') {}
 (5,1) node[b] (5) {} (5,-1) node[b] (5') {}
 (6,1) node[w] (6) {} (6,-1) node[w] (6') {}
 (7,1) node[w] (7) {} (7,-1) node[w] (7') {}
 (8,1) node[w,label=above:$u_2$] (8) {} (8,-1) node[w,label=below:$w_2$] (8') {}
 (9,0) node[b,label=right:$v_2$] (9) {}  (4,3.5) node[empty]{$T$} ;
   \draw (4')--(3')--(2')--(1')--(0)--(1)--(2)--(3)--(4)--(5)--(6)
   --(7)--(8)--(9)--(8')--(7')--(6')--(5')--(4') ;
   \draw (3.5,4) node[w,label=above:$w$] (w) {} (4.5,4) node[b,label=above:$v$] (v) {}
   (4,2.5)  node[w,label=left:$u$] (u) {}    (v)--(w)--(u)--(v) (u)--(4) (w)--(3) ;
 \draw (1)--(1') (8)--(8') (4')--(4);
\end{tikzpicture}}
\caption{$|T\cap C'|=0$}\label{prehole:0vertex}
\end{figure}
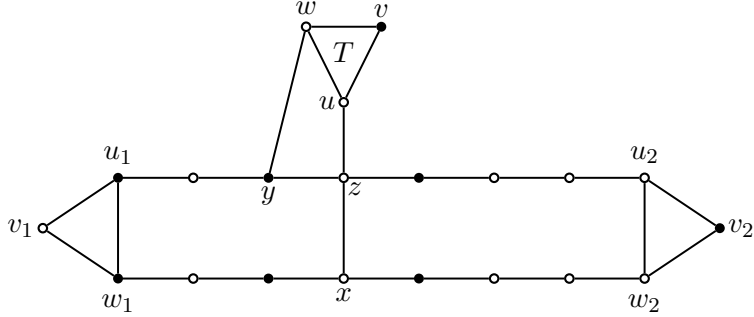
\end{enumerate}
In all cases, we have a contradiction to the minimality of $C$, and hence $M$ can have no triangle, and so is monotone. Note also, in all cases, that if there are edges other than those shown, they are either irrelevant, if they are $L$:$L$ or $R$:$R$, or can be used to shorten the prehole further, if they are \LR. See Fig.~\ref{fig:another v chord}, for example.
\end{proof}
Lemma~\ref{prehole:lem50} implies a polynomial time algorithm for detecting a minimal prehole, in a similar way to the algorithm of section~\ref{ss:findprehole}.

\subsection{Preholes containing 5-holes and triangles}

It remains to consider preholes in graphs which contain 5-holes, and may also contain triangles.
Preholes determined by two triangles will be dealt with as in section~\ref{ss:noholes}.
\begin{lemma} \label{l:le12}
Let $C$ be a minimal prehole in a flawless graph $G$ which contains no odd hole of size greater than five. If $C$ connects a 5-hole and a triangle, or if $C$ connects two 5-holes, then $|C|\leq 12$.
\end{lemma}
\begin{proof}
The situation is as shown in Fig.~\ref{fig:53prehole}. The prehole $C$ connects a hole $H$ and a triangle $T$, though the argument applies equally if $T$ is replaced by a 5-hole $H'$, as indicated in Fig.~\ref{fig:55prehole}.

\begin{figure}[H]
\tikzset{every node/.style={circle,draw,fill=none,inner sep=0pt,minimum size=1.25mm}}
\centerline{\begin{tikzpicture}[xscale=0.8,yscale=0.6,font=\small]
\draw (0,0) node[b] (0) {}
 (1,1.3) node[w] (1) {} (1,-1.3) node[w] (1') {}
 (2,1) node[b] (2) {} (2,-1) node[b,label=below:$u$] (2') {}
 (3,1) node[w] (3) {} (3,-1) node[w] (3') {}
 (4,1) node[b] (4) {} (4,-1) node[b] (4') {}
 (5,1) node[w] (5) {} (5,-1) node[w] (5') {}
 (6,1) node[b,label=above:$v$] (6) {} (6,-1) node[b,label=below:$w$] (6') {}
 (7.15,0) node[w] (7) {};
 \draw (3')--(2')--(1')--(0)--(1)--(2)--(3)--(4)--(5)--(6)--(7)--(6')--(5')--(4')--(3') ;
 \draw (2)--(2') (6)--(6') (2')--(6) (4')--(6) ;
 \draw (1.2,0) node[empty] {$H$}  (6.4,0) node[empty] {$T$} (7.65,0)
 (4,1.5) node[empty] {$P_1$}  (4,-1.5) node[empty] {$P_2$}  ;
\end{tikzpicture}}
\caption{A prehole determine by a 5-hole and a triangle}\label{fig:53prehole}
\end{figure}
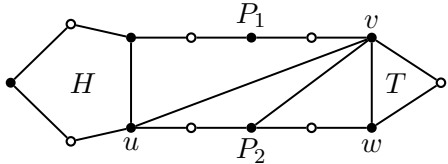
 Now $H$ and $T$ are joined by two paths $P_1,P_2$, as shown. We will not assume that these paths are of equal length. Let $v$ be the unique vertex in $T\cap P_1$, and $w$ the unique vertex in $T\cap P_2$. From Lemma~\ref{lem:quasi35}, $v$ is incident to at most three chords in $C$, which must be at distance 2 on $C$. However, by Lemma~\ref{lem:quasi10}, $v$ must be adjacent to $u\in H$. Since $\dist(u,w)\leq 4$, $|P_2|\leq 4$. The same argument applied to $w$ gives $|P_1|\leq 4$. Thus $|C|\leq 12$.

 \begin{figure}[H]
\tikzset{every node/.style={circle,draw,fill=none,inner sep=0pt,minimum size=1.25mm}}
\centerline{\begin{tikzpicture}[xscale=0.7,yscale=0.6,font=\small]
\draw (0,0) node[b] (0) {}
 (1,1.3) node[w] (1) {} (1,-1.3) node[w] (1') {}
 (2,1) node[b] (2) {} (2,-1) node[b,label=below:$u$] (2') {}
 (3,1) node[w] (3) {} (3,-1) node[w] (3') {}
 (4,1) node[b] (4) {} (4,-1) node[b] (4') {}
 (5,1) node[w] (5) {} (5,-1) node[w,label=below:$b$] (5') {}
 (6,1) node[b,label=above:$v$] (6) {} (6,-1) node[b,label=below:$w$] (6') {}
 (7,1.3) node[w,label=above:$a$] (7u) {} (7,-1.3) node[w] (7d) {}
 (8,0) node[b] (8) {}  ;
 \draw (3')--(2')--(1')--(0)--(1)--(2)--(3)--(4)--(5)--(6) (6')--(5')--(4')--(3') ;
 \draw (2)--(2') (6)--(6') (2')--(6) (4')--(6) ;
 \draw (6)--(7u)--(8)--(7d)--(6') ;
 \draw[dotted] (5')--(7u) ;
 \draw (1.2,0) node[empty] {$H$}   (7,0) node[empty] {$H'$}
 (4,1.5) node[empty] {$P_1$}  (4,-1.5) node[empty] {$P_2$} (4,-2.5) node[empty] {(a)} ; \end{tikzpicture}
\hspace*{2cm}
\begin{tikzpicture}[xscale=0.75,yscale=0.6,font=\small]
\draw (-1,0) node[b] (0) {}
 (0,1.3) node[w] (1) {} (1,-1.3) node[w] (1') {}
 (1,1) node[b] (2) {} (2,-1) node[b,label=below:$u$] (2') {}
 (2,1) node[w] (3) {} (3,-1) node[w] (3') {}
 (3,1) node[b] (4) {} (4,-1) node[b] (4') {}
 (4,1) node[w] (5) {} (5,-1) node[w,label=below:$b$] (5') {}
 (5,1) node[b,label=above:$v$] (6) {} (6,-1.3) node[b,label=below:$w$] (6') {}
 (6,1) node[w,label=above:$a$] (7u) {} (8,0) node[w] (7d) {}
 (7,1.3) node[b] (8) {}  ;
 \draw (3')--(2')--(1')--(0)--(1)--(2)--(3)--(4)--(5)--(6) (6')--(5')--(4')--(3') ;
 \draw (2)--(2') (5')--(7u) ;
 \draw (6)--(7u)--(8)--(7d)--(6') ;
 \draw (0.7,0) node[empty] {$H$}   (6.5,0) node[empty] {$H''$}
 (3.5,1.5) node[empty] {$P'_1$}  (3.5,-1.5) node[empty] {$P'_2$} (4,-2.5) node[empty] {(b)} ;
\end{tikzpicture}}
\caption{A prehole determine by two 5-holes}\label{fig:55prehole}
\end{figure}
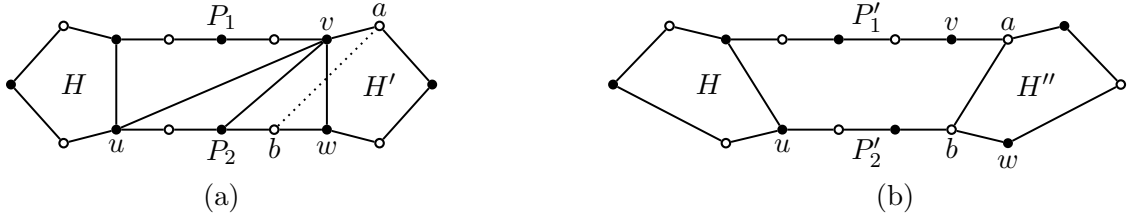
If $T$ is replaced by $H'$, let us assume that $P_2$ is the shorter path, and is as short as possible. Then using the argument above, $|P_1|,|P_2|\leq 4$. If $|P_2|\geq 3$, $v$ must have three edges to $P_2\cup H$. (In Fig.~\ref{fig:55prehole}(a), $|P_2|=4$.) Then  $\{u,\ldots,w,v,a\}$ form a stirrer, unless $ab$ is an edge. If so, $C$ has a representation with $|P'_1|>4$, as shown in Fig.~\ref{fig:55prehole}(b), so $C$ cannot be a minimal prehole.  It follows that  $|P_1|,|P_2|\leq 2$, and hence  $|C|\leq 12$.
\end{proof}
It follows that there is an $O(n^{12})$ time algorithm for detecting all minimal preholes in a graph with no holes of length greater than 5, using simple enumeration.

\section{Recognition algorithm}\label{sec:algorithm}

The pseudocode on page \pageref{algo:recognition} summarises our algorithm.

\begin{algorithm}
  \SetNlSty{small}{}{}
  \SetArgSty{textnormal}
  \SetKwInOut{Input}{input}
  \SetKwInOut{Output}{output}
  \SetKw{KwAnd}{and}
  \SetKw{KwOr}{or}
  \SetKw{KwAccept}{accept}
  \SetKw{KwReject}{reject}

  \Input{a connected graph $G=(V,E)$}
  \Output{accept if $G$ is quasi-monotone, reject otherwise}

  \BlankLine

  \Begin{
    \lIf{\label{l2}
      $G$ contains a flaw or a prehole of length $12$ or less
    }{
      \KwReject
    }
    \If{\label{l03}$G$ contains a hole}{
      \label{l3}find a hole $C$ in $G$\;
      \lIf(\tcc*[f]{Lemma \ref{prehole:lem30}}){\label{l4}
        $|C|$ is even \KwOr $G$ contains a triangle}{\KwReject
      }
      \label{l5}reduce it to a short hole $C'$;
      \tcc*[f]{Lemma \ref{lem:shorthole}} \\
      \lIf{\label{l6}$|C'|$ is even}{\KwReject}
    }
    \eIf{\label{l7}$C'$ is defined \KwAnd $|C'| \ge 7$}{
      \label{l14}choose vertices $v,w \in V$ with $\dist(v,w) \ge 3$\; 
      \eIf{$\cG_v$ and $\cG_w$ are monotone}{
        \label{l8}partition $V$ into independent sets
          $D_1, D_2, \dots, D_{\ell}$; \tcc*[f]{section \ref{ss:splitting}} \\
        \For{\label{l9}$i \gets 1$ \KwTo $\ell$}{
          \For{\label{l10}every $4$-cycle $(a,c,b,d)$ in $G_i$}{
            \lIf(\label{l11}\tcc*[f]{section \ref{ss:findprehole}}){
              disj.~paths $P_{ad}$ and $P_{bc}$ in $\ol{G}_i$ exist}{\KwReject}
          }
        }
        \KwAccept
      }{\KwReject}
    }{\label{l17} $\cT \gets {}$the set of all triangles in $G$;
      \tcc*[f]{Lemma \ref{l:le12}} \\
      \For{\label{l18}every pair of vertex-disjoint triangles $T_1 \in \cT$ and $T_2 \in \cT$}%
      {\label{l19}$U_1\gets{}$ the vertex set of the component of $G \sm T_2$ containing $T_1$\;
       \label{l20}$U_2\gets{}$ the vertex set of the component of $G \sm T_1$ containing $T_2$\;
       \label{l21}$U_3 \gets U_1 \cap U_2$; $U_4 \gets T_1 \cup T_2 \cup U_3$\;
       \lIf{\label{l22}$G[U_3]$ is bipartite \KwAnd $|U_4|>12$ \KwAnd
            $G[U_4]$ contains two disjoint $T_1$--$T_2$-paths}{\KwReject}
      }
      \KwAccept
    }
  }

  \caption{Algorithm recognising quasi-monotone graphs}
  \label{algo:recognition}
\end{algorithm}

For the run-time analysis we assume a connected graph $G=(V,E)$ as input with
$|V|=n$ and $|E|=m$. Line \ref{l2} can be executed in time $O(n^{12})$. In
line \ref{l3} the algorithm can find a hole in $G$ if there is one in time
$O(m^2)$, see \cite{NikPal07}. In line \ref{l5} we shorten a long hole,
$O(nm)$. The tests in lines \ref{l4}, \ref{l6} and \ref{l7} require time
$O(n)$. In the same time we choose $v$ and $w$ in line \ref{l14}. The graphs
$\cG_v$ and $\cG_w$ can be constructed and recognised as monotone graphs in
linear time \cite{SpBrSt87}. This also gives the partition into chain graphs
in line \ref{l8}. In lines \ref{l9} and \ref{l10} we consider $O(n^4)$
quadrangles. The body of the for-loops in line \ref{l11} can be implemented
by max flow in time $O(n^3)$. That is, the then-branch of the conditional
statement starting at line \ref{l7}, which deals with holes of length seven or
more, requires time $O(n^7)$ in total.

In the else-branch we only consider triangles, since every prehole containing
a five-hole has length at most twelve by Lemma \ref{l:le12}, and was therefore
detected in line \ref{l2}.  The set $\cT$ in line \ref{l17} can be constructed
in time $O(n^3)$, and $n^3$ also a bound on its size. Therefore the for-loop
starting in line \ref{l18} is executed at most $n^6$ times. In lines
\ref{l10}--\ref{l22} we construct $U_1$, $U_2$, $G[U_3]$ and $G[U_4]$ in
linear time. The disjoint paths in line \ref{l22} can be found in time
$O(n^3)$, which gives a total time of $O(n^9)$ for the else-branch.
Consequently line \ref{l2} determines the overall running time of $O(n^{12})$.

\section{Recognising a prehole} \label{sec:NPc}

The recognition algorithm finds a prehole in a flawless graph, if there is one.
However, it relies heavily on the absence of flaws. Therefore it is quite natural to ask
whether we can find a prehole in any graph in polynomial time. This is exactly the
recognition problem for the class \och.

Here we consider a related problem: Given a graph, is it a prehole?
This is the question of whether the graph is a cycle with only even chords.
We will show that this is an \NP-complete problem. Of course, this does not mean that
the recognition of odd-chordal graphs is \NP-complete, since that is the problem of determining
whether the graph \emph{contains} a prehole. To illustrate the difference, consider
the question of whether a graph is a cycle, having either odd or even chords. This is
\NP-complete, since it is the Hamilton  cycle problem. By contrast, the question of
whether a graph \emph{contains} a cycle is trivial. It simply involves determining
whether the graph is a forest.

A graph $G=(V,E)$ is a prehole if there is a bipartition $L,R$ of
$V$ such that $G[\LR]$ is a hole. That is, $G[\LR]$ is a
Hamilton cycle of $G$ and every other edge in $E$ has both endpoints
in $L$ or both in $R$. If $G$ is a prehole then such an $L,R$
is called \emph{certifying bipartition}. For example, the graph $G$ depicted in
Fig.~\ref{fig:rec} is a prehole. Three certifying bipartitions are shown.
$G$ has no more certifying bipartitions. To see
this, start from a bipartition of one of the triangles, and observe that it
uniquely extends to a bipartition of the whole graph.

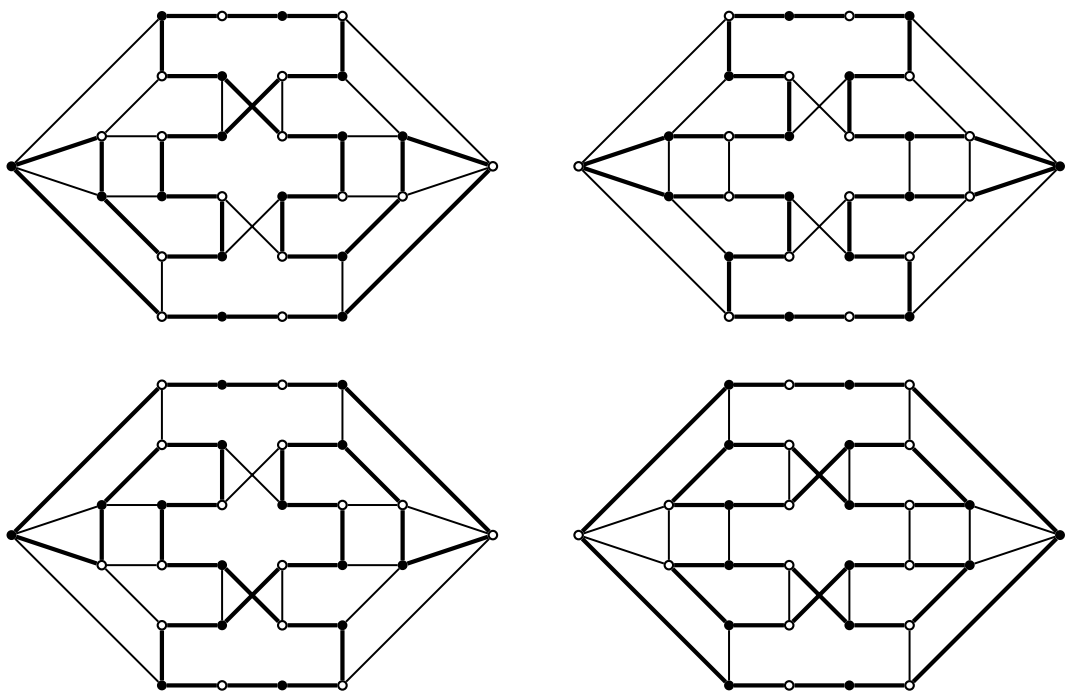
\begin{figure}[htbp]
  \centering
  \setlength{\tabcolsep}{5mm}
  \renewcommand{\arraystretch}{5.0}
  \begin{tabular}{cc}
    \begin{tikzpicture}[scale=0.4]
      \node[b] (16) at (-3, 5) {};  \node[w] (26) at (-1, 5) {};
      \node[b] (36) at ( 1, 5) {};  \node[w] (46) at ( 3, 5) {};
      \node[w] (15) at (-3, 3) {};  \node[b] (25) at (-1, 3) {};
      \node[w] (35) at ( 1, 3) {};  \node[b] (45) at ( 3, 3) {};
      \node[w] (14) at (-3, 1) {};  \node[b] (24) at (-1, 1) {};
      \node[w] (34) at ( 1, 1) {};  \node[b] (44) at ( 3, 1) {};
      \node[b] (13) at (-3,-1) {};  \node[w] (23) at (-1,-1) {};
      \node[b] (33) at ( 1,-1) {};  \node[w] (43) at ( 3,-1) {};
      \node[w] (12) at (-3,-3) {};  \node[b] (22) at (-1,-3) {};
      \node[w] (32) at ( 1,-3) {};  \node[b] (42) at ( 3,-3) {};
      \node[w] (11) at (-3,-5) {};  \node[b] (21) at (-1,-5) {};
      \node[w] (31) at ( 1,-5) {};  \node[b] (41) at ( 3,-5) {};
      \node[w] (b1) at (-5, 1) {};  \node[b] (b2) at (-8, 0) {};
      \node[b] (b3) at (-5,-1) {};  \node[b] (d1) at ( 5, 1) {};
      \node[w] (d2) at ( 8, 0) {};  \node[w] (d3) at ( 5,-1) {};
      \draw[ultra thick] (b1)--(b3)--(12)--(22)--(23)--(13)--(14)--(24)--(35)--(45)--(46)--(36)--(26)--(16)--(15)--(25)--(34)--(44)--(43)--(33)--(32)--(42)--(d3)--(d1)--(d2)--(41)--(31)--(21)--(11)--(b2)--(b1);
      \draw (13)--(b3)--(b2)--(16)  (14)--(b1)--(15);
      \draw (11)--(12)  (22)--(33)  (23)--(32)  (24)--(25)  (34)--(35);
      \draw (46)--(d2)--(d3)--(43)  (45)--(d1)--(44)  (42)--(41);
    \end{tikzpicture}
    &
    \begin{tikzpicture}[scale=0.4]
      \node[w] (16) at (-3, 5) {};  \node[b] (26) at (-1, 5) {};
      \node[w] (36) at ( 1, 5) {};  \node[b] (46) at ( 3, 5) {};
      \node[b] (15) at (-3, 3) {};  \node[w] (25) at (-1, 3) {};
      \node[b] (35) at ( 1, 3) {};  \node[w] (45) at ( 3, 3) {};
      \node[w] (14) at (-3, 1) {};  \node[b] (24) at (-1, 1) {};
      \node[w] (34) at ( 1, 1) {};  \node[b] (44) at ( 3, 1) {};
      \node[w] (13) at (-3,-1) {};  \node[b] (23) at (-1,-1) {};
      \node[w] (33) at ( 1,-1) {};  \node[b] (43) at ( 3,-1) {};
      \node[b] (12) at (-3,-3) {};  \node[w] (22) at (-1,-3) {};
      \node[b] (32) at ( 1,-3) {};  \node[w] (42) at ( 3,-3) {};
      \node[w] (11) at (-3,-5) {};  \node[b] (21) at (-1,-5) {};
      \node[w] (31) at ( 1,-5) {};  \node[b] (41) at ( 3,-5) {};
      \node[b] (b1) at (-5, 1) {};  \node[w] (b2) at (-8, 0) {};
      \node[b] (b3) at (-5,-1) {};  \node[w] (d1) at ( 5, 1) {};
      \node[b] (d2) at ( 8, 0) {};  \node[w] (d3) at ( 5,-1) {};
      \draw[ultra thick] (b2)--(b1)--(14)--(24)--(25)--(15)--(16)--(26)--(36)--(46)--(45)--(35)--(34)--(44)--(d1)--(d2)--(d3)--(43)--(33)--(32)--(42)--(41)--(31)--(21)--(11)--(12)--(22)--(23)--(13)--(b3)--(b2);
      \draw (15)--(b1)--(b3)--(12)  (16)--(b2)--(11)  (14)--(13);
      \draw (25)--(34)  (24)--(35)  (23)--(32)  (22)--(33);
      \draw (45)--(d1)--(d3)--(42)  (46)--(d2)--(41)  (44)--(43);
    \end{tikzpicture}
    \\
    \begin{tikzpicture}[scale=0.4]
      \node[w] (16) at (-3, 5) {};  \node[b] (26) at (-1, 5) {};
      \node[w] (36) at ( 1, 5) {};  \node[b] (46) at ( 3, 5) {};
      \node[w] (15) at (-3, 3) {};  \node[b] (25) at (-1, 3) {};
      \node[w] (35) at ( 1, 3) {};  \node[b] (45) at ( 3, 3) {};
      \node[b] (14) at (-3, 1) {};  \node[w] (24) at (-1, 1) {};
      \node[b] (34) at ( 1, 1) {};  \node[w] (44) at ( 3, 1) {};
      \node[w] (13) at (-3,-1) {};  \node[b] (23) at (-1,-1) {};
      \node[w] (33) at ( 1,-1) {};  \node[b] (43) at ( 3,-1) {};
      \node[w] (12) at (-3,-3) {};  \node[b] (22) at (-1,-3) {};
      \node[w] (32) at ( 1,-3) {};  \node[b] (42) at ( 3,-3) {};
      \node[b] (11) at (-3,-5) {};  \node[w] (21) at (-1,-5) {};
      \node[b] (31) at ( 1,-5) {};  \node[w] (41) at ( 3,-5) {};
      \node[b] (b1) at (-5, 1) {};  \node[b] (b2) at (-8, 0) {};
      \node[w] (b3) at (-5,-1) {};  \node[w] (d1) at ( 5, 1) {};
      \node[w] (d2) at ( 8, 0) {};  \node[b] (d3) at ( 5,-1) {};
      \draw[ultra thick] (b3)--(b1)--(15)--(25)--(24)--(14)--(13)--(23)--(32)--(42)--(41)--(31)--(21)--(11)--(12)--(22)--(33)--(43)--(44)--(34)--(35)--(45)--(d1)--(d3)--(d2)--(46)--(36)--(26)--(16)--(b2)--(b3);
      \draw (14)--(b1)--(b2)--(11)  (13)--(b3)--(12);
      \draw (16)--(15)  (25)--(34)  (24)--(35)  (22)--(23)  (32)--(33);
      \draw (41)--(d2)--(d1)--(44)  (42)--(d3)--(43)  (46)--(45);
    \end{tikzpicture}
    &
    \begin{tikzpicture}[scale=0.4]
      \node[b] (16) at (-3, 5) {};  \node[w] (26) at (-1, 5) {};
      \node[b] (36) at ( 1, 5) {};  \node[w] (46) at ( 3, 5) {};
      \node[b] (15) at (-3, 3) {};  \node[w] (25) at (-1, 3) {};
      \node[b] (35) at ( 1, 3) {};  \node[w] (45) at ( 3, 3) {};
      \node[b] (14) at (-3, 1) {};  \node[w] (24) at (-1, 1) {};
      \node[b] (34) at ( 1, 1) {};  \node[w] (44) at ( 3, 1) {};
      \node[b] (13) at (-3,-1) {};  \node[w] (23) at (-1,-1) {};
      \node[b] (33) at ( 1,-1) {};  \node[w] (43) at ( 3,-1) {};
      \node[b] (12) at (-3,-3) {};  \node[w] (22) at (-1,-3) {};
      \node[b] (32) at ( 1,-3) {};  \node[w] (42) at ( 3,-3) {};
      \node[b] (11) at (-3,-5) {};  \node[w] (21) at (-1,-5) {};
      \node[b] (31) at ( 1,-5) {};  \node[w] (41) at ( 3,-5) {};
      \node[w] (b1) at (-5, 1) {};  \node[w] (b2) at (-8, 0) {};
      \node[w] (b3) at (-5,-1) {};  \node[b] (d1) at ( 5, 1) {};
      \node[b] (d2) at ( 8, 0) {};  \node[b] (d3) at ( 5,-1) {};
      \draw[ultra thick] (b1)--(15)--(25)--(34)--(44)--(d1)--(45)--(35)--(24)--(14)--(b1);
      \draw[ultra thick] (b2)--(16)--(26)--(36)--(46)--(d2)--(41)--(31)--(21)--(11)--(b2);
      \draw[ultra thick] (b3)--(13)--(23)--(32)--(42)--(d3)--(43)--(33)--(22)--(12)--(b3);
      \draw (b1)--(b2)--(b3)--(b1)  (22)--(23)  (24)--(25);
      \draw (11)--(12)  (13)--(14)  (15)--(16);
      \draw (41)--(42)  (43)--(44)  (45)--(46);
      \draw (d1)--(d2)--(d3)--(d1)  (32)--(33)  (34)--(35);
    \end{tikzpicture}
  \end{tabular}
  \vspace{-\baselineskip}
  \caption{Three certifying bipartitions of a graph $G$.
    For the fourth bipartition $L,R$ the graph $G[\LR]$ is $2$-regular
    but disconnected.}
  \label{fig:rec}
\end{figure}

The decision problem PH asks, given a graph $G$, whether $G$ is a prehole.
Clearly PH is a problem in \NP. We show it is \NP-complete by a reduction from
NAE3SAT.

An instance of NAE3SAT is a boolean formula $\phi = \bigwedge_{j=1}^m c_j$ in
CNF. Each clause $c_j = \ell_{j,1} \vee \ell_{j,2} \vee \ell_{j,3}$ consists
of exactly three literals. If $X = \{x_1,x_2,\dots,x_n\}$ is the set of
variables occurring in $\phi$ then every literal $\ell_{j,k}$ is either a
variable $x_i$ or its negation $\neg x_i$. For a truth assignment
$a : X \to \{0,1\}$ let $\ol{a} : X \to \{0,1\}$ be defined by $\ol{a}(x) =
a(\neg x)$ for all $x \in X$. The instance $\phi$ of NAE3SAT is accepted if
there is a truth assignment $a$ such that $a(\phi)=1$ and $\ol{a}(\phi)=1$.
That is, for each clause $c_j$ not all literals receive equal truth value.

Given an instance $\phi$ of NAE3SAT we construct a graph $G=(V,E)$ as follows:
\begin{enumerate}
\setlength{\itemsep}{0pt}
\item For each variable we create a \emph{truth assignment component} (tac) as
  shown in Fig.~\ref{fig:tac}, which also defines the vertices $x_i^+$ and
  $x_i^-$.
\item For each clause we create a \emph{satisfaction test component} (stc) as
  shown in Fig.~\ref{fig:stc}, which also defines the vertices $b_{j,k}$ and
  $d_{j,k}$. The stc is obtained from the graph $G$ shown in Figure
  \ref{fig:rec} by cutting an edge that connects two vertices of degree two
  into two half-edges.
\item We link these components in a circular way as shown in Figure
  \ref{fig:ring}.
\item We add the edges in the set
  \[F = \{\{x_i^+,b_{j,k}\}, \{x_i^-,d_{j,k}\} \mid \ell_{j,k} = x_i\}\} \cup
        \{\{x_i^+,d_{j,k}\}, \{x_i^-,b_{j,k}\} \mid \ell_{j,k} = \neg x_i\}\}\]
  where $1 \le i \le n$, $1 \le j \le m$ and $1 \le k \le 3$.
\end{enumerate}
This completes the construction of $G$.

\begin{figure}[htbp]
  \hspace*{\fill}
  \begin{tikzpicture}[scale=0.3]
    \node[w] (L) at (0,3) {}; \node[w] (l) at (2,3) {};
    \node[w, label=below:$x_i^-$] (1) at (5,0) {};
    \node[w](2) at (5,2) {}; \node[w](3) at (5,4) {};
    \node[w, label=above:$x_i^+$] (4) at (5,6) {};
    \node[w] (r) at (8,3) {}; \node[w] (R) at (10,3) {};
    \draw (-2,3)--(L)--(l)--(4)--(r);
    \draw (l)--(1)--(r)--(R)--(12,3);
    \draw (1)--(2)--(3)--(4);
  \end{tikzpicture}
  \hspace*{\fill}
  \raisebox{12mm}[0pt][0pt]{%
  \begin{tikzpicture}[scale=0.2]
    \node[tac, label=center:$x_i$] (x) at (5,3) {};
    \draw (0,3)--(x)--(10,3);
  \end{tikzpicture}
  }
  \hspace*{\fill}

  \caption{The truth assignment component for $x_i$,
    on the left in full and on the right symbolically.}
  \label{fig:tac}
\end{figure}
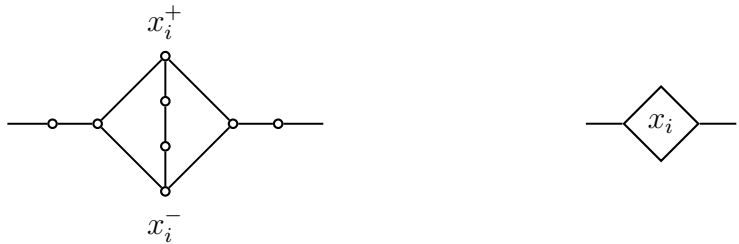

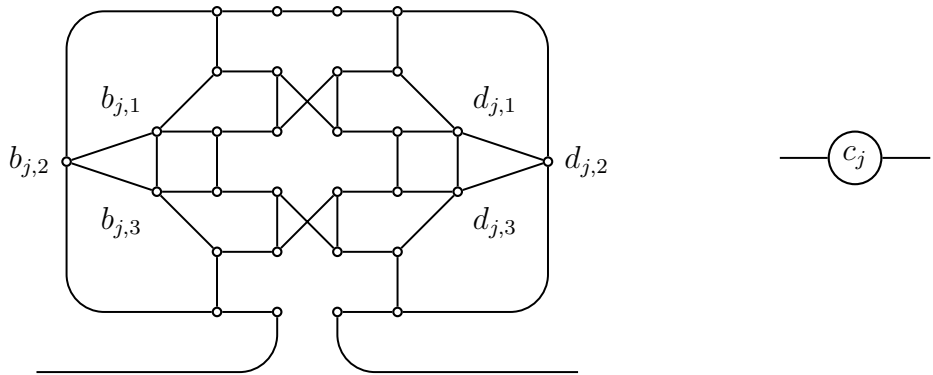
\begin{figure}[htbp]
  \hspace*{\fill}
  \begin{tikzpicture}[scale=0.4]
    \node[w] (16) at (-3, 5) {};  \node[w] (26) at (-1, 5) {};
    \node[w] (36) at ( 1, 5) {};  \node[w] (46) at ( 3, 5) {};
    \node[w] (15) at (-3, 3) {};  \node[w] (25) at (-1, 3) {};
    \node[w] (35) at ( 1, 3) {};  \node[w] (45) at ( 3, 3) {};
    \node[w] (14) at (-3, 1) {};  \node[w] (24) at (-1, 1) {};
    \node[w] (34) at ( 1, 1) {};  \node[w] (44) at ( 3, 1) {};
    \node[w] (13) at (-3,-1) {};  \node[w] (23) at (-1,-1) {};
    \node[w] (33) at ( 1,-1) {};  \node[w] (43) at ( 3,-1) {};
    \node[w] (12) at (-3,-3) {};  \node[w] (22) at (-1,-3) {};
    \node[w] (32) at ( 1,-3) {};  \node[w] (42) at ( 3,-3) {};
    \node[w] (11) at (-3,-5) {};  \node[w] (21) at (-1,-5) {};
    \node[w] (31) at ( 1,-5) {};  \node[w] (41) at ( 3,-5) {};
    \node[w, label=above left:$b_{j,1}$] (b1) at (-5, 1) {};
    \node[w, label=left:$b_{j,2}$] (b2) at (-8, 0) {};
    \node[w, label=below left:$b_{j,3}$] (b3) at (-5,-1) {};
    \node[w, label=above right:$d_{j,1}$] (d1) at (5, 1) {};
    \node[w, label=right:$d_{j,2}$] (d2) at ( 8, 0) {};
    \node[w, label=below right:$d_{j,3}$] (d3) at (5,-1) {};
    \draw (b1)--(15)--(25)--(34)--(44)--(d1)--(45)--(35)--(24)--(14)--(b1);
    \draw (16)--(26)--(36)--(46)  (41)--(31)  (21)--(11);
    \draw (b3)--(13)--(23)--(32)--(42)--(d3)--(43)--(33)--(22)--(12)--(b3);
    \draw (b1)--(b2)--(b3)--(b1)  (22)--(23)  (24)--(25);
    \draw (11)--(12)  (13)--(14)  (15)--(16);
    \draw (41)--(42)  (43)--(44)  (45)--(46);
    \draw (d1)--(d2)--(d3)--(d1)  (32)--(33)  (34)--(35);
    \draw[rounded corners=5mm] (-9,-7)--(-1,-7)--(21);
    \draw[rounded corners=5mm] ( 9,-7)--( 1,-7)--(31);
    \draw[rounded corners=5mm] (16)--(-8, 5)--(b2)--(-8,-5)--(11);
    \draw[rounded corners=5mm] (46)--( 8, 5)--(d2)--( 8,-5)--(41);
  \end{tikzpicture}
  \hspace*{\fill}
  \raisebox{25mm}[0pt][0pt]{%
  \begin{tikzpicture}[scale=0.2]
    \node[stc] (c) at (5,3) {$c_j$};
    \draw (0,3)--(c)--(10,3);
  \end{tikzpicture}
  }
  \hspace*{\fill}

  \caption{The satisfaction test component for $c_j$,
    on the left in full and on the right symbolically.}
  \label{fig:stc}
\end{figure}

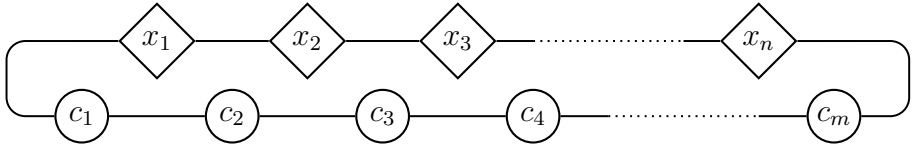
\begin{figure}[htbp]
  \centering
  \begin{tikzpicture}
    \node[tac, label=center:$x_1$] (x1) at ( 3,1) {};
    \node[tac, label=center:$x_2$] (x2) at ( 5,1) {};
    \node[tac, label=center:$x_3$] (x3) at ( 7,1) {};
    \node[tac, label=center:$x_n$] (xn) at (11,1) {};
    \node[stc] (c1) at ( 2,0) {$c_1$};
    \node[stc] (c2) at ( 4,0) {$c_2$};
    \node[stc] (c3) at ( 6,0) {$c_3$};
    \node[stc] (c4) at ( 8,0) {$c_4$};
    \node[stc] (cm) at (12,0) {$c_m$};
    \draw (x1)--(x2)--(x3)--(8,1)  (10,1)--(xn);
    \draw[dotted] (8,1)--(10,1);
    \draw (c1)--(c2)--(c3)--(c4)--(9,0)  (11,0)--(cm);
    \draw[dotted] (9,0)--(11,0);
    \draw[rounded corners=2.5mm] (x1)--(1,1)--(1,.5)--(1,0)--(c1);
    \draw[rounded corners=2.5mm] (xn)--(13,1)--(13,.5)--(13,0)--(cm);
  \end{tikzpicture}
  \caption{The components linked.}
  \label{fig:ring}
\end{figure}

Next we assume $a$ is a truth assignment such that $a(\phi)=1$ and
$\ol{a}(\phi)=1$. We construct a bipartition $(L,R)$ of $G$ that certifies
that $G$ is a prehole. Inside each tac for $x_i$ we put $x_i^+ \in R$ and
$x_i^- \in L$ if $a(x_i)=1$ and the other way around if $a(x_i)=0$, see Figure
\ref{fig:taca}. The vertices $b_{j,k}$ and $d_{j,k}$ are put in the $L$ or $R$
such that all edges in $F$ have both endpoints in the same partite set. Now
the triangles in the stc's are partitioned into two nonempty sets because
now all literals in one clause have the same truth value. The bipartition of
the triangles extends to the whole stc as shown in Fig.~\ref{fig:rec}.

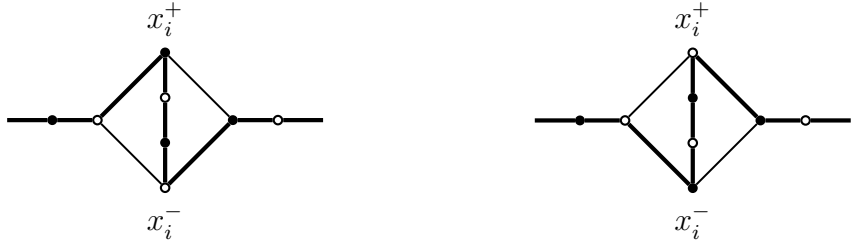
\begin{figure}[htbp]
  \hspace*{\fill}
  \begin{tikzpicture}[scale=0.3]
    \node[b] (L) at (0,3) {}; \node[w] (l) at (2,3) {};
    \node[w, label=below:$x_i^-$] (1) at (5,0) {};
    \node[b](2) at (5,2) {}; \node[w](3) at (5,4) {};
    \node[b, label=above:$x_i^+$] (4) at (5,6) {};
    \node[b] (r) at (8,3) {}; \node[w] (R) at (10,3) {};
    \draw[ultra thick] (-2,3)--(L)--(l)--(4)--(3)--(2)--(1)--(r)--(R)--(12,3);
    \draw (l)--(1)  (4)--(r);
  \end{tikzpicture}
  \hspace*{\fill}
  \begin{tikzpicture}[scale=0.3]
    \node[b] (L) at (0,3) {}; \node[w] (l) at (2,3) {};
    \node[b, label=below:$x_i^-$] (1) at (5,0) {};
    \node[w](2) at (5,2) {}; \node[b](3) at (5,4) {};
    \node[w, label=above:$x_i^+$] (4) at (5,6) {};
    \node[b] (r) at (8,3) {}; \node[w] (R) at (10,3) {};
    \draw[ultra thick] (-2,3)--(L)--(l)--(1)--(2)--(3)--(4)--(r)--(R)--(12,3);
    \draw (l)--(4)  (1)--(r);
  \end{tikzpicture}
  \hspace*{\fill}

  \caption{The bipartition for $a(x_i)=1$ on the left, and
    for $a(x_i)=0$ on the right.}
  \label{fig:taca}
\end{figure}

Finally we assume a certifying bipartition $(L,R)$ of the vertices of $G$.
We define a truth assignment $a$ by $a(x_i)=1$ if and only if $x_i^+ \in R$.
To see $a(\phi)=1$ and $\ol{a}(\phi)=1$ observe:
\begin{enumerate}
\item Every edge of $G$ that is incident to a vertex of degree 2 is contained
  in every Hamilton cycle of $G$. That is, its endpoints belong to
  different sides of the bipartition.
\item Every edge in $F$ is contained in no Hamilton cycle of $G$. That
  is, its endpoints belong to the same side of the bipartition.
\item Consequently, the bipartition of each component is one of the cases
  depicted in Figures \ref{fig:rec} (for stc) or \ref{fig:taca} (for tac).
\end{enumerate}
Now assume there is a clause $c_j$ such that $a$ assigns the same truth
value to all three literals. Then one of the triangles in the corresponding
stc has all vertices in $L$ and the other one all vertices in $R$, see
the bottom-right bipartition in Fig.~\ref{fig:rec}. This contradicts
the connectedness of $G[L{:}R]$. Hence $a(\phi)=1$ and $\ol{a}(\phi)=1$.

\section{Conclusion and discussion}\label{sec:conclusions}

In~\cite{DyeMu17a} we considered the problem of ergodicity and rapid mixing of the switch chain in hereditary graph classes. We gave a complete answer to the ergodicity question, and showed rapid mixing for the new class of quasimonotone graphs. This led us to introduce a new ``quasi-'' operator on bipartite graph classes, which is of independent interest. Quasimonotone graphs are a particular case of this construction. Another interesting class is the class of odd-chordal graphs, which are the quasi-chordal bipartite graphs. This is close to the largest class for which the switch chain is ergodic.

In this paper, we have investigated recognition of the quasimonotone graphs, and shown that this is in \PP. This is intended only to be a proof-of-concept. Our algorithms are far from optimal, and can certainly be improved. However, we do not believe that this class can be recognised in linear time, as for monotone graphs.

A more straightforward approach to recognising quasimonotone graphs would be provided by a polynomial time recognition algorithm for odd-chordal graphs. This is equivalent to the detection of preholes in a graph. We have considered this question, but we leave it as an open problem. The only evidence we can provide is that it is \NP-complete to determine if a graph is a prehole, which may be a harder question, Nonetheless, the \NP-completeness proof suggests that an efficient algorithm for recognising odd-chordal graphs may be elusive.



\end{document}